\documentclass[12pt, draftcls, onecolumn]{IEEEtran}
\usepackage{color}     
\usepackage{epsf,psfrag,amssymb,amsfonts,latexsym,cite,verbatim}
\usepackage{enumerate,times,array}
\usepackage{amsmath,cases,graphicx}
\usepackage[mathscr]{eucal}

\def\psfancypar#1#2{\begingroup\def\par{\endgraf\endgroup\lineskiplimit=0pt}
               \setbox2=\hbox{\large\sc #2}
               \newdimen\tmpht \tmpht \ht2 \advance\tmpht by \baselineskip
               \font\hhuge=Times-Bold at \tmpht
               \setbox1=\hbox{{\hhuge #1}}
               \count7=\tmpht \count8=\ht1
               \divide\count8 by 1000 \divide\count7 by \count8 
               \tmpht=.001\tmpht\multiply\tmpht by \count7 
               \font\hhuge=Times-Bold at \tmpht
               \setbox1=\hbox{{\hhuge #1}}
               \noindent
                \hangindent1.05\wd1
               \hangafter=-2 {\hskip-\hangindent
               \lower1\ht1\hbox{\raise1.0\ht2\copy1}%
                \kern-0\wd1}\copy2\lineskiplimit=-1000pt}

\newcommand{\etabf}{\mbox{${\bf\eta}$}}

\newcommand{\E}{\mbox{{\rm E}}}

\newcommand{\abf}{\mbox{${\bf a}$}}

\def\boxit#1{\vbox{\hrule\hbox{\vrule\kern3pt
        \vbox{\kern3pt#1\kern3pt}\kern3pt\vrule}\hrule}}

\def\reals{ { {\rm  I \kern-0.15em R }  } }
\def\complex{ {\,{{\rm C} \kern-0.50em \raise0.20ex {  |}}\, }}

\def\etabf{\hbox{\boldmath$\eta$\unboldmath}}

\def\mubf{\hbox{\boldmath$\mu$\unboldmath}}

\def\phibf{\hbox{\boldmath$\phi$\unboldmath}}

\def\Sigmabf{\hbox{$\bf \Sigma$}}

\def\Lambdabf{\mbox{$ \bf \Lambda $}}

\def\abf{{\bf a}}

\def\dbf{{\bf d}}
\def\ebf{{\bf e}}

\def\gbf{{\bf g}}
\def\hbf{{\bf h}}

\def\nbf{{\bf n}}

\def\ubf{{\bf u}}

\def\wbf{{\bf w}}
\def\xbf{{\bf x}}
\def\ybf{{\bf y}}

\def\xbf{{\bf x}}
\def\ybf{{\bf y}}
\def\Abf{{\bf A}}

\def\Dbf{{\bf D}}

\def\Fbf{{\bf F}}
\def\Gbf{{\bf G}}
\def\Hbf{{\bf H}}
\def\Ibf{{\bf I}}

\def\Pbf{{\bf P}}

\def\Rbf{{\bf R}}

\def\Ubf{{\bf U}}
\def\Vbf{{\bf V}}
\def\Wbf{{\bf W}}
\def\Xbf{{\bf X}}

\def\Cc{{\cal C}}

\def\Hc{{\cal H}}
\def\Ic{{\cal I}}

\def\Kc{{\cal K}}

\def\Nc{{\cal N}}

\def\Rc{{\cal R}}
\def\Sc{{\cal S}}

\def\Vc{{\cal V}}
\def\Wc{{\cal W}}

\def\be{\vskip .3cm \begin{equation}}
\def\ee{\end{equation} \vskip .4cm \noindent}
\def\defeq{{\stackrel{\Delta}{=}}}
\newcommand{\R}{\mbox{$\hat {\bf R}_{N}$}}

\def\Rxx{\Rbf_{\ssstyle X\kern-.1em X}}

\let\ssstyle=\scriptscriptstyle

\def\Kout{\setbox1=\hbox{\Huge\bf K}\hbox to
1.05\wd1{\hspace{.05\wd1}
}}
\def\Sout{\setbox1=\hbox{\Huge\bf S}\hbox to 1.05\wd1{\hspace{.05\wd1}
}}

  \ifx\LabelFigloaded\MYundefined\relax
  \else
   \fi

  \def\LabelFigloaded{\relax}

  \chardef\LabelFigCatAt\the\catcode`\@
  \catcode`\@=11

 \let\LabelFigwlog@ld\wlog
 \def\wlog#1{\relax}

 \ifx\\\MYundefined@
    \let\\\relax
 \fi

  \def\ms@g{\immediate\write16}

 \def\N@wif{\csname newif\endcsname }
 \def\Temp@ {\N@wif\ifIN@}
 \ifx\INN@\MYundefined@
    \else \let\Temp@\relax
 \fi
 \Temp@

  \def\IN@{\expandafter\INN@\expandafter}
  \long\def\INN@0#1@#2@{\long\def\NI@##1#1##2##3\ENDNI@
    {\ifx\m@rker##2\IN@false\else\IN@true\fi}%
     \expandafter\NI@#2@@#1\m@rker\ENDNI@}
  \def\m@rker{\m@@rker}
 
  \newtoks\Initialtoks@  \newtoks\Terminaltoks@
  \def\SPLIT@{\expandafter\SPLITT@\expandafter}
  \def\SPLITT@0#1@#2@{\def\TTILPS@##1#1##2@{%
     \Initialtoks@{##1}\Terminaltoks@{##2}}\expandafter\TTILPS@#2@}

 \def\Shifted@@#1#2#3{\setbox0=\hbox{#3}%
   \raise -\dp0\vbox {\kern-#2%
       \hbox {\kern#1\unhbox0\kern-#1}%
           \kern#2}}

 \newcount\gridcount
 \newbox\auxGridbox@ \newbox\hGridbox@ \newbox\vGridbox@
 \newbox\Labelbox@ \newbox\auxLabelbox@
 \newbox\Coordinatebox@
 \newtoks\Labeltoks@
 \newdimen\Wdd@ \newdimen\Htt@
 \newdimen\Wddd@ \newdimen\Httt@
 
 \def\Wr@{\immediate\write16}

 \newdimen\GL@wd
 \def\GridLineWidth#1{\GL@wd=#1}

 \def\gobble#1{}
 \def\EdgeErr@{\Wr@{}%
      \Wr@{\string\Edges\space argument
      1, 10, 100 or 1000 please\string!}%
      }

 \newcount\Edgect@

 \def\Sweepup#1\endSweepup{}

 \def\SetEdges@{%
    \edef\Zr@@s{\expandafter\gobble\number\Edgect@\empty}%
        \count255=0\Zr@@s\relax
        \ifnum\count255=\z@\else\EdgeErr@\show\tailtest\fi
        \count255=1\Zr@@s\relax
        \ifnum\count255=\Edgect@\relax\else\EdgeErr@\show\leadtest\fi
    \EdgGl@b\edef\Zr@s{\expandafter\gobble\Zr@@s\empty}
    \ifnum\Edgect@>\@ne\relax\EdgGl@b\let\L@Dc\empty
        \else\EdgGl@b\edef\L@Dc{\string.}\fi
    \ifnum\Edgect@>\@ne\relax
        \EdgGl@b\edef\Edgescale@##1{\divide##1 by \Edgect@}%
        \else\EdgGl@b\edef\Edgescale@##1{}\fi
    }

 \def\Edges#1{\Edgect@=#1\relax
     \let\EdgGl@b\global \SetEdges@}

 \Edges{1}

 \def\hhrule{\hrule height \GL@wd\vskip-.\GL@wd}

 \def\hRule@{%
   \advance\gridcount -2%
   \vfil\hhrule\vfil
   \llap{\smash{\raise -2.5pt
     \hbox{\L@Dc\number\gridcount\Zr@s\kern2pt}}}%
   \hhrule
   }

\def\vvrule{\vrule width \GL@wd \kern-\GL@wd}

 \def\vRule@{\advance\gridcount 2%
   \hfil\vvrule\hfil
   \setbox\auxGridbox@=\vbox to 0pt
      {\vskip \Htt@\vskip 2pt
        \hbox to 0pt{\hss\L@Dc\number\gridcount\Zr@s\hss}\vss}%
      \wd\auxGridbox@=0pt \box\auxGridbox@
   \vvrule
   }

 \def\PlaceGrid@@{\gridcount=10 
  \setbox\hGridbox@=\hbox{%
        \hbox{%
             \hskip-.4pt\vrule
             \vbox to \Htt@{%
               \offinterlineskip\parindent=\z@\relax
               \hbox to \Wdd@{\hfil}
               \hRule@\hRule@\hRule@\hRule@
               \vfil\hhrule\vfil}%
             \vrule\hskip-.4pt}
    }%
  \gridcount=0%
  \setbox\vGridbox@=\hbox{%
      \vbox{\offinterlineskip\parindent=0pt\hsize=0pt
         \vskip-.4pt\hrule%
         \hbox to \Wdd@{%
                 \vtop to \Htt@{\vfil}%
                 \vRule@\vRule@\vRule@\vRule@
                 \hfil\vvrule\hfil}%
         \hrule\vskip-.4pt}}%
  \wd\hGridbox@=0pt\ht\hGridbox@=0pt
  \wd\vGridbox@=0pt\ht\vGridbox@=0pt
  \hbox{\box\hGridbox@\box\vGridbox@}%
  }

 \def\LabelsGlobal{\def\LabGl@b{\global}}
 \def\LabelsLocal{\def\LabGl@b{}}
 \LabelsGlobal 

 \def\SetLabels#1\endSetLabels{%
   \LabGl@b\Labeltoks@={#1()\\}%
   }

 \LabGl@b\Labeltoks@={()\\}

 \def\ShowGrid{\LabGl@b\let\PlaceGrid@\PlaceGrid@@}
 \def\HideGrid{\LabGl@b\let\PlaceGrid@\relax}
 \def\Grids{\ShowGrid\LabGl@b\let\GridSwitch@\ShowGrid}
 \def\noGrids{\HideGrid\LabGl@b\let\GridSwitch@\HideGrid}

 \noGrids

 \def\bAdjust@@{%
     \setbox\auxLabelbox@=\hbox{\raise \dp\auxLabelbox@
            \box\auxLabelbox@}}
 \def\bAdjust@{\let\vAdjust@\bAdjust@@}

 \def\eAdjust@@{\dimen0=-.5\ht\auxLabelbox@
     \advance\dimen0 by .5\dp\auxLabelbox@
     \setbox\auxLabelbox@=
            \hbox{\raise\dimen0\box\auxLabelbox@}}
 \def\eAdjust@{\let\vAdjust@\eAdjust@@}

 \def\tAdjust@@{%
     \setbox\auxLabelbox@=\hbox{\raise-\ht\auxLabelbox@
            \box\auxLabelbox@}}
 \def\tAdjust@{\let\vAdjust@\tAdjust@@}

 \let\vAdjust@\relax

 \def\lAdjust@{\let\hAdjust@\rlap}
 \def\rAdjust@{\let\hAdjust@\llap}

 \let\hAdjust@\relax\let\vAdjust@\relax

 \def\FetchLabel@#1(#2)#3\\{%
     \IN@0#2@@\ifIN@
        \setbox0=\hbox{\ignorespaces#1#3\unskip}%
        \ifdim\wd0>0pt
           \ms@g{}%
           \ms@g{ !!! Bad label(s)? !!!}%
           \message{ #1(#2)#3}%
        \fi
        \def\LabelMole@##1\endFetchLabel@{%
            \IN@0()\\@##1@%
            \ifIN@\def\Temp@{\FetchLabel@##1\endFetchLabel@}%
            \else\def\Temp@{}%
            \fi
            \Temp@
           }%
     \else
       \ignorespaces#1\unskip
       \setbox\auxLabelbox@=%
         \hbox to 0pt{\hss\ignorespaces\hAdjust@
          {\ignorespaces#3\unskip}\hss}%
       \vAdjust@
       \let\hAdjust@\relax\let\vAdjust@\relax
       \AugmentLabelBox@@{#2}%
       \ht\Labelbox@=0pt\dp\Labelbox@=0pt
       \let\LabelMole@\FetchLabel@%
     \fi\LabelMole@}

 \newtoks\XYSep@ 
 \def\SetXYSeparator#1{%
     \IN@0#1@@\ifIN@\XYSep@{*}%
     \else
     \XYSep@{#1}%
     \fi
     }

 \SetXYSeparator*

 \def\AugmentLabelBox@@#1{%
     \IN@0\the\XYSep@ @#1@\ifIN@
       \SPLIT@0\the\XYSep@ @#1@%
       \setbox\Labelbox@=\hbox to 0pt{%
         \unhbox\Labelbox@
         \Shifted@@{\the\Initialtoks@\Wddd@}%
         {\the\Terminaltoks@\Httt@}%
         {\box\auxLabelbox@}}%
     \else
         \ms@g{}%
         \ms@g{ !!! Bad insertion point. !!!}%
         \message{ (#1\ this point was rejected.)}%
     \fi
    }

 \def\FetchOption@#1[#2]#3\endFetchOption@{%
    \def\temp{#1}
    \ifx\temp\empty
       \Edgect@=#2\relax
       \let\EdgGl@b\relax
       \SetEdges@
       \Cleaner@#3%
    \fi}

 \def\Cleaner@#1[@]{\Labeltoks@{#1}}
     
 \def\PlaceLabels@@{\mathsurround=0pt
     \def\Cr@{\\}%
     \let\L\lAdjust@\let\R\rAdjust@
     \let\B\bAdjust@\let\E\eAdjust@\let\T\tAdjust@
     \expandafter\FetchOption@\the\Labeltoks@[@]\endFetchOption@
     \Wddd@=\Wdd@ \Edgescale@\Wddd@ 
     \Httt@=\Htt@ \Edgescale@\Httt@
     \expandafter\FetchLabel@\the\Labeltoks@\endFetchLabel@
     \box\Labelbox@
     }%

 \let \PlaceLabels@\PlaceLabels@@

 \def\AffixLabels#1{\setbox\Coordinatebox@=\hbox{#1}%
      \Wdd@=\wd\Coordinatebox@ \Htt@=\ht\Coordinatebox@
      \advance\Htt@ \dp\Coordinatebox@
      \hbox{\copy\Coordinatebox@\kern-\Wdd@ 
           \Shifted@@{0pt}{-\dp\Coordinatebox@}%
           {\PlaceLabels@\PlaceGrid@}%
           \kern\Wdd@}%
      \GridSwitch@ 
      \LabGl@b\Labeltoks@{()\\}%
      }
 
   \let\wlog\LabelFigwlog@ld   
   \catcode`\@=\LabelFigCatAt  

                                By

              Raymond S\'eroul <A18645@FRCCSC21.BITNET>
                                and 
              Laurent Siebenmann <lcs@topo.math.u-psud.fr>
    
              VERSIONS: July 1991, Oct 1991, Jan 1992, July 1992

INTRODUCTION

      This labelling package is intended for TeX users who
rely on non-TeX sources for for their graphics inserts.  It
provides means for adding TeX labels to such inserts with a
minimum of fuss. 

       For most labels, TeX users have in the past found it
reasonably convenient to rely on non-TeX sources. Typical
occasions when an inescapable need for TeX labels seemed to
arise are

 (a) when the graphics program lacks certain exotic or complex
mathematical symbols

 (b) when the very highest typographical quality is wanted for the
labels

 (c) when labels included with the graphics fail to print, 
 and you cannot figure out why (cf. boxedeps.doc).  The labels
 provided by labelfig.tex are 100

       Since this package first appeared, many users, who in the
past scarcely dreamed of using TeX labels, have come to use
nothing but.  So it is now appropriate to add

Intoxication Warning:  TeX labels may be addictive and expensive. 

     If you have a fast preview you may disagree, and even find
that this package provides an agreeable paste-up environment; see
extra applications at end.

     Note to publishers: It is possible and convenient to ultimately
export the TeX labels produced by labelfig.tex to become an integral
part of the EPS file. This is often desired by a publisher who typically
uses an "upmarket" graphics or page layout program, with which the
staff is skilled in perfecting figures.  See Appendix I for
a recipe.

     The authors are grateful to Patrick Ion of Math Reviews for
helpful comments and encouragement.

BASIC INSTRUCTIONS

    After reading in the macro file using

preview or proof your figure with a coordinate grid printed on
top, by typing the following:

    \ShowGrid  
    \AffixLabels{<the graphics insertion>}

Here <the graphics insertion> is what you would type to insert
the graphics object alone without the grid.  This must provide
for the space around it. For example <the graphics insertion>
might well be \BoxedEPSF{MyFigure scaled 700} using the
boxedeps.tex macro package (from same source); this provides a
TeX box containing the encapsulated PostScript insert specified by
the file MyFigure. \AffixLabels{...} provides the grid (supposing
\ShowGrid is present) and later, once you have specified labels
using the grid, it will "tack on" the labels.

     The grid is a sort of (usually elongated) checkerboard of
ten rows and ten columns and its (internal) partitions are by
default numbered  .1, ... ,.9  both horizontally (X-coordinate
running left to right) and vertically (Y-coordinate running bottom
to top).  Thus the points enclosed by the grid correspond to the
points of the unit square in the cartesian "X-Y" plane, the lower
left corner corresponding to the origin (0,0).  By extrapolation,
the full page corresponds to a larger rectangle in the plane.

     These coordinates serve to position labels as follows.
Before the \AffixLabels{...} command type label specifications:

  \SetLabels
   (<X-coordinate>*<Y-coordinate>) <first label> \\
   .
   .
   .
   (<X-coordinate>*<Y-coordinate>)  <last label> \\
  \endSetLabels

Each row specifies one label and is terminated by \\.  In each
row, the position indicator comes first; it is written as a
standard cartesian point except that the X- and Y- coordinates
are separated by * rather than a comma because TeX allows a
comma as decimal point. There are no dimension units to specify
as the unit is the grid itself.

     By default, this cartesian point specifies where the middle
of the baseline of the label will be located.  However if you precede
the point by \L [or \R] the left [or right] edge of the baseline will
be located there. Similarly you may also precede the point by \T, \E,
or \B to vertically align the top equator or bottom of the label box
at the specified point.  This gives nine standard positions of
the label with respect to the insertion point --- corresponding to
the eight principle points of the compas and the center

                     \L\T     \T      \R\T

                     \L\E     \E      \R\E

                     \L\B     \B      \R\B

But this neglects the default "baseline" level of TeX,
giving potentially three more positions

                     \L    <no tag>   \R

For text, the baseline level is often the preferred. Its relation to
the others is variable. It will often coincide with the bottom level,
as happens for "X".  But it is often distinct, as for "g", in which
case you have in all 12 distinct positions rather than 9.

     It is convenient to think of this specification of label
position as attaching the label by a thumb-tack to the coordinate
grid. There are up to twelve positions of the thumb-tack on the
label, while the position of the thumb-tack on the coordinate grid is
arbitrary.  Normally, one choses the position of the thumb-tack on
the label to be the one that is the closest to the item being
labeled.  There are good reasons for this "rule of thumb":

   (a)  It facilitates correct positioning at first try.

   (b)  If the scale of the figure must be altered after labels
have been affixed, the labels have a good chance of remaining well
positioned.

   (c)  The visible grid need not extend beyond the "bounding box"
for the figure, because the best preferred position is always
(at least almost) within the bounding box .

The second reason is particularly important. Indeed it often
happens that scale has to be altered after labelling begins, in
order to either provide space for the labels, or to adjust
proportions between the labels and the figure.  (The size of labels
is unaffected by scaling.)

     Here is an artificial but self-contained test which uses
TeX rules to make a graphics object.

TEST

    Do not skip this!

 \def\FrameIt#1{\hbox{\vrule$\vcenter {\hrule\kern3pt%
             \hbox {\kern3pt #1\kern3pt}%
               \kern3pt\hrule}$\relax\vrule}}

 \def\Caption#1#2{\FrameIt{%
       \vtop {\hsize=#1\relax \parindent=0pt
         \leftskip=0pt \rightskip=0pt plus15pt
         \parfillskip=0pt
         \lineskip=1pt\baselineskip=0pt
         #2}}}

 \def\FirstQuadrant{\hbox to 100pt{\vrule\vbox to 100pt{%
        \hbox to 100pt{\hfil}\vfil\hrule}\hss}}

  \SetLabels
    \R(.5*.2) $\zeta\,\cdot$\\
    (.9*-.10) $\xi$\\
    \R(-.03*.9) $\eta$\\
    \T(.5*.9) \Caption{70pt}{%
          \it The norm of
          $g(\xi+i\eta)$ is indicated on
          contours of this invisible surface.}\\
  \endSetLabels

  \AffixLabels{\FirstQuadrant}

  \end

  Note that the coordinates to use for labels are indicated on the
edges of the grid (when visible) corresponding to the conventional
x- and y- axes of the Cartesian plane. By default the grid is
1-by-1. However, by the command \Edges{100}, you can change this
to 100-by-100 and many users find this alternative most
convenient. Place the command \Edges{...} in your style file (or
header) since its effect is is global. Other possible edge values
are 10 and 1000.

  If you use the command \Edges{...} at all, do so with care.  For
if you accidentally delete an \Edges{...} command your labels will
abruptly be badly misplaced and may logically but mysteriously
generate "dimension too big" errors under TeX and "off page" errors
under your driver.  

  You can dictate the edgescale for an individual figure by giving
the scale in brackets immediately after \SetLabels.  Thus, to
import into an article using say \Edge{100} a figure labelled using
another edgescale, say the original 1-by-1 default, you can use
\SetLabels[1]...\endSetLabels.

GETTING IT DOWN PAT

     Complicated labeling deserves the same respect as
complicated mathematics.  Do not expect it to come out perfect the
first time!  What is needed in either case is a mechanism to
repeatedly typeset troublesome pieces.

     One mechanism is always available.  One does complicated
labelling in a separate "test" file involving just the figure being
labelled;  a texpert will know how to \dump TeX's current state as
a temporary format that restarts rapidly at each retry.  Usually,
one then pastes the completed labelled figure back into the main
TeX file, but, of course, one can also \input it as an auxiliary
file.

     If you do not have a TeXpert at handy, here is a first
approximation to an efficient setup. By deletions reduce a copy
of your article to just a few lines before and after the figure.
Now label the figure, and finally, copy and paste the labelled
figure to the original article. Then copy the next figure to label
into this testbed and repeat. The TeXpert can improve the  speed
at which TeX starts up, by compiling a format specifically for
your article; just one caution: best NOT include in the format
ephemeral details of setup like \Set<mydriver>ArtSpecials (from
boxedeps.tex because this reads  figure dimensions which you may
change during your work session.

     An improved mechanism to repeatedly typeset troublesome
pieces is now available on the Macintosh; it is called LinoTeX;
see the same ftp sources.  It could be set up on many types
of computer.

     Before using labelfig.tex to attach labels to a graphics
object inserted using boxedeps.tex or BoxedArt.tex, make it a
firm rule to carefully adjust the bounding box using the trimming
commands of these packages, and also at least tentatively scale
and position the object. Beware of changing the grid inadvertently
after the labels have been positioned.  For example, correcting
the bounding box of a PostScript graphics object can foul up the
labels by changing the coordinate grid to which the labels are
attached. This is particularly true for the trimming  commands of
boxedeps.tex and BoxedArt.tex. However, as noted already, change
of scale is much less disruptive, and modest adjustments should be
well tolerated.

     Sometimes the labels protrude so far from the bounding box
of a figure that the figure has to be repositioned.  Best do this
by ad hoc spacing, say using \hglue and \vglue; altering the
bounding box would create a vicious circle.

     Remember that you are responsible for preventing labels
from overlapping. You are responsible for all label typography
including size and style. A label is really just about anything
that can be put in a TeX box. Note that spaces at the beginning
and end of labels will normally be suppressed; if you really want
them you must protect them with TeX braces.

     This package temporarily sets the \mathsurround parameter
of TeX to zero  while the labels are being affixed. This is done
because nonzero \mathsurround space would influence the position
of left and right aligned labels; then, when a texpert or printer
modifies mathsurround, diagram labeling might be disastrously
altered. There is a small price to pay involving labels that are
formatted as caption boxes including mathematics: you  may want or
need to specify an explicit mathsurround space within the caption
box; it will not influence anything outside.

     Those hostile to the use of * as separator between
the X and Y coordinates of label insertion points, are free to
impose another using \SetXYSeparator{<the new separator>}.  
Americans may prefer "," to "*" since they never use a 
comma as a decimal point; on the other hand, * may be more visible.

APPENDIX (I)  MERGING labelfig.tex LABELS INTO AN EPSF GRAPHICS OBJECT.

     As promised in the introduction, here is a recipe useful for
publishers. It works at least on Macintosh and at least for vectorized
graphics and Adobe type1 fonts.  (There is surely a similar recipe for
PCs under MSWindows.)

 (a)  Use boxedeps.tex utility to integrate the figure given by the eps
file, "x.eps" say, with a visible frame around it.  See
\ShowDisplacementBoxes command in boxedeps.tex.  To get precise results
automatically it is important to use the \Trim... commands of
boxedeps.tex making the "DisplacementBox" neatly fit the figure.

 (b)  Use the TeX printer driver and LaserWriter (versions >= 8.1.1) to
export to an EPSF the DVI page containing the integrated, labelled
figure. You now have an EPS file  "xx.eps"  that contains too much, and at
the wrong scale, and at wrong position.

 (c)  Convert the EPSF to an Adode Illustrator format EPSF using
the shareware utility called epsConvert by Sam Weiss
1993-- (currently $25).

 (d)  In Illustrator (or a compatible program), group the labels and the
"DisplacementBox"; copy them to the clipboard and paste them into "x.ps".
This step requires that all the label fonts be "visible to the Macintosh.

 (e)  Translate and scale the pasted group consisting of the labels plus
the "DisplacementBox" so as to make the "DisplacementBox" the bounding
box of (labelless) figure represented by "x.eps".  At this point the
labels will be correctly placed on the figure "x.eps".

 (f)  Ungroup and delete the "DisplacementBox".  The result is the
desired single EPS file, "x+.eps" say, It contains the original figure
plus its labels.  

     Using grouping and ungrouping appropriately in "x+.eps", a
publisher's staff can very efficiently improve label positions etc.

APPENDIX II)  SOME EXOTIC APPLICATIONS

     The grid of labelfig.tex is analogous to a light-table in
classical page makeup with wax or latex glue.  In principle, you
can use it to compose any page from its indivisible parts.  This
even has some of the artisanal charm of classical paste-up
provided you have a fast screen preview to make the process
"interactive".

     In practice labelfig.tex is a tool for nonstandard jobs.
Here are a few going beyond the labelling already discussed.

(I)  GRAPHICS INTEGRATION.

     This is accomplished by treating the imported graphics
objects as labels.  The underlying graphics object is then
typically an empty  \vbox to <dimension>{\vfill} in a TeX
\midinsert...\endinsert construction.  A label line
might be of the form

   (.1*.1) \\

The exact form of the special command varies from driver to
driver.  However, in the case of encapsulated PostScript graphics
(EPSF norm), by relying on boxedeps.tex, one can have the
following standard syntax (independant of driver  (see
boxedeps.doc for details.
  
  (.1*.1) \BoxedEPSF{MyFigure scaled <scale in mils>}\\

This may be slow since it requires TeX to read the PostScript
file to read bounding box using many complex macros.  So you
may want to try

  (.1*.1) \EPSFSpecial{MyFigure}{<scale in mils>}\\

which is fast and driver independant, but it squashes the
bounding box, normally to its lower left corner.

     Similarly for graphics of the Macintosh PICT norm ---
using BoxedArt.tex (same sources) in place of boxedeps.tex.

     This approach to integration is to be recommended when
one is assembling a composite graphics object.

 (II)  COMMUTATIVE DIAGRAM ENHANCEMENT

     Commutative diagrams or arrays of mathematical objects
connected by arrows of various sorts are common in mathematics.
The mathematical objects require the use of TeX.  Recently TeX
acquired a good collection of arrows of all slopes --- that of
LamSTeX --- plus pwerful macros to build the diagrams.

     However, even the LamSTeX collection is often
inadequate; it lacks for example double shafted arrows, dotted
arrows and curved arrows. Fortunately it is possible to produce
such arrows on an individual basis using sophisticated graphics
programs such as Illustrator and AldusFreehand (both serving
the EPSF norm) or using Metafont (with its public domain norm).
Since the creation of each new arrow is a work of love, you
probably want to limit the number of arrows by using LamSTeX
for most arrows. The 40K commutative diagram module of LamSTeX
has been adapted to work with AmSTeX and a copy may be posted
with LabelFig and related files. Unfortunately no one has yet
offered a version that works with Plain TeX or LaTeX.

       Suffice it here to say that when the exotic arrow has
been somehow imported into TeX, labelfig.tex treats it as a
label that one affixes to the commutative diagram.  Two other
steps will be treated in separate notes, namely the matter of
extracting the dimension specifications for the arrow and the
construction of the arrow --- for these steps are far from
unique and often depend intimately on your computer environment. 
Notes for the Macintosh-Textures-Illustrator combination are
found in the file ExoticArrows.doc.

 (III) NESTING 

Ingenuity pays off in exploiting labelfig.tex. One can
mix graphics and typography quite freely.  labelfig.tex is good
for freeform or overlapping arrangements, while boxedeps.tex (or
BoxedArt.tex) is best for regimented non-overlapping
arrangements --- and the two can be combined.

     The default behavior of labelfig.tex is not ideal 
for nesting objects, because to prevent trouble for beginners
the register for labels is globally cleared when \AffixLabels
concludes.  But there are switches available

      \LabelsGlobal      \LabelsLocal

which change this.  To understand this, extend the above test 
by something like:

 \LabelsLocal

 \SetLabels
    (.5*.5) AAA\\
 \endSetLabels

 {
 \SetLabels
    (.5*.5) ZZZ\\
 \endSetLabels
   \AffixLabels{\FirstQuadrant}
 }

   \AffixLabels{\FirstQuadrant}

     There are however potential pitfalls.  Neither
labelfig.tex nor boxedeps.tex has been tested under extreme
conditions. Problems may occur if their procedures are
indiscriminately nested. For boxedeps.tex (not labelfig.tex)
there is a precise cause for worry, namely many of its
variables are "global", which means that TeX braces will not
provide the protection one might expect.

COMMAND SUMMARY FOR labelfig.tex

  Here [...] means optional (one or zero)
       [...]* means any number of such constructs

  \SetLabels
    [[<P>](<X><Sep><Y>) <label> \\]*
  \endSetLabels
  \ShowGrid  
  \AffixLabels{<the figure>}

   --- <P> is tack position, one of eleven or empty
              order irrelevant

                   \L\T      \T      \R\T

                   \L\E      \E      \R\E

                     \L               \R

                   \L\B      \B      \R\B

   --- (<X><Sep><Y>) insertion point;
  <Sep> is separator, = * by default;
  \SetXYSeparator{<Sep>} changes it.
   <X> and <Y> are real numbers

  --- <label> a label to attach 

  --- <the figure> the figure to label 

  \GlobalLabels (default)     
  \LocalLabels  setting for nested constructs.

 \Grids makes ALL grids appear; \HideGrid then makes just next disappear.
 \noGrids returns to default.  The commands are always global.

 \GridLineWidth{<dimension>} adjusts width of grid lines. Default is very
small, to give "hairline" effect. If your grid lines are missing try
setting \GridLineWidth{1pt}.

 \Edges#1 globally changes the edge size of all grids to the numerical 
value #1, which must be 1, 10, 100, or 1000.  The default is 1.

VERSION HISTORY.
 --- Jan 1993: \Edges#1 and [??] option after \SetLabels
 --- July 1992: \Grids, \noGrids, \HideGrid;
       Gridlines become hairlines; \GridLineWidth{<dimension>}.
 --- Oct 1991, Jan 1992: \SetXYSeparator{<Sep>},  \LabelsGlobal,
       \LabelsLocal.
 --- July 1991: first release

Address for bugs and other feedback:

        Raymond S\'eroul
        IREM and Lab. de Typographie Informatise
        Univ. Rene Descartes
        Strasbourg

    Tel 33-88-41-63-45
    Email:  A18645@FRCCSC21.BITNET

        Laurent Siebenmann
        Mathematique, Bat. 425,
        Univ de Paris-Sud,
        91405-Orsay,
        France

    Tel 33-1-6941-7949; 
    Email: lcs@topo.math.u-psud.fr

\def\scalefig#1{\epsfxsize #1\textwidth}
\def\defeq{\stackrel{\Delta}{=}}

\newcommand {\Ebb}{{\mathbb{E}}}

\newtheorem{theorem}{Theorem}
\newtheorem{lemma}{Lemma}
\newtheorem{definition}{Definition}

\newtheorem{remark}{Remark}

\newtheorem{condition}{Condition}
\newtheorem{algorithm}{Algorithm}

\IEEEoverridecommandlockouts 

\title{{\huge A New Approach to User Scheduling in Massive Multi-User MIMO Broadcast Channels
}}

\author{Gilwon Lee, {\em Student~Member, IEEE} and Youngchul Sung$^\dagger$\thanks{$^\dagger$ Corresponding author}, {\em Senior~Member, IEEE} \thanks{Gilwon Lee and Youngchul Sung
are with Dept. of Electrical Engineering, KAIST, Daejeon, 305-701, South Korea.
E-mail: \{gwlee@ and ysung@ee\}.kaist.ac.kr.
This research was supported by Basic Science Research Program through the National Research Foundation of Korea (NRF) funded by the Ministry of Education (2013R1A1A2A10060852). A preliminary version of this work was submitted to 2014 SPAWC \cite{Lee&Sung:14SPAWC}.}}

\begin{document}

\maketitle

\begin{abstract}
In this paper,  a new  user-scheduling-and-beamforming method is proposed for
multi-user massive multiple-input multiple-output (massive MIMO)
broadcast channels in the context of two-stage beamforming.  The key ideas of  the
proposed scheduling method are 1) to use a set of {\em orthogonal reference} beams and construct a {\em double cone} around each reference beam to select `nearly-optimal' semi-orthogonal users based only on channel quality indicator (CQI) feedback and 2) to apply {\em post-user-selection beam refinement} with zero-forcing beamforming (ZFBF) based  on channel state information (CSI) feedback only from the selected users.
It is proved that the proposed scheduling-and-beamforming method is
asymptotically optimal as the number of users increases.
Furthermore,  the proposed scheduling-and-beamforming method almost achieves the performance of the existing semi-orthogonal user selection with ZFBF (SUS-ZFBF) that requires full CSI feedback from all users,  with significantly reduced feedback overhead which is even less than that required by random beamforming.
\end{abstract}

\begin{keywords}
User scheduling, multi-user MIMO, massive MIMO, two-stage beamforming, multi-user diversity, zero-forcing beamforming
\end{keywords}

\section{Introduction}
\label{sec:intro}

The multi-user multiple-input and multiple-output (MU-MIMO)
technology has served as one of the core technologies of the
fourth generation (4G) wireless systems. With the current
 interest in large-scale antenna arrays at base
stations (BSs), the importance of the MU-MIMO technology further increases
for future wireless systems. The MU-MIMO technology supports users
in the same frequency band and time simultaneously based on
spatial-division multiplexing, exploiting the degrees-of-freedom
(DoF) in the spatial domain. There has been extensive research on
MU-MIMO ranging from transmit signal or beamformer design to user
scheduling in the past decade
\cite{Caire&Shamai:03IT,Weingarten&Steinberg&Shamai:04ISIT,Sharif&Hassibi:05IT,Yoo&Goldsmith:06JSAC}.
It is known that the capacity of a Gaussian MIMO broadcast channel
is achieved by dirty paper coding (DPC)
\cite{Costa:83IT,Caire&Shamai:03IT,Weingarten&Steinberg&Shamai:04ISIT}.
However, due to the difficulty of practical implementation of DPC,  linear beamforming schemes for transmit signal design for MU-MIMO
have become dominant in current cellular standards
\cite{Liuetal:12COMMAG}. In general,  linear beamforming schemes such as
zero-forcing beamforming (ZFBF) and minimum mean-square-error
(MMSE) beamforming perform worse than DPC. However, an astonishing
remedy was brought to these linear beamforming schemes for MU-MIMO
downlink, based on multi-user diversity
\cite{Knopp&Humblet:95ICC,Viswanath&Tse&Laroia:02IT,Sharif&Hassibi:05IT,Yoo&Goldsmith:06JSAC}.
That is, with proper user selection or scheduling, the performance
degradation of the linear beamforming schemes compared to DPC is
negligible as the number of users in the served cell becomes large
\cite{Sharif&Hassibi:05IT,Yoo&Goldsmith:06JSAC}, and the seminal
results in \cite{Sharif&Hassibi:05IT,Yoo&Goldsmith:06JSAC} have
provided guidance on how to select simultaneous users in practical
MU-MIMO downlink systems.

In this paper, we revisit the user scheduling and beamforming problem for MU-MIMO
downlink  in the context of up-to-date massive MU-MIMO downlink
with two-stage beamforming \cite{Adhikary&Nam&Ahn&Caire:13IT},
although the proposed user scheduling method can readily be applied to
conventional single-stage MU-MIMO downlink. User scheduling for MU-MIMO cellular downlink has been
investigated extensively in the past decade
\cite{Sharif&Hassibi:05IT,Yoo&Goldsmith:06JSAC,
Alnaffouri&Sharif&Hassibi:09COM,Adhikary&Caire:13arXiv}. Among
many proposed user scheduling methods for beamforming-based
MU-MIMO downlink are two well-known
  user selection schemes already mentioned in the above sitting on opposite sides on the scale of feedback overhead:  random (orthogonal)
beamforming (RBF) \cite{Sharif&Hassibi:05IT} and semi-orthogonal
user selection with  ZFBF (SUS-ZFBF) \cite{Yoo&Goldsmith:06JSAC}.
Both schemes are asymptotically optimal, i.e., the sum-rate
scaling law with respect to (w.r.t.) the number of users in the
cell is the same as that of DPC-based MU-MIMO, but with difference in the amount of feedback
required for user selection, the two schemes yield
significantly different sum-rate performance in the practical
finite-user case. The SUS-ZFBF method in
\cite{Yoo&Goldsmith:06JSAC} requires  channel state information (CSI) from every user terminal
(UT) at the BS, exploits full CSI from all users, and
provides a smart selection of users whose channel directions are
almost orthogonal ('semi-orthogonal') to yield good performance
with ZFBF. On the other hand, the RBF scheme in
\cite{Sharif&Hassibi:05IT} selects a group of users that are  matched
to predetermined random orthogonal beam directions, and  only requires
the feedback of the best beam direction index and the
corresponding signal-to-interference-plus-noise ratio (SINR) value
from each user. Hence, the feedback overhead is
 reduced significantly in the RBF scheme. (Note that SUS-ZFBF requires CSI from users, whereas RBF requires channel quality indicator (CQI) from users.) Due to this feedback advantage, the RBF scheme was extended to the correlated channel case
\cite{Alnaffouri&Sharif&Hassibi:09COM}. Recently, the RBF scheme
was applied to user scheduling in massive MU-MIMO downlink with
two-stage beamforming with multiple correlated channel groups
\cite{Adhikary&Caire:13arXiv}. However, it is known that the RBF
scheme yields  poor performance compared to SUS-ZFBF
utilizing full CSI in the practical case of finite users
\cite{Yoo&Goldsmith:06JSAC}.

In this paper, we propose a new user-selection-and-beamforming method for (massive)
MU-MIMO downlink that maintains the advantage of  SUS-ZFBF, overcomes the disadvantage of RBF, and
 requires  a less amount of feedback than RBF.
Our approach starts with identifying the factors that make RBF
yield poor sum-rate performance and the factors that make
 SUS-ZFBF better than  RBF in sum-rate performance, and ends
with correcting the loss factors associated with RBF and
implementing the gain factors associated with SUS-ZFBF without
full CSI at the BS.
 RBF selects a set of users with roughly orthogonal channels. However, as we shall see later in Section \ref{subsec:AUS_background}, the main loss factors associated with RBF are 1) the criterion of SINR \cite{Sharif&Hassibi:05IT} associated with orthogonal beams and 2) the absence of any beam refinement after user selection.
 {\em The use of SINR as the selection criterion cannot properly take the channel magnitude into consideration but only considers the angle between each predetermined beam direction and  the user channel vector for user selection}. On the other hand,
SUS-ZFBF also provides a set of users with semi-orthogonal channel
vectors, based on full CSI at the BS. However, SUS-RBF selects
semi-orthogonal users  with  large channel magnitude. Furthermore,
SUS-ZFBF performs post-user-selection beam design based on the
selected users' CSI, i.e., SUS-ZFBF uses ZFBF for selected users
to eliminate the inter-user interference, and the effective channel
gain loss associated with ZFBF is managed by controlling the
thickness of the user-selection hyperslab\footnote{A hyperslab in
${\mathbb{C}}^M$ is defined as a set $\{ \wbf \in {\mathbb{C}}^n:
\frac{|\nbf^H \wbf|}{||\nbf||||\wbf||} \le \gamma \}$ for a given vector
$\nbf$ and a constant $\gamma
>0$.} \cite{Yoo&Goldsmith:06JSAC}. Consequently, SUS-ZFBF selects
a set of (semi-)orthogonal users with large channel
magnitude and the loss associated with ZFBF is made small by
increasing the orthogonality of the selected user channel vectors
by thinning the user-selection hyperslab when the number of users
in the cell increases.  As in RBF, we use a set of {\em orthogonal
reference} beam directions, but correct the first loss factor of
RBF by defining a new criterion names as {\em quasi-SINR} that
can incorporate the channel magnitude. The user selection is done
based on quasi-SINR feedback. Then, we apply post-user-selection
ZFBF for selected users. Here, we do control the semi-orthogonality of selected users and  the effective channel
gain loss associated with ZFBF by defining a user-selection {\em
double cone} around each reference beam direction and by
controlling the angle of the user-selection double cone.  Under
the proposed user selection method, user selection is done based
on quasi-SINR, which is a CQI, and the post-user-selection beam
refinement is done based on the feedback of CSI from the selected
users only.  In this way, the main advantage of SUS-ZFBF is
implemented in the proposed scheme without full CSI at the BS. The
proposed scheduling-and-beamforming method is asymptotically optimal and the proposed  method almost achieves the performance of SUS-ZFBF, with significantly reduced feedback overhead which is even less than that required by RBF.


\textit{Notation and Organization:} We will make use of standard notational
conventions. Vectors and matrices are written in boldface with
matrices in capitals. All vectors are column vectors. For a matrix
$\Abf$, $\Abf^T$, $\Abf^H$, $\mbox{tr}(\Abf)$, and $[\Abf]_{i,j}$ indicate the transpose, conjugate
transpose, trace, and  entry at the $i$-th row and $j$-th column of
$\Abf$, respectively. $\Abf(:,c_1:c_2)$ is the submatrix of $\Abf$ consisting of the columns from $c_1$ to $c_2$. $\mbox{diag}(\Abf_1,\cdots,\Abf_n)$ denotes
a diagonal matrix composed of diagonal elements $\Abf_1$, ...
,$\Abf_n$. For vector $\abf$, $\|\abf\|$  represents the 2-norm of
$\abf$. $\Ibf_K$ is the $K \times K$ identity matrix.
$\xbf\sim\Cc\Nc(\mubf,\Sigmabf)$ means that random vector $\xbf$
is complex Gaussian distributed with mean $\mubf$ and covariance
matrix $\Sigmabf$, and $\theta \sim \text{Unif}(a,b)$ means that
$\theta$ is uniformly distributed for $\theta \in [a,b]$.
$\Ebb[\cdot]$ denotes statistical expectation.  ${\mathbb{R}}$, ${\mathbb{R}}^+$, and ${\mathbb{C}}$ are the sets of real, non-negative real, and complex numbers, respectively. $\iota = \sqrt{-1}$.

The remainder of this paper is organized as follows.
The system model and
preliminaries  are described in Section \ref{sec:sysmodel}. The proposed user-scheduling-and-beamforming
method is presented in Section
\ref{sec:AUS} and its asymptotic optimality is proved in Section \ref{sec:AUS_opt}. Fairness issues are discussed in Section \ref{sec:fairness}. Numerical results are provided in Section
\ref{sec:num_res}, followed by conclusions in Section
\ref{sec:conclusion}.

\section{System Model and Preliminaries}
\label{sec:sysmodel}

We consider massive MU-MIMO downlink. One of the major challenges to implement
 massive MIMO systems in the real world is the  design of a practical precoding architecture for multi-user massive MIMO downlink together with the estimation of channels with high dimensions.
  Designing precoding vectors or matrices with
 very high dimensions with the scale of hundred to support a large number of simultaneous UTs without introducing an
 efficient structure requires heavy
 computational burden and a huge amount of CSI feedback.
 One practical precoding solution to multi-user massive
 MIMO downlink is two-stage beamforming, which is based on a
 `divide-and-conquer' approach. Recently, Adhikary {\it et al.}
 proposed an efficient two-stage beamforming architecture named `Joint Spatial Division and Multiplexing
 (JSDM)' for multi-user massive MIMO downlink
\cite{Adhikary&Nam&Ahn&Caire:13IT}. The key ideas of JSDM are 1) to
partition users in a sector (or cell) into multiple groups so that each group has a
distinguishable linear subspace spanned by the dominant eigenvectors of the group's  channel
covariance matrix and 2) to divide transmit beamforming into two
stages based on this grouping, as shown in Fig. \ref{fig:two_stage_bf}: The first stage is pre-beamforming that separates multiple groups by using a pre-beamforming matrix designed for each group to filter the dominant
eigenvectors of each group's channel covariance matrix,  and the second stage is conventional MU-MIMO precoding  that separates and thus simultaneously supports the users within a group based on
 the effective channel which is given by the product of the pre-beamforming matrix and the actual
channel matrix between the antenna array and UTs.
One key advantage of such two-stage beamforming is that
the
pre-beamforming matrices can be designed not based on CSI but based on the channel statistics information, i.e., the channel covariance matrix. The channel statistics information varies slowly compared to CSI and thus can be estimated more easily than  CSI, based on some subspace tracking algorithm without knowing
instantaneous CSI \cite{Herdin&Bonek:04ISTMWC,Hoydis&Hoek&Wild&Brink:12ISWCS,Ispas&Dorpinghaus&Ascheid&Zemem:13SP,Ispas&Schneider&Ascheid&Thoma:10VTC}.
Practically,  the channel covariance matrix of
a UT can be determined {\it a priori} based on some side
information \cite{Adhikary&Nam&Ahn&Caire:13IT,Noh&Zoltowski&Sung&Love:14JSTSP}.
Furthermore, the channel covariance matrix associated with a UT or a group in a realistic
environment has a much smaller rank than the actual size of the antenna array and therefore, the dimension of the effective channel whose state information is necessary at the BS is significantly reduced since the effective channel is formed as a precoding matrix and the actual channel matrix with two-stage beamforming.

\begin{figure}[t]
\begin{psfrags}
        \psfrag{d1}[c]{\small $\dbf_1$} %
        \psfrag{dG}[c]{\small $\dbf_G$} %
        \psfrag{W1}[c]{\small $\Wbf_1$} %
        \psfrag{WG}[c]{\small $\Wbf_G$} %
        \psfrag{b1}[c]{\small $b_1$} %
        \psfrag{bG}[c]{\small $b_G$} %
        \psfrag{vd}[c]{\small $\vdots$} %
        \psfrag{V1}[c]{\small $\Vbf_1$} %
        \psfrag{VG}[c]{\small $\Vbf_G$} %
        \psfrag{1}[c]{\small $1$} %
        \psfrag{2}[c]{\small $2$} %
        \psfrag{M}[c]{\small $M$} %
        \psfrag{S1}[c]{\small $S_1$} %
        \psfrag{SG}[c]{\small $S_G$} %
        \psfrag{group1}[c]{\small Group $1$} %
        \psfrag{group2}[c]{\small Group $2$} %
        \psfrag{group3}[c]{\small Group $3$} %
        \psfrag{groupG}[c]{\small Group $G$} %
        \psfrag{base}[c]{\small Base station} %
    \centerline{ \scalefig{0.75} \epsfbox{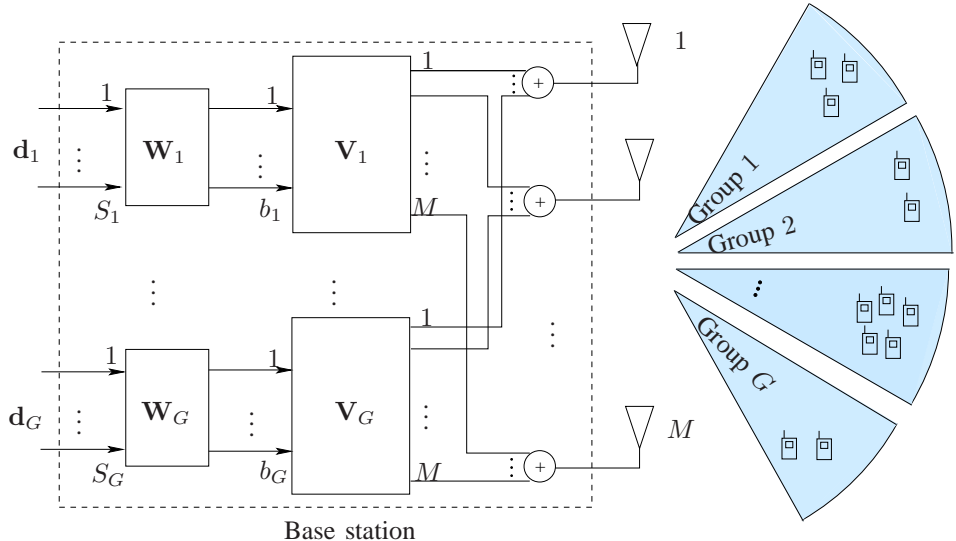} }
    \caption{Multi-group MU-MIMO downlink with
    two-stage beamforming}
    \label{fig:two_stage_bf}
\end{psfrags}
\end{figure}

With the above-mentioned advantages in mind, specifically we consider a single-cell massive MIMO downlink system adopting JSDM composed of a single BS employing $M$ transmit antennas and $K$ single-antenna UTs, as shown in Fig.  \ref{fig:two_stage_bf}.
We consider the large network regime, i.e.,
 $K \gg M$,
  and assume that the BS chooses $S~(\le M)$ users among
the $K$ users within the cell and broadcasts independent data streams to the
selected users.
We assume that
the users in the cell are partitioned into $G$ groups such that (s.t.) $\sum_{g=1}^G K_g = K$ and
$\sum_{g=1}^G S_g =S$,  where $K_g$ and $S_g$ are the number of users and the number of
independent data streams in group $g$, respectively. We also assume that each
 group has a different channel covariance matrix and every user in a group has the same channel covariance matrix, as in \cite{Adhikary&Nam&Ahn&Caire:13IT}.
Then, the $M \times 1$ channel vector $\hbf_{g_k}$ of user $k$ in group $g$
is given by
\begin{equation}  \label{eq:channel_linear_comb}
\hbf_{g_k} = \Ubf_g \Lambdabf_g^{1/2} \etabf_{g_k},
\end{equation}
where  $\Rbf_g = \Ubf_g \Lambdabf_g \Ubf_g^H$ is the eigendecomposition of
the channel covariance matrix $\Rbf_g$ of group $g$;
 $\Ubf_g$ is the $M \times r_g$ matrix composed of
the orthonormal eigenvectors corresponding to the $r_g$ non-zero
eigenvalues of $\Rbf_g$;
 $\Lambdabf_g$ is the $r_g \times r_g$ diagonal matrix
composed of the non-zero eigenvalues of $\Rbf_g$; and
$\etabf_{g_k} \sim \Cc\Nc({\bf 0},\Ibf_{r_g})$. Let the elements of $\Ubf_g$, $\Lambdabf_g$ and $\etabf_{g_k}$ be
\begin{align}
\Ubf_g &=[\ubf_{g,1},\ubf_{g,2},\cdots, \ubf_{g,r_g}] \\
\Lambdabf_g  &= \mbox{diag}(\lambda_{g,1},\cdots,\lambda_{g,r_g}), ~~~ \lambda_{g,1} > \lambda_{g,2} > \cdots > \lambda_{g,r_g}  \label{eq:model_channel_eigenvalue} \\
\etabf_{g_k} &= [\eta_{g_k,1},\eta_{g_k,2},\cdots,\eta_{g_k,r_g}]^T \sim \Cc\Nc({\bf 0},\Ibf_{r_g}). \label{eq:model_channel_component}
\end{align}
One widely-used practical channel model is the {\em one-ring} model
 introduced by Jakes \cite{Jakes:book}, which captures the situation in which the BS antenna is positioned in high elevation and UTs are surrounded by radio scatters as in a typical urban cell.  In the one-ring model, the channel covariance matrix is determined by
the angle
spread (AS), angle of arrival (AoA), and antenna
geometry \cite{Shiu&Foschini&Gans&Kahn:00COM}, and  in the case
of a ULA at the BS with the antenna spacing $\lambda_c
D$, the channel covariance matrix is expressed as \cite{Shiu&Foschini&Gans&Kahn:00COM}
\begin{align}
[\Rbf_t]_{i,j} &= \frac{1}{2\Delta}
\int^{\theta+\Delta}_{\theta-\Delta} e^{-\iota 2\pi (i-j) D
\sin \omega} d\omega, \label{eq:channelCovOneRing}
\end{align}
where $\lambda_c$ is the carrier wavelength, $\theta$ is the AoA,  and $\Delta$ is the AS.
One useful thing to note is that with a large uniform linear or planar antenna array
(ULA)  at the BS, the eigenvectors of the channel covariance
matrix  reduce to certain columns of the discrete fourier transform
(DFT) matrix depending on the AoA and AS of UTs
\cite{Adhikary&Nam&Ahn&Caire:13IT,Noh&Zoltowski&Sung&Love:14JSTSP}.

Denoting  the
$K_g \times M$ channel matrix for the users in group $g$ by $\Hbf_g = [\hbf_{g_1},\cdots,\hbf_{g_{K_g}}]^H$ and stacking $\{\Hbf_g, g=1,\cdots,G\}$, we have
the overall $K \times M$ channel matrix $\Hbf = [\Hbf_1^H, \cdots,
\Hbf_G^H]^H$. Then, the received signal vector containing all user signals in the
cell is given by
\begin{equation}\label{eq:receiv_sig}
\ybf = \Hbf\xbf+ \nbf,
\end{equation}
where $\xbf$ is the $M \times 1$ transmitted signal vector at the BS,
$\nbf \sim \Cc \Nc ({\bf 0}, \Ibf_K)$ is the noise vector, and the
BS has an average power constraint $\mathbb{E}[\|\xbf\|^2]\le P$.
In the considered two-stage beamforming, precoding of the data vector $\dbf$ is done in two steps: first, by a $b \times S$ MU-MIMO precoder $\Wbf$
and then by a $M \times b$ pre-beamformer $\Vbf$,  i.e.,
\[
\xbf =
\Vbf\Wbf\dbf,
\]
 where  $\dbf \sim \Cc\Nc({\mathbf{0}}, \Ibf_S)$.  As mentioned earlier, the pre-beamforming matrix $\Vbf = [\Vbf_1,\cdots,\Vbf_G]$ is
designed based not on the instantaneous CSI but on the channel {\it statistics} information
$\{\Ubf_g,\Lambdabf_g\}$, where the $M \times b_g$
submatrix $\Vbf_g$  is  the
pre-beamforming matrix for group $g$. Then,  the received signal in \eqref{eq:receiv_sig} can be
rewritten  as \cite{Adhikary&Nam&Ahn&Caire:13IT}
\begin{equation}
\ybf = \Gbf\Wbf\dbf + \nbf,
\end{equation}
where
\begin{equation}   \label{eq:equivalentCh}
\Gbf :=\Hbf\Vbf = \left[
\begin{matrix}
{\Hbf_1\Vbf_1 } & {\Hbf_1\Vbf_2} & {\cdots} & {\Hbf_1\Vbf_G} \\
{\Hbf_2\Vbf_1 } & {\Hbf_2\Vbf_2} & {\cdots} & {\Hbf_2\Vbf_G} \\
{\vdots } & {\vdots} & {\ddots} & {\vdots} \\
{\Hbf_G\Vbf_1 } & {\Hbf_G\Vbf_2} & {\cdots} & {\Hbf_G\Vbf_G}
\end{matrix}
\right].
\end{equation}
Although the MU-MIMO precoder $\Wbf$ can be designed with full freedom, for simplicity, $\Wbf$ is designed in a block-diagonal
form as $\Wbf = \text{diag}(\Wbf_1,\cdots,\Wbf_G)$,
where $\Wbf_g$ is the $b_g \times S_g$ MU-MIMO precoder
 and depends  on
the {\em effective channel} $\Gbf_g:=\Hbf_g\Vbf_g$ for group $g$ only.
Hence, in JSDM, the received signal vector for
the users in group $g$ is given by
\begin{equation}  \label{eq:RXsigybfg}
\ybf_g = \Gbf_g\Wbf_g\dbf_g
+ \sum_{g'\neq g} \Hbf_g \Vbf_{g'} \Wbf_{g'} \dbf_{g'}
+ \nbf_g,
\end{equation}
where $\dbf_g$ and $\nbf_g$ are the data and noise vectors for
group $g$, respectively. Decomposing $\Gbf_g$ and $\Wbf_g$ as  $\Gbf_g =
[\gbf_{g_1},\cdots,\gbf_{g_{K_g}}]^H$ and
$\Wbf_g=[\wbf_{g_1},\cdots,\wbf_{g_{K_g}}]$, respectively, we have
the received signal of user $k$ in group $g$ (from here on, we will simply say user $g_k$ for user $k$ in group $g$),  given by
\begin{align}\label{eq:rec_gk}
y_{g_k} &= \gbf_{g_k}^H \wbf_{g_k} d_{g_k}
+ \sum_{k' \neq k} \gbf_{g_k}^H \wbf_{g_{k'}}d_{g_{k'}}
+
\sum_{g' \neq g} \hbf_{g_k}^H \Vbf_{g'} \Wbf_{g'}\dbf_{g'} + n_{g_k}
\end{align}
where $\gbf_{g_k}$, $\wbf_{g_k}$, $d_{g_k}$ and $n_{g_k}$ are the $b_g \times 1$
effective channel, $b_g \times 1$ MU-MIMO precoding vector, data and
noise symbols of user $g_k$, respectively. Note that the dimension of the effective channel $\gbf_{g_k}$ for user $g_k$ is reduced to $b_g$, and $b_g \ll M$ in typical cellular environments with sufficiently high carrier frequency \cite{Adhikary&Nam&Ahn&Caire:13IT}. The second
and third terms in the right-hand side (RHS) of \eqref{eq:rec_gk}
are the intra-group and inter-group interference, respectively.
Concerning the inter-group interference, we assume that at least the approximate block diagonalization (BD) condition in the below holds \cite{Adhikary&Nam&Ahn&Caire:13IT}:

\begin{condition}[Inter-group interference condition \cite{Adhikary&Nam&Ahn&Caire:13IT}]  \label{cond:approxBD}

~

\begin{itemize}
\item {\em Exact BD}: Each group has a sufficient signal space to transmit
$S_g$  data streams, that does not interfere with the signal spaces of other groups,
i.e., \cite{Adhikary&Nam&Ahn&Caire:13IT}
\begin{equation}
\dim\left(\text{span}(\Ubf_g) \cap \text{span}^\perp(
\{\Ubf_{g'}:g' \neq g\}) \right) \ge S_g.
\end{equation}

\item {\em Approximate BD}: When exact BD is impossible,
approximate BD can be attained  by selecting
a matrix $\Ubf_g^*$ consisting of the $r_g^* ~(\le r_g)$
dominant eigenvectors of the channel covariance matrix for each group $g$ such that
 \cite{Adhikary&Nam&Ahn&Caire:13IT}
\begin{equation}
\dim\left(\text{span}(\Ubf_g^*) \cap \text{span}^\perp(
\{\Ubf_{g'}^*:g' \neq g\}) \right) \ge S_g.
\end{equation}
\end{itemize}
\end{condition}

Note that in the case of approximate BD, $r_g^*$ is a control parameter and the inter-group
interference still remains in \eqref{eq:rec_gk} because of  the weakest
$r_g-r_g^*$ eigenvectors of the channel covariance matrix not included in $\Ubf_g^*$. Note that
both $\Ubf_g$ and $\Ubf_g^*$ have orthonormal columns since they are column-wise submatrices of a
unitary matrix. Hence, the average transmit power for user $g_k$ is given by
\begin{align}  \label{eq:power_formula}
P_{g_k}^{actual} &= \mbox{tr}(\Vbf_g \wbf_{g_k} \wbf_{g_k}^H\Vbf_g^H)=\|\wbf_{g_k}\|^2
\end{align}
when we set $\Vbf_g=\Ubf_g^*$ for the pre-beamforming matrix, since the variance of $d_{g_k}$ is set to one.

\section{The Proposed User Scheduling Method}
\label{sec:AUS}

In this section, we propose a user-scheduling-and-beamforming algorithm for a
given pre-beamformer $\Vbf=[\Vbf_1,\cdots,\Vbf_G]$, adopting ZFBF
for the second-stage MU-MIMO precoder $\Wbf_g$. For the sake of
simplicity, we assume $b_g=S_g=r_g^*$ and $\Vbf_g = \Ubf_g^*$ for
all $g$. We also assume that each receiving user $g_k$ (not the
BS) knows its {\em effective} CSI $\gbf_{g_k}$.

\subsection{Background}

\label{subsec:AUS_background}

Before presenting the proposed user-scheduling-and-beamforming method, we briefly
examine the two disparate user-scheduling-and-beamforming methods in
\cite{Sharif&Hassibi:05IT} and \cite{Yoo&Goldsmith:06JSAC} devised under the linear beamforming framework.  For simplicity,
let us just consider one group. First, consider the random
(orthogonal) beamforming (RBF) in \cite{Sharif&Hassibi:05IT}. In
this method, the BS just randomly determines a set of orthonormal
beam vectors $\{ \phibf_1, \phibf_2, \cdots, \phibf_{r_g^*} \}$,
and then transmits each beam  sequentially in time during the
training period. In the setting of JSDM, this beam selection
corresponds to \cite{Adhikary&Caire:13arXiv}
\begin{equation}
\phibf_i = \ubf_{g,i}, ~~i=1,2,\cdots, r_g^* ~~~\mbox{and}~~~\Wbf_g = \Ibf_{r_g^*}.
\end{equation}
During the training period, user $g_k$ computes the SINR of each beam direction $i$ as \cite{Sharif&Hassibi:05IT}
\begin{equation}  \label{eq:RBFsinr}
{\mathrm{SINR}}_{g_k,i} = \frac{|\hbf_{g_k}^H \phibf_i|^2}{1+ \sum_{i^\prime \ne i}|\hbf_{g_k}^H \phibf_{i^\prime}|^2},~~~i=1,\cdots,r_g^*,
\end{equation}
assuming that $\sum_{i=1}^{r_g^*} d_{\kappa_i}\phibf_i$ will be
transmitted during the data transmission period, where $\kappa_i$
is the selected user index for beam direction $i$. (The
inter-group interference is neglected for simplicity.) Then, each
user $g_k$ feeds back its maximum SINR value and the corresponding
beam index $i$  \cite{Sharif&Hassibi:05IT}.  After the feedback
period is finished,  the BS selects a user for each beam direction
$i$ such that the selected user for beam direction $i$ has the
maximum SINR for the considered beam direction $i$. (For
simplicity, let us neglect the case that one user can be selected
for more than one beam direction.) After the selection is done,
the BS transmits $\sum_{i=1}^{r_g^*} d_{\kappa_i}\phibf_i$ to
serve the selected $r_g^*$ users.

\begin{figure}[ht]
\begin{psfrags}
    \psfrag{o}[l]{\small origin} %
    \psfrag{cth}[c]{\small $\theta$} %
    \psfrag{h}[l]{\small $\hbf_{g_{\hat{k}}}$} %
    \centerline{ \scalefig{0.5} \epsfbox{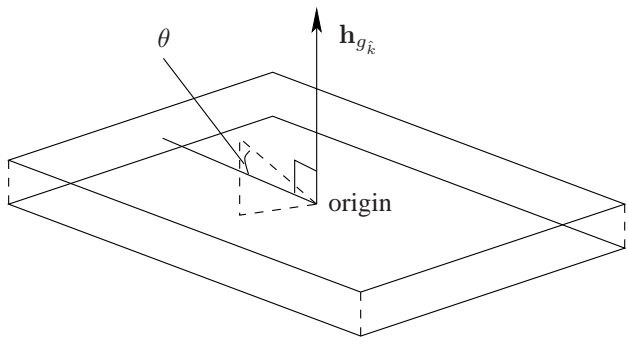} }
    \caption{A hyperslab constructed by an already-included user channel vector}
    \label{fig:sus_hyperslab}
\end{psfrags}
\end{figure}

On the other hand, under SUS-ZFBF in \cite{Yoo&Goldsmith:06JSAC},
the BS collects CSI $\hbf_{g_k}$ from every UT $g_k$ in the
beginning, and sequentially selects $r_g^*$ users by enforcing
semi-orthogonality among the selected users. That is, the BS first
selects the user that has the largest channel magnitude. Let the
firstly selected user's index be $g_{\hat{k}}$. Then, based on the CSI
$\hbf_{g_{\hat{k}}}$, SUS-ZFBF constructs a user-selection
hyperslab defined as \cite{Yoo&Goldsmith:06JSAC}
\begin{equation}
\Hc_{g,1} = \left\{ \wbf \in {\mathbb{C}}^M:
\frac{|\hbf_{g_{\hat{k}}}^H
\wbf|}{||\hbf_{g_{\hat{k}}}||\cdot||\wbf||} \le \gamma \right\}
\end{equation}
as shown in Fig. \ref{fig:sus_hyperslab}.  This means that in Fig.
\ref{fig:sus_hyperslab}, the angle $\theta$ is determined to
satisfy  $\cos \theta \le \gamma$. Note that if a vector $\wbf$ is
contained in $\Hc_{g,1}$, $\wbf$ is semi-orthogonal to
$\hbf_{g_{\hat{k}}}$. Then, SUS-ZFBF selects the user whose
channel vector is contained in the hyperslab $\Hc_{g,1}$ and that
has maximum channel vector magnitude within $\Hc_{g,1}$. After the
second user is selected, another hyperslab contained in the first
hyperslab is constructed based on the secondly selected user's
channel vector. Thus, the newly constructed hyperslab is
semi-orthogonal to both the firstly and secondly selected users'
channel vectors. In this way, at each step the user with the
largest channel magnitude is selected while semi-orthogonality is
maintained among the selected users. Furthermore, the effective
channel gain loss associated with later ZFBF can be made small by making
the thickness of the hyperslab small for a large number of users
in the served cell. (For detail, refer to
\cite{Yoo&Goldsmith:06JSAC}.)

Now, consider RBF explained previously again. First, we examine
the SINR defined in \eqref{eq:RBFsinr}. Consider a very practical
scenario of signal-to-noise ratio (SNR) of 3 dB and  four beam
directions. Assume that $\hbf_{g_k}$ has equal size components in
$\{\phibf_i\}$. Then, $\frac{|\hbf_{g_k}^H \phibf_i|^2}{1}=2$ and
$\sum_{i^\prime \ne i}|\hbf_{g_k}^H \phibf_{i^\prime}|^2=6$. (Note
that 3GPP LTE-Advanced supports 4 $\times$ 4 or 8 $\times$ 8
MU-MIMO.) Then, it is easy to see that the term '1' in the
denominator of the RHS of \eqref{eq:RBFsinr} is negligible, and
the SINR becomes
\begin{equation}  \label{eq:RBFapproxSINR}
{\mathrm{SINR}}_{g_k,i} \approx \frac{|\hbf_{g_k}^H \phibf_i|^2}{\sum_{i^\prime \ne i}|\hbf_{g_k}^H \phibf_{i^\prime}|^2} =\frac{|\hbf_{g_k}^H \phibf_i|^2/||\hbf_{g_k}||^2}{\sum_{i^\prime \ne i}|\hbf_{g_k}^H \phibf_{i^\prime}|^2/||\hbf_{g_k}||^2}.
\end{equation}
A key point to observe here in  \eqref{eq:RBFapproxSINR} is that
the SINR is almost {\em independent of the user channel vector
magnitude!} Thus, the user whose channel vector has the {\em
smallest angle} with the given beam direction regardless of its
channel magnitude is selected for the given beam direction. (As
SNR and the number of streams increase, this effect becomes more
evident. Operating SNR of real cellular systems for data (or
packet) transmission requiring user scheduling is higher than 3
dB. Note also that when we have only one beam direction, the problem does not occur but in this case, there is no spatial multiplexing.) Of course, this user selection is optimal, when $\Wbf_g=\Ibf$
and thus $\sum_{i=1}^{r_g^*} d_{\kappa_i}\phibf_i$ is indeed the
transmitted signal. Now, one can see that RBF does not take the
channel vector magnitude into account for user selection and
furthermore this is because there is no {\em post-user-selection
beam refinement or adjustment}. In RBF, only orthogonality among
the selected users is pursued with neglecting the channel
magnitude. Compare this with SUS-ZFBF. In SUS-ZFBF, the user with
the maximum channel vector magnitude is selected at each inclusion
step while semi-orthogonality among the selected users' channels is
maintained. In the next subsection, we propose a user-selection-and-beamforming
method that corrects the disadvantages of RBF and maintains the
advantages of SUS-ZFBF {\em without} full CSI at the BS.

\subsection{The proposed user selection method}

Here, we consider the original two-stage beamforming setting
again. (The proposed method can readily be applied to conventional
single-stage MU-MIMO downlink too.) As in RBF, we use a set of
{\em orthogonal  reference} beam directions, and use
$\ubf_{g,1},\ubf_{g,2},\cdots,\ubf_{g,r_g^*}$ as the orthogonal
reference beam directions here. Then, as in SUS-ZFBF, we enforce
semi-orthogonality among the selected users by constructing a {\em
double cone} $\Cc_{g,i}$ around each reference beam direction $i$,
as shown in Fig. \ref{fig:aus_region}, defined as
\begin{equation}
\Cc_{g,i} = \left\{ \hbf_{g_k} :  \frac{|\hbf_{g_k}^H\ubf_{g,i}|}{||\hbf_{g_k}||} \ge \alpha^\prime \right\}, ~~i=1,2,\cdots, r_g^*,
\end{equation}
and by checking if the user channel vector $\hbf_{g_k}$ is
contained in $\Cc_{g,i}$ for each $i$. (From here on, we will
refer to double cone simply as cone.) Note that this checking is
done at UTs not at the BS.  To construct $\Cc_{g,i}$, we need the
original channel vector $\hbf_{g_k}$ for user $g_k$. However, we
are assuming that only the equivalent channel state information
$\gbf_{g_k}$ is available at user $g_k$ for the two-stage
beamforming.  Note that from \eqref{eq:equivalentCh}, we have
\begin{equation}
\gbf_{g_k}^H=\hbf_{g_k}^H \Vbf_g= \hbf_{g_k}^H \Ubf_g^*=[\hbf_{g_k}^H
\ubf_{g,1}, \cdots, \hbf_{g_k}^H \ubf_{g,r_g^*}].
\end{equation}
Hence, the cone-containment checking can be done simply by
computing $\frac{\gbf_{g_k}^H}{||\gbf_{g_k}||}$ and checking if
the absolute value of each of its elements is larger than or equal
to a new threshold $\alpha \defeq \alpha^\prime
\frac{||\hbf_{g_k}||}{||\gbf_{g_k}||}$, since
\begin{equation}
\frac{\gbf_{g_k}^H}{||\gbf_{g_k}||}=\frac{||\hbf_{g_k}||}{||\gbf_{g_k}||}\left[\frac{\hbf_{g_k}^H
\ubf_{g,1}}{||\hbf_{g_k}||}, \cdots, \frac{\hbf_{g_k}^H \ubf_{g,r_g^*}}{||\hbf_{g_k}||}\right].
\end{equation}
Note that $0 \le \alpha \le 1$ since each element of the
normalized vector ${\gbf_{g_k}^H}/{||\gbf_{g_k}||}$ is compared
with $\alpha$, and $\alpha$ is a system design parameter that
controls\footnote{Controlling $\alpha$ plays the same role as
controlling the thickness of the user-selection hyperslab in
SUS-ZFBF shown in the previous subsection.} the semi-orthogonality
of the selected user channels.
  If the
channel $\hbf_{g_k}$ of user $g_k$ is contained in
\begin{figure}
\begin{psfrags}
\small
                \psfrag{u1}{$\ubf_{g,1}$}
                \psfrag{u2}{$\ubf_{g,2}$}
                \psfrag{u3}{$\ubf_{g,r_g^*}$}
                \psfrag{h1}{$\hbf_{g_i}$}
                \psfrag{h2}{$\hbf_{g_j}$}
                \psfrag{h3}{$\hbf_{g_k}$}
                \psfrag{h4}{$\hbf_{g_m}$}
                \psfrag{th}{$\theta$}
                \centerline{ \scalefig{0.55} \epsfbox{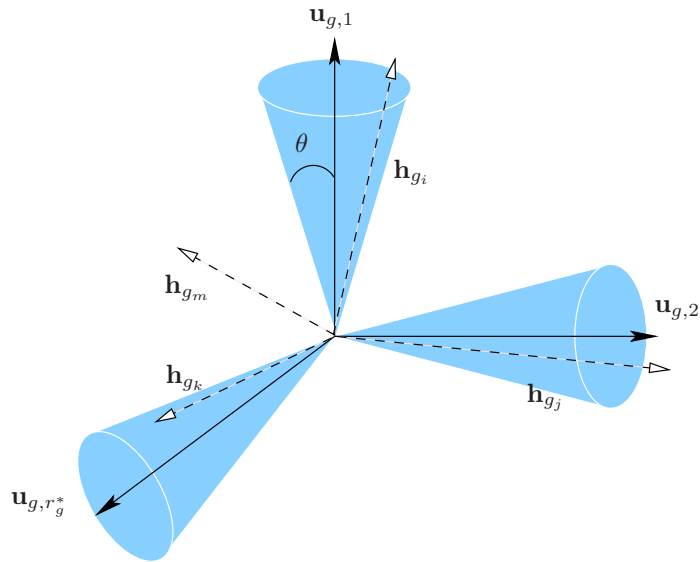} }
    \caption{User selection double cones (the other half of each double cone is not shown)}
    \label{fig:aus_region}
\end{psfrags}
\end{figure}
cone $i$, user $g_k$'s channel is well aligned with the reference
direction $i$ and user $g_k$ belongs to the $i$-th candidate set.
Now, we face the question ``which user in the candidate set $i$
should be selected?''  Since the semi-orthogonality of users to be
selected is already guaranteed by the user-selection cones, we
should choose the user in the candidate set $i$ that has the
maximum channel vector magnitude. This subsequently answers what
should be the CQI that should be feedbacked.  In the two-stage
beamforming setting with the assumption of the availability of the
effective channel $\gbf_{g_k}$, we just use $||\gbf_{g_k}||^2$
since
\begin{eqnarray}
||\gbf_{g_k}||^2&=&||\hbf_{g_k}^H\Ubf_g^*
||^2=||(\sum_{j=1}^{r_g}c_{g_k}^j\ubf_{g,j})^H\Ubf_g^*||^2= \sum_{j=1}^{r_g^*} |c_{g_k}^j|^2 \nonumber \\
&\stackrel{(a)}{\approx}& \sum_{j=1}^{r_g} |c_{g_k}^j|^2 = ||\hbf_{g_k}||^2,
\end{eqnarray}
where $\hbf_{g_k}=\sum_{j=1}^{r_g}c_{g_k}^j\ubf_{g,j}$ with complex linear combination coefficients $c_{g_k,j}$ by \eqref{eq:channel_linear_comb}, and  step (a) is valid because the most dominant $r_g^*$ eigenvectors are included.
  In this way, we can select a set of semi-orthogonal users with
large channel magnitude. Once the user selection is done, we do
not use $\sum_{i=1}^{r_g^*} d_{\kappa_i}\phibf_i$ as the transmit
signal as in RBF, but apply ZFBF with water-filling power allocation based on the effective CSI
obtained from the selected users. In general, the performance of
ZFBF degrades due to noise enhancement in the inversion process
and this degradation appears as the effective channel gain loss.
However, this effective channel gain loss is managed by the
semi-orthogonality of the selected users controlled by the
parameter $\alpha$, as we shall see in Section \ref{sec:AUS_opt}.
This post-user-selection beam refinement requires additional
effective CSI feedback only from the (a few) selected users.

For further improvement in the two-stage beamforming setting, we
can take the inter-group interference into consideration. When
$\Vbf_g = \Ubf_g^*$ for all $g$ and equal power $\rho
=\frac{P}{\sum_{g=1}^G r_g^*}$ for every scheduled user,   the norm of every column of the second-stage beamformer
$\Wbf_g$ is $\rho$ from \eqref{eq:power_formula}. From
\eqref{eq:rec_gk}, the average power of the inter-group
interference-plus-noise is upper bounded by
\begin{equation} \label{eq:noiseInterGintUB}
1+r_g^* \rho\sum_{g' \neq
g}\|\hbf_{g_k}^H\Vbf_{g'}\|^2
\end{equation}
 by norm's submultiplicativity and $\|\cdot\| \le \|\cdot\|_F$, where $\|\cdot\|_F$ is the Frobenius norm. ($\|\Wbf_{g'}\|_F = r_g^* \rho$.)
We define a {\em quasi-SINR} as
\begin{equation}  \label{eq:quasiSINRdef}
\Rc(g_k) := \frac{\|\gbf_{g_k}\|^2}{ \frac{1}{\rho}+r_g^* \sum_{g' \neq
g} \|\hbf_{g_k}^H\Vbf_{g'}\|^2}.
\end{equation}
In the definition, the intra-group interference does not exist
because ZFBF will be used in later post-user-selection beamforming. Without the
inter-group interference, the quasi-SINR is simply a scaled
version of the square of the effective channel vector magnitude.

\begin{remark}
The reason for the chosen definition of the quasi-SINR will become
clear in Section \ref{sec:AUS_opt}. This metric guarantees the
asymptotic optimality of the proposed method under the assumption
of the approximated BD in Condition \ref{cond:approxBD} for inter-group interference.
\end{remark}

\begin{remark} The effective CSI and the average inter-group-interference-plus-noise power can easily be estimated at UTs
during the downlink training period.   The received signal model
\eqref{eq:rec_gk} at user $g_k$ can be rewritten by combining all intra-group signals as
\begin{align} \label{eq:channelnoiseest}
y_{g_k} &= \underbrace{\hbf_{g_k}^H \Vbf_g}_{\gbf_{g_k}^H} \Wbf_g \dbf_g +
\sum_{g' \neq g} \hbf_{g_k}^H \Vbf_{g'} \Wbf_{g'}\dbf_{g'} + n_{g_k}.
\end{align}
First, the effective CSI $\gbf_{g_k}$ can easily be estimated at
UTs  during the downlink training period. Please see
\cite{Noh&Zoltowski&Sung&Love:14JSTSP} for this. Furthermore,
during the downlink training period,  the average inter-group
interference-plus-noise power  $1+\rho\sum_{g' \neq
g}\|\hbf_{g_k}^H\Vbf_{g'}\|^2$ can also be estimated easily based
on \eqref{eq:channelnoiseest}. That is, the training $\Wbf_g$ and
$\dbf_g$ are known to all UTs in group $g$. Once $\gbf_{g_k}$ is
estimated at user $g_k$, user $g_k$ constructs $\gbf_{g_k}^H\Wbf_g
\dbf_g $, computes $y_{g_k}-\gbf_{g_k}^H\Wbf_g \dbf_g = \sum_{g'
\neq g} \hbf_{g_k}^H \Vbf_{g'} \Wbf_{g'}\dbf_{g'} + n_{g_k}$,
squares $y_{g_k}-\gbf_{g_k}^H\Wbf_g \dbf_g$, and averages the
result over a few training symbol times to obtain the desired
value.  If the training $\Wbf_{g'}$ and the actual data-transmitting
$\Wbf_{g'}$ have similar norm, the estimated average inter-group
interference-plus-noise power will be valid for the
data-transmission period.  Thus, UTs can easily compute the
proposed quasi-SINR during the training period.
\end{remark}

\vspace{0.5em}

Now, we  present our proposed user-scheduling-and-beamforming
algorithm named {\em `REference-based Distributed
(semi-)Orthogonal user Selection with Post-selection Beam
Refinement (ReDOS-PBR)'}:

\vspace{0.5em}

\begin{algorithm}[The Proposed User-Scheduling-And-Beamforming Method: ReDOS-PBR] \label{algo:ReDOS}

~
\begin{enumerate}
\item[ 0)] $\alpha \in (0,~1)$ is a pre-determined  parameter and
is shared by the BS and all UTs. The BS initializes
\begin{align}
\Wc_{g,i} &= \emptyset, ~\text{for}~ i=1,\cdots,r_g^*  \label{eq:algorithm_Wgi} \\
\Sc_g &= \emptyset.
\end{align}
Every user $g_k$ estimates $\gbf_{g_k}$ and $1+\rho\sum_{g' \neq
g}\|\hbf_{g_k}^H\Vbf_{g'}\|^2$.

\item Each user $g_k$ independently computes the following set:
\begin{equation} \label{eq:align}
\Ic_{g_k}:= \left\{i: \left|(\ebf_{i}^{(g)})^T
\frac{\gbf_{g_k}}{\|\gbf_{g_k}\|}\right| \ge \alpha,
~i=1,\cdots,r_g^*\right\},
\end{equation}
where $\ebf_{i}^{(g)}$ is the $i$-th column of $\Ibf_{r_g^*}$.
($(\ebf_{i}^{(g)})^T \frac{\gbf_{g_k}}{\|\gbf_{g_k}\|}$ is simply
the $i$-th element of $\frac{\gbf_{g_k}}{\|\gbf_{g_k}\|}$.)

If user $g_k$ has $\Ic_{g_k} \neq \emptyset$, then user $g_k$
finds
\begin{equation} \label{eq:most_align}
i^*_{g_k}=\underset{i \in \Ic_{g_k}}{\arg\max}
\left|(\ebf_{i}^{(g)})^T\frac{\gbf_{g_k}}{\|\gbf_{g_k}\|}\right|
\end{equation}
and feedbacks the CQI pair $(i^*_{g_k},\Rc(g_k))$ to the BS. If
$\Ic_{g_k} = \emptyset$,
 user $g_k$ does not feedback.

After the feedback, the BS updates $\Wc_{g,i^*_{g_k}} \leftarrow
\Wc_{g, i^*_{g_k}} \cup \{k\}$ and stores $\Rc(g_k)$. %
\item For $i=1,\cdots,r_g^*$, the BS finds
\begin{equation} \label{eq:max_sinr}
\kappa_{g,i} = \underset{k \in \Wc_{g,i}}{\arg\max} ~ \Rc(g_k),
\end{equation}
and updates
\begin{equation}
\Sc_g \leftarrow \Sc_g \cup \{\kappa_{g,i}\}.
\end{equation}
 \item The BS transmits a paging
signal to notify that the users in $\Sc_g$ are scheduled and then,
only the corresponding scheduled UTs feedback their effective CSI to the BS.
Finally, the BS constructs the MU-MIMO ZFBF precoder with water-filling power allocation for each
group based on the signal model \eqref{eq:RXsigybfg} and the acquired
effective CSI from the scheduled users, and transmits data  to the
scheduled UTs.
\end{enumerate}
\end{algorithm}

\vspace{0.5em}

In step 1), each user checks if its channel vector is contained in
each of the user-selection cones. If  user $g_k$ has a non-empty
set $\Ic_{g_k}$, then user $g_k$ finds the reference direction
that has the largest channel component and feedbacks the
corresponding reference direction index $i^*_{g_k}$ and the
quasi-SINR $\Rc(g_k)$ to the BS. If $\Ic_{g_k}=\emptyset$, then
user $g_k$ does not feedback any information to the BS. After the
feedback period is over, the BS makes $r_g^*$ candidate sets
$\Wc_{g,1},\cdots,\Wc_{g,r_g^*}$ for the $r_g^*$ reference directions for group
$g$, based on the CQI feedback information. Here,
 $\Wc_{g,i}$ represents the set of users whose channels
are contained in the user-selection cone around the $i$-th
reference direction. In step 2), the BS chooses the user
$\kappa_{g,i}$ having the largest quasi-SINR $\Rc(g_k)$ in each
set $\Wc_{g,i}$, $i=1,\cdots,r_g^*$, to construct the set $\Sc_g$
of scheduled users for each group $g$. In step 3), ZFBF is used
for the scheduled users. Here, more sophisticated MU-MIMO BF like
MMSE BF can also be used for the post-user-selection beam
refinement to yield better performance, if additional inter-group
interference and noise variance information is available at the BS
for the signal model \eqref{eq:RXsigybfg}. In the case of such
advanced post-user-selection beam refinement,
$\{\Wbf_g,g=1,\cdots,G\}$ should be designed jointly. However, since the semi-orthogonality among the selected users for each group and the approximated BD condition for inter-group interference are satisfied, ZFBF should be sufficient.

\vspace{0.5em}

\begin{remark}[Amount of feedback] First note that in ReDOS-PBR, {\em  user selection is done based on only CQI
feedback from possibly all users}
 and {\em post-user-selection beam refinement is done based on the CSI feedback from only the scheduled users}.
 The feedback difference in CQI and CSI is significant in  MIMO
 systems.
 The amount of feedback
required for the proposed method for group $g$ for one scheduling
interval is $\sum_{i=1}^{r_g^*}|\Wc_{g,i}|$ integers for user beam
index feedback, $\sum_{i=1}^{r_g^*}|\Wc_{g,i}|$ real numbers for
quasi-SINR feedback, and $2(r_g^*)^2$ real numbers for later
effective CSI feedback  because only $r_g^*$ users per group need
to feedback their effective CSI $\gbf_{\kappa_{g,i}}$ of complex
dimension $r_g^*$  for $\Vbf_g=\Ubf_g^*$.  As shown in Lemma
\ref{lemma:alpha_min} in the below, when $\alpha \le
1/\sqrt{r_g^*}$, $\Ic_{g_k}$ is a non-empty set for all $g_k$ and
thus, every user feedbacks its quasi-SINR to the BS. Hence, in
this case, $\sum_{i=1}^{r_g^*}|\Wc_{g,i}|$ reduces to $K_g$. When
$\alpha > 1/\sqrt{r_g^*}$, on the other hand,
$\Ic_{g_k}=\emptyset$ for some users and thus in this case,
$\sum_{i=1}^{r_g^*}|\Wc_{g,i}|$ can be less than $K_g$. In Section
\ref{sec:num_res}, numerical results show that many users do not
feedback even CQI to the BS for optimally chosen $\alpha$ and the
feedback overhead is reduced drastically.
\end{remark}

\vspace{0.5em}

\begin{remark}[Feedback structure and delay]
ReDOS-PBR requires the above-mentioned two-step feedback: The CQI
feedback phase and the CSI feedback phase. In practical cellular
systems, time is segmented into contiguous radio frames and each
radio frame is one scheduling interval.  If both feedback phases
can be finished within one data transmission radio frame by using some control channel, there is
no additional delay in feedback.
\end{remark}

\begin{lemma}  \label{lemma:alpha_min}
When $\alpha \le \alpha_{min}:=1/\sqrt{r_g^*}$, ~$\Ic_{g_k}$ is a non-empty set for all $g_k$.
\end{lemma}
\begin{proof}
Suppose that there is a user $g_k$ such that $\Ic_{g_k} = \emptyset$.
Then,
$\left|(\ebf_i^{(g)})^T\frac{\gbf_{g_k}}{\|\gbf_{g_k}\|} \right| < 1/\sqrt{r_g^*}~$
for all $i=1,\cdots,r_g^*$.
Therefore, we have
\begin{equation}
1 = \left\|\frac{\gbf_{g_k}}{\|\gbf_{g_k}\|}\right\|^2
= \sum_{i=1}^{r_g^*}\left|(\ebf_i^{(g)})^T \frac{\gbf_{g_k}}{\|\gbf_{g_k}\|}\right|^2
<1,
\end{equation}
and have contradiction. Hence, the claim follows.
\end{proof}

\begin{remark} \label{remark:full_disjoint_cones}
Lemma \ref{lemma:alpha_min} implies that $ \bigcup_{i=1}^{r_g^*} \Cc_{g,i} \supset {\mathbb{C}}^{r_g*}$ for $\alpha \le \alpha_{min}=1/\sqrt{r_g^*}$.  On the other hand, when $\alpha > \frac{1}{\sqrt{2}}$, $\Cc_{g,i} \cap \Cc_{g,j} = \emptyset $ for $i \ne j$, because the angle $\theta$ in Fig.  \ref{fig:aus_region} is $\pi/4$ when $\alpha=1/\sqrt{2}$.
\end{remark}

Lemma \ref{lemma:alpha_min} and $\alpha_{min}$  will be useful in
Section \ref{sec:fairness}.

\section{Optimality of The Proposed Method}
\label{sec:AUS_opt}

In this section, we prove the asymptotical optimality of the proposed method
 as $K \to \infty$. We begin with the optimal sum capacity scaling law of
a $K$-user MIMO broadcast channel consisting of multiple
groups with each group's having the same channel covariance
matrix, provided in \cite{Adhikary&Caire:13arXiv}.

\vspace{0.5em}
\begin{theorem}\cite{Adhikary&Caire:13arXiv} \label{thm:DPC}
In a MU-MIMO downlink system composed of a BS with $M$ transmit antennas
and total power constraint $P$ and $K$ users each with a single receive
antenna divided into $G$ groups of equal size $K'=K/G=K_g$, where
the channel vector of each user in group $g$ is independent and
identically distributed (i.i.d.) from $\Cc\Nc({\bf 0}, \Rbf_g)$
for $g=1,\cdots,G$, the sum capacity (which is achieved by DPC)
scales as
\begin{equation}  \label{eq:R_DPC}
R_{DPC} = \beta \log\log (K') + \beta \log \frac{P}{\beta} + O(1)
\end{equation}
where $\beta = \min\{M, \sum_{g=1}^G r_g\}$ and $O(1)$ denotes a bounded constant
independent of $K'$, as $K' \to \infty$.
\end{theorem}
\begin{proof}
See Theorem 1 in \cite{Adhikary&Caire:13arXiv}.
\end{proof}

\vspace{0.5em}

 The same scaling law is achieved by
ReDOS-PBR under the approximate BD condition in Condition  \ref{cond:approxBD}.

\vspace{0.5em}

\begin{theorem} \label{theo:OptimalityProof}
In the system described in Theorem \ref{thm:DPC}, the sum rate of
the scheduled sets $\{\Sc_g\}$ by ReDOS-PBR scales as
\begin{equation} \label{eq:optimality}
\mathbb{E}\left[\sum_{g=1}^GR_{ZF,g}(\Sc_g)\right] \sim R_{DPC},
\end{equation}
where $x \sim y$ indicates that $\underset{K' \to
\infty}{\lim}x/y=1$. Here, $R_{ZF,g}(\Sc_g)$ is the sum rate of
the users in $\Sc_g$ determined by the proposed user-selection method with ZFBF
second-stage precoding.
\end{theorem}

\vspace{0.5em}

\begin{proof}
Similarly to the asymptotic optimality proof of SUS-ZFBF in
\cite{Yoo&Goldsmith:06JSAC}, our proof of the asymptotic
optimality of ReDOS-PBR is by showing first that the effective channel gain associated with ReDOS-PBR is bounded below away from zero for some {\em fixed}  $\alpha$ strictly less than one and then showing that the multi-user diversity
gain reduction associated with ReDOS-PBR for that fixed $\alpha$ become negligible as  $K'
\to \infty$.\footnote{We borrowed the flow of our proof  from
\cite{Yoo&Goldsmith:06JSAC}. However, different techniques and
ideas are used for our proof of the asymptotic optimality of
ReDOS-PBR.}

From  \eqref{eq:RXsigybfg}, we have the received signal model for the scheduled users in $\Sc_g$ as
\begin{equation}  \label{eq:RXsigybfgSch}
\ybf_g(\Sc_g) = \Gbf_g(\Sc_g)\Wbf_g(\Sc_g)\dbf_g (\Sc_g)
+ \sum_{g'\neq g} \Hbf_g (\Sc_g) \Vbf_{g'} \Wbf_{g'} \dbf_{g'}
+ \nbf_g(\Sc_g),
\end{equation}
where $\Wbf_g(\Sc_g)=[\{\wbf_{g_k}\}_{k \in
\Sc_g}]=[\wbf_{\kappa_{g,1}},\cdots,\wbf_{\kappa_{g,r_g^*}}]$ and
$\Gbf_g(\Sc_g)=[\{\gbf_{g_k}\}_{k \in
\Sc_g}]^H=[\gbf_{\kappa_{g,1}},\cdots,$
$\gbf_{\kappa_{g,r_g^*}}]^H$ are respectively the submatrices of
$\Wbf_g$ and $\Gbf_g$ corresponding to the users in $\Sc_g$
obtained by ReDOS-PBR.

\textit{i) Lower bound on the effective channel gain:}
Since ZFBF is assumed for the second-stage beamforming with the signal model \eqref{eq:RXsigybfgSch}, we have
\begin{align}
\Wbf_g &:=\Wbf_g(\Sc_g) = [\wbf_{\kappa_{g,1}},\cdots,\wbf_{\kappa_{g,r_g^*}}] \nonumber\\
&=\Gbf_g^H(\Sc_g)\left[\Gbf_g(\Sc_g)\Gbf_g^H(\Sc_g)\right]^{-1}\Pbf_g\nonumber\\
&=:\tilde{\Wbf}_g\Pbf_g=[\tilde{\wbf}_{\kappa_{g,1}},\cdots,\tilde{\wbf}_{\kappa_{g,r_g^*}}]\Pbf_g,
\end{align}
where  $\Pbf_g=\text{diag}(\sqrt{P_{\kappa_{g,1}}},
\cdots,\sqrt{P_{\kappa_{g,r_g^*}}})$, and  $P_{\kappa_{g,i}}$ is
the transmit power scaling factor\footnote{Since the
pseudo-inverse
$\Gbf_g^H(\Sc_g)[\Gbf_g(\Sc_g)\Gbf_g^H(\Sc_g)]^{-1}$ is fixed for
the given set of the scheduled users' effective channel vectors,
we need $\Pbf_g$ to control the user power.} for the scheduled
user $\kappa_{g,i} \in \Sc_g$.  Substituting the above ZF
$\wbf_{\kappa_{g,1}},\cdots,\wbf_{\kappa_{g,r_g^*}}$ into the
received signal model \eqref{eq:rec_gk} of user $\kappa_{g,i}$ yields
\begin{align}\label{eq:zfbf_resulting_ygk}
y_{\kappa_{g,i}} &= \sqrt{P_{\kappa_{g,i}}} d_{\kappa_{g,i}} +
\sum_{g' \neq g} \hbf_{g_k}^H \Vbf_{g'} \Wbf_{g'}\dbf_{g'} +
n_{\kappa_{g,i}}, ~~~i=1,\cdots,r_g^*,
\end{align}
since $\Gbf_g(\Sc_g)\Wbf_g=\Pbf_g=\text{diag}(\sqrt{P_{\kappa_{g,1}}},
\cdots,\sqrt{P_{\kappa_{g,r_g^*}}})$. From
\eqref{eq:zfbf_resulting_ygk}, the sum rate of the ZF MU-MIMO
broadcast channel consisting of users
$\{\kappa_{g,1},\cdots,\kappa_{g,r_g^*}\}$ with power scaling
$\{P_{\kappa_{g,1}},\cdots,P_{\kappa_{g,r_g^*}}\}$ is given by
\cite{Cover&Thomas:book}
\begin{align}
R_{ZF,g}(\Sc_g) =& \underset{\{P_{\kappa_{g,i}}\}}{\max}\sum_{i=1}^{r_g^*}
\log\left(1+\frac{P_{\kappa_{g,i}}}{1 + \sum_{g'\neq g}
\|\hbf_{g_{\kappa_{g,i}}}^H\Vbf_{g'}\Wbf_{g'}\|^2}\right) \nonumber\\
& \text{s.t.} ~~\sum_{i=1}^{r_g^*} \gamma_{\kappa_{g,i}}^{-1}
P_{\kappa_{g,i}}
\le r_g^*\rho,  \label{eq:zf_sumrate}
\end{align}
where $r_g^*  \rho$ is the total transmit power\footnote{ We
assume that the total transmit power assigned to group $g$ is
proportional to the number of the scheduled users in group $g$,
and hence, it is  $r_g^*  \rho$. } assigned to group $g$;  by
\eqref{eq:power_formula} the actual power assigned to user
$\kappa_{g,i}$ is given by
$P_{\kappa_{g,i}}^{actual}=\|\wbf_{\kappa_{g,i}}\|^2 =
\gamma_{\kappa_{g,i}}^{-1}P_{\kappa_{g,i}}$; and the effective
channel gain\footnote{Note in the constraint \eqref{eq:zf_sumrate}
that the ZF loss appears as the shrinkage of the feasible region
of $(P_{\kappa_{g,1}},\cdots,P_{\kappa_{g,1}})$.  If
$\gbf_{\kappa_{g,1}},\cdots,$ $\gbf_{\kappa_{g,r_g^*}}$ are
perfectly orthogonal, then
$\gamma_{{\kappa_{g,i}}}=||\gbf_{\kappa_{g,i}}||^2$ and there is
no ZF loss.} $\gamma_{{\kappa_{g,i}}}$ for user $\kappa_{g,i}$ is
given by \cite{Yoo&Goldsmith:06JSAC,Dimic&Sidiropoulos:05SP,Peel&Hochwald&Swindlehurst:05COM}
\begin{equation}\label{eq:eff_gain}
\gamma_{{\kappa_{g,i}}} =
\frac{1}{[(\Gbf_g(\Sc_g)\Gbf_g(\Sc_g)^H)^{-1}]_{i,i}}.
\end{equation}
This is because
$\|\wbf_{\kappa_{g,i}}\|^2=\|\tilde{\wbf}_{\kappa_{g,i}}\|^2P_{\kappa_{g,i}}$
and $\|\tilde{\wbf}_{\kappa_{g,i}}\|^2=[
\tilde{\Wbf}_g^H\tilde{\Wbf}_g]_{i,i}=[(\Gbf_g(\Sc_g)\Gbf_g(\Sc_g)^H)^{-1}]_{i,i}$.

Now consider the denominator term in the RHS of
\eqref{eq:eff_gain}. Since $[\Gbf_g(\Sc_g)\Gbf_g(\Sc_g)^H]_{i,j} =
\gbf_{\kappa_{g,i}}^H\gbf_{\kappa_{g,j}}$, $\forall i,j$, it can be decomposed
 as
\begin{equation}\label{eq:decom}
\Gbf_g(\Sc_g)\Gbf_g(\Sc_g)^H= \Dbf\tilde{\Gbf}\Dbf,
\end{equation} where
$\Dbf=\text{diag}(\|\gbf_{\kappa_{g,1}}\|,\cdots,\|\gbf_{\kappa_{g,r_g^*}}\|)$
and
\begin{equation}
\tilde{\Gbf}= \left[
\begin{matrix}
{1} & {\tilde{\gbf}_{\kappa_{g,1}}^H\tilde{\gbf}_{\kappa_{g,2}}} & {\cdots} & {\tilde{\gbf}_{\kappa_{g,1}}^H\tilde{\gbf}_{\kappa_{g,r_g^*}}} \\
{\tilde{\gbf}_{\kappa_{g,2}}^H\tilde{\gbf}_{\kappa_{g,1}}} & {1} & {} & {\vdots} \\
{\vdots} & {} & {\ddots} & {\tilde{\gbf}_{\kappa_{g,r_g^*-1}}}^H\tilde{\gbf}_{\kappa_{g,r_g^*}} \\
\tilde{\gbf}_{\kappa_{g,r_g^*}}^H\tilde{\gbf}_{\kappa_{g,1}} & \cdots & \tilde{\gbf}_{\kappa_{g,r_g^*}}^H\tilde{\gbf}_{\kappa_{g,r_g^*-1}} & {1}
\end{matrix}
\right]
\end{equation}
with $\tilde{\gbf}_{\kappa_{g,i}}=\frac{\gbf_{\kappa_{g,i}}}{\|\gbf_{\kappa_{g,i}}\|}, ~\forall i$. Substituting \eqref{eq:decom} into \eqref{eq:eff_gain}, we have
\begin{align}
\gamma_{\kappa_{g,i}} &=
\frac{1}{[(\Gbf_g(\Sc_g)\Gbf_g(\Sc_g)^H)^{-1}]_{i,i}} \nonumber \\
&= \frac{1}{[\Dbf^{-1}\tilde{\Gbf}^{-1}\Dbf^{-1}]_{i,i}} \nonumber \\
&= \frac{\|\gbf_{\kappa_{g,i}}\|^2}{[\tilde{\Gbf}^{-1}]_{i,i}}. \label{eq:gamma_kappa_gi}
\end{align}
Consider the term $[\tilde{\Gbf}^{-1}]_{i,i}$ in \eqref{eq:gamma_kappa_gi}.  By Lemma \ref{lemma:ConeDistance} in Appendix \ref{appen:cone_inner_prod}, we have
\begin{equation}\label{eq:inner}
|\tilde{\gbf}_{\kappa_{g,i}}^H \tilde{\gbf}_{\kappa_{g,j}}|
\le 2 \alpha \sqrt{1-\alpha^2}~~ \text{for} ~~ i \neq j
\end{equation}
when
$\alpha \ge 1/\sqrt{2}$.
 By the Gershgorin
circle theorem \cite{Horn&Johnson:book} and \eqref{eq:inner},
every eigenvalue of the Hermitian matrix $\tilde{\Gbf}$ is in a
Gershgorin disk,\footnote{All Gershgorin disks of $\tilde{\Gbf}$
have the same center of one and the same radius upper bound. So,
we can use any of the Gershgorin disks of $\tilde{\Gbf}$.}  i.e.,
\begin{align}
\lambda(\tilde{\Gbf}) &\in \{z \in \mathbb{R}^+:
|z - 1|\le (r_g^* - 1)2\alpha\sqrt{1-\alpha^2}\}, \nonumber\\
&=\{ z \in \mathbb{R}^+:  1 - (r_g^* - 1)2\alpha\sqrt{1-\alpha^2} \le z \le 1+(r_g^* - 1)2\alpha\sqrt{1-\alpha^2} \}\label{eq:g_gorin}
\end{align}
where $\lambda(\tilde{\Gbf})$ is the set of eigenvalues of $\tilde{\Gbf}$. When $(r_g^*-1)2\alpha\sqrt{1-\alpha^2}<1$, equivalently,
\begin{equation}\label{eq:alphagershlower}
\alpha > \sqrt{\frac{1+\sqrt{\frac{r_g^*-2}{r_g^*-1}}}{2}},
\end{equation}
 we have a non-trivial lower bound on $\lambda_{min}(\tilde{\Gbf})$ and
\begin{equation} \label{eq:Gbfiiinv_lower_bound}
[\tilde{\Gbf}^{-1}]_{i,i} \le [\lambda_{min}(\tilde{\Gbf})]^{-1}
\overset{(a)}{\le} \frac{1}{1-(r_g^*-1)2\alpha\sqrt{1-\alpha^2}},
\end{equation}
since $\tilde{\Gbf}$ is self-adjoint and  (a) follows from \eqref{eq:g_gorin}, where $\lambda_{min}(\tilde{\Gbf})$ is the minimum eigenvalue of $\tilde{\Gbf}$.
Thus, from \eqref{eq:gamma_kappa_gi} and \eqref{eq:Gbfiiinv_lower_bound}, the effective channel gain $\gamma_{\kappa_{g,i}}$ is lower bounded by
\begin{equation} \label{eq:eff_gain_low}
\gamma_{\kappa_{g,i}} \ge \frac{\|\gbf_{\kappa_{g,i}}\|^2}
{\frac{1}{1-(r_g^*-1)2\alpha\sqrt{1-\alpha^2}}}.
\end{equation}
Note that the derived lower bound  \eqref{eq:eff_gain_low} on the effective channel gain is valid for any {\em fixed} $\alpha$ satisfying
\begin{equation}
\alpha ~>~ \sqrt{\frac{1+\sqrt{\frac{r_g^*-2}{r_g^*-1}}}{2}} ~\stackrel{(a)}{\ge}~ \frac{1}{\sqrt{2}},
\end{equation}
where $(a)$ for the validity of \eqref{eq:inner} is valid for any $r_g^* \ge 2$. By making $\alpha \uparrow 1$, we can completely eliminate the ZFBF loss. However, $\alpha \uparrow 1$ will lose the multiuser diversity gain. So, we fix $\alpha$ to an arbitrary number $\bar{\alpha}$ strictly less than one, independent of $K'$ such that
\begin{equation}  \label{eq:baralpha_range}
\bar{\alpha} \in \left(\sqrt{\frac{1+\sqrt{\frac{r_g^*-2}{r_g^*-1}}}{2}}, ~~~1 \right).
\end{equation}

\textit{ii) Multi-user diversity gain:} There are several
difficult points in handling the multi-user diversity gain of
ReDOS-PBR with the multi-group setting of JSDM.  The first point
is that only users whose channel vectors are contained in one of
the user-selection cones report quasi-SINR and  the second point
is that we should handle the inter-group interference properly.
Despite such difficulty we were able to show that the multi-user
diversity gain is still preserved for ReDOS-PBR under the
approximate BD condition. The main insight is that with fixed
$\bar{\alpha}$ in \eqref{eq:baralpha_range} strictly less than
one, independent of $K'$, the number of users whose channel
vectors are contained in each user-selection cone tends to
infinity as $K' \rightarrow \infty$ since each user-selection cone
occupies certain fixed non-trivial measure (or volume) in
${\mathbb{C}}^{r_g^*}$.

As in \cite{Yoo&Goldsmith:06JSAC}, the first difficulty mentioned above can be handled by defining
\begin{equation} \label{eq:phigkdef}
\phi_{g_k}^i= \left\{
\begin{array}{ll}
\Rc(g_k),& k \in \Wc_{g,i}, \\
0,& \text{otherwise}
\end{array}
\right.
\end{equation}
for all users $k=1\cdots,K_g=K'$ in group $g$.  Then, for a given $i$, the random variable $\phi_{g_k}^i$ is i.i.d. across $k$ in the same group $g$ since $\hbf_{g_k} \stackrel{i.i.d.}{\sim}\Cc\Nc({\mathbf{0}},\Rbf_g)$. Note that
\begin{equation}
\kappa_{g,i} = \mathop{\arg \max}_{k\in \Wc_{g,i}} \Rc(g_k)= \mathop{\arg \max}_{k\in \{1,\cdots,K_g=K^\prime\}} ~ \phi_{g_k}^i.
\end{equation}
The multi-user diversity gain results from choosing the best user among all users with i.i.d. channel realizations.
However, with ReDOS-PBR, for each data stream, the best user within $\Wc_{g,i}$ is chosen, and thus there exists some loss in the multi-user diversity gain. However, based on extreme value theory we have that
 for each $i$
\begin{equation}  \label{eq:multiuserdiversitygain}
\text{Pr}\{\phi_{\kappa_{g,i}}^i > u_{g}^i\} \ge 1-O(1/K'),
\end{equation}
for ReDOS-PBR under the approximate BD condition in Condition \ref{cond:approxBD},
where
\begin{equation}  \label{eq:multiDivugi}
u_{g}^i=(\lambda_{g,1}\log K' - \lambda_{g,1}\log \log K'
+ a_i)/(1/\rho+\epsilon);
\end{equation}
$\lambda_{g,1}$ is the maximum eigenvalue of $\Rbf_g$ (see \eqref{eq:model_channel_eigenvalue});
and  $a_i$ and $\epsilon$ are constants independent of $K'$. Proof of (\ref{eq:multiuserdiversitygain}, \ref{eq:multiDivugi}) is in Appendix \ref{sec:Appen-proof} with some prerequisite on extreme value theory in Appendix
\ref{appen:extrem}.

\textit{iii)} Finally, we show the asymptotic optimality \eqref{eq:optimality} of ReDOS-PBR based on {\it i)} and {\it ii)}. Fix $\alpha$ as $\bar{\alpha}$ in \eqref{eq:baralpha_range}. Then, we have
\begin{align}
&\mathbb{E}\left[\sum_{g=1}^G R_{ZF,g}(\Sc_g)\right] \nonumber\\
&\stackrel{(a)}{\ge} {\mathbb E}
\left[\sum_{g=1}^G\sum_{i=1}^{r_g^*}\log
\left(1 + \frac{\rho \gamma_{\kappa_{g,i}}}{1+
\sum_{g'\neq g}\|\hbf_{\kappa_{g,i}}^H\Vbf_{g'}\Wbf_{g'}\|^2}
\right) \right] \nonumber\\
&\overset{(b)}{\ge} \mathbb{E}\left[\sum_{g=1}^G
\sum_{i=1}^{r_g^\star}
\log\left(1+ \frac{ \|\gbf_{\kappa_{g,i}}\|^2
[1-(r_g^\star-1)2\alpha\sqrt{1-\alpha^2}]}{\frac{1}{\rho} +
r_g^* \sum_{g'\neq g}
\|\hbf_{\kappa_{g,i}}^H\Vbf_{g'}\|^2} \right) \right] \nonumber\\
&\overset{(c)}{\ge} \sum_{g=1}^G\sum_{i=1}^{r_g^\star}
\text{Pr}\{\phi_{\kappa_{g,i}}^i > u_g^i\}\log \left(1+u_g^i
[1-(r_g^\star -1)2\alpha\sqrt{1-\alpha^2}] \right) \nonumber\\
&\overset{(d)}{\ge} \sum_{g=1}^G\sum_{i=1}^{r_g^\star}
\left[1- O\left(\frac{1}{K'} \right) \right]\log \left(1+u_g^i
[1-(r_g^\star -1)2\alpha\sqrt{1-\alpha^2}] \right) \nonumber\\
&\stackrel{(e)}{\sim} \sum_{g=1}^G\sum_{i=1}^{r_g^*}\log\left(
1+ \left(\frac{1-(r_g^*-1)2\alpha\sqrt{1-\alpha^2}}
{1/\rho+\epsilon}\right)\lambda_{g,1}\log K'
\right) \label{eq:proof_theo2_e}\\
&\overset{(f)}{\sim}  \sum_{g=1}^G r_g^* \log(1+\rho \lambda_{g,1}\log K')
\label{eq:proof_theo2_f}\\
&\sim \left(\sum_{g=1}^G r_g^* \right) \log \rho
+ \sum_{g=1}^G r_g^* \log \lambda_{g,1}
+ \left(\sum_{g=1}^G r_g^* \right) \log \log K' \label{eq:last_equation}
\end{align}
where (a) follows from the suboptimal equal power allocation $\rho=\frac{P}{\sum_{g=1}^G r_g^*}=
\|\wbf_{\kappa_{g,i}}\|^2  =\gamma_{\kappa_{g,i}}^{-1}P_{\kappa_{g,i}}, ~\forall g,i$; (b) is obtained by
\eqref{eq:eff_gain_low} and \eqref{eq:noiseInterGintUB} valid for $\bar{\alpha}$; (c) holds by the definition \eqref{eq:quasiSINRdef} of quasi-SINR $\Rc(g_k)$ and the definition \eqref{eq:phigkdef} of $\phi^i_{g_k}$, and $\Ebb f(X) =\int_0^\infty f(x) p(x)dx\ge \mbox{Pr}(X \ge u)f(u)$ for a monotone increasing function $f$ (here, $f=\log$); (d) holds by \eqref{eq:multiuserdiversitygain}; (e) follows from $(1-O(1/K')) \sim 1$ and
$u_g^i \sim (\lambda_{g,1} \log K')/(1/\rho + \epsilon)$ from \eqref{eq:multiDivugi};
and (f) follows since the difference between  the two logarithmic terms in \eqref{eq:proof_theo2_e} and \eqref{eq:proof_theo2_f} converges to a constant independent of $K'$, given by
\[
\sum_{g=1}^G r_g^* \log\left(\frac{1+\rho  \epsilon}{1-(r_g^*-1)2\alpha
\sqrt{1-\alpha^2}}\right).
\]
Finally, consider \eqref{eq:last_equation}.
In both cases of
$\sum_{g=1}^G r_g <M$ and $\sum_{g=1}^G r_g \ge M$, we can
choose $r_g^*$ such that $\sum_{g=1}^G r_g^* = \min\{M,
\sum_{g=1}^G r_g\}=\beta$. Then, \eqref{eq:last_equation} is the same as \eqref{eq:R_DPC} since $P/\beta = \rho$.
\end{proof}

\vspace{0.5em}

Note that fixed $\alpha$ in the range of \eqref{eq:baralpha_range} {\em guarantees} the asymptotic optimality of ReDOS-PBR. We do not know whether $\alpha$ outside this range yields asymptotic optimality or not. (This depends on the tightness of the bound given by the Gershgorin circle theorem used in \eqref{eq:g_gorin}.) For proof of asymptotic optimality, the existence of one $\alpha$ value, i.e., $\bar{\alpha}$, is sufficient.  In the practical case of {\em finite} users in the cell, optimal $\alpha$ may be smaller than $\sqrt{\frac{1+\sqrt{\frac{r_g^*-2}{r_g^*-1}}}{2}}$. Numerical results in Section \ref{sec:num_res} shows that the performance of ReDOS-PBR in the finite-user case is quite insensitive to $\alpha$.

\section{Extension}
\label{sec:fairness}

In the previous section, we only discussed  user selection and beamforming
for maximizing the sum rate. Now, consider fairness among users. If the channel statistics are the same across users and the channel realizations are i.i.d. across scheduling intervals, the fairness issue will be resolved automatically  \cite{Viswanath&Tse&Laroia:02IT}.
However,  in slow-fading environments or in practical
downlink systems with different large-scale fading for users at different locations, some scheme should be implemented to impose fairness among users. Among several well-known fairness-imposing schemes
\cite{Viswanath&Tse&Laroia:02IT,Yoo&Goldsmith:06JSAC,Huang&Rao:13WC}, we here consider the round-robin (RR) scheme and the proportional fairness (PF) scheme, and modify ReDOS-PBR in the previous section for RR and PF. During this modification, we exploit the degree-of-freedom associated with the parameter $\alpha$ of ReDOS-PBR (i.e., cone-containment checking is done at UTs and $\alpha$ can be adapted properly) and the fact that every UT reports CQI when $\alpha \le \alpha_{min}$ by Lemma \ref{lemma:alpha_min}.

\subsection{ReDOS-PBR for Round Robin}

There can be many modified versions ReDOS-PBR for RR
(ReDOS-PBR-RR). Here we consider the following modified scheme. In
RR, all users should be served in one round of scheduling.  For
this, we successively apply ReDOS-PBR to each scheduling interval
with controlling $\alpha$, until no unserved users are left.  For
the proposed ReDOS-PBR-RR, we assume that $\alpha$ is adapted at
the BS every scheduling interval and there exists a downlink
broadcasting control channel that informs every UT of the new
$\alpha$ value each scheduling interval.

Since  large $\alpha$ reduces the effective channel gain loss of the assumed ZFBF,
 large $\alpha$ is desired from the perspective of the effective channel gain. However, when $\alpha$ is too large (close to 1),
we would have $\Wc_{g,i} = \emptyset$ for some $i$, even though there are some users whose channels are
roughly aligned to the $i$-th reference direction. In this case, no user will be selected for the $i$-th reference direction and the spatial multiplexing gain will be reduced. Such an event can be avoided by reducing $\alpha$. In the proposed ReDOS-PBR-RR,  to detect such an event, every UT feeds back the most aligned reference direction index $i_{g_k}^*$ all the time, but feeds back $\Rc(g_k)$ only when the user's channel vector is contained in the cone $\Cc_{g,i_{g_k}^*}$. After the BS collects CQI from all UTs, the BS checks if there exists a reference direction index that has no associated $\Rc(g_k)$ feedback. Then, the BS knows whether the current $\alpha$ value is too high or not.

We now present the proposed ReDOS-PBR-RR trying to attain good
trade-off between the effective channel gain  and the spatial
multiplexing gain by exploiting the considered CQI feedback
strategy.

\vspace{0.5em}

\begin{algorithm}[ReDOS-PBR-RR]

~
\begin{enumerate}

\item[0)] Initialize $\alpha_{g}(1) \in [\alpha_{min},1)$,
$\Delta_\alpha>0$,
$\Kc_g = \{1,\cdots,K_g\}$,  and $t=1$.
\item At the scheduling interval $t$, choose the set $\Sc_g(t)$ of users among the users in $\Kc_g$ by ReDOS-PBR
with $\alpha_g(t)$. On the contrary to the original ReDOS-PBR,  every user with  $\Ic_{g_k} = \emptyset$ also feeds back its reference direction
index without the corresponding quasi-SINR in the CQI feedback phase for the modified version.

\item If $|\Sc_g(t)| < r_g^*$,
update $\alpha_g(t+1) \leftarrow \alpha_g(t) - \Delta_\alpha$. (That is, target more spatial multiplexing gain.)
If $|\Sc_g(t)| = r_g^*$,
 update $\alpha_g(t+1) \leftarrow \alpha_g(t)  +\Delta_\alpha$. (That is, target more effective channel gain.)
When $\alpha_g(t+1) \notin [\alpha_{min},1)$,
$\alpha_g(t+1) \leftarrow \alpha_g(t)$. The new $\alpha(t+1)$ is broadcast to all UTs.

\item Page the selected users $\Sc_g(t)$, obtain CSI from them, transmit data to them with ZFBF, and  update
$\Kc_g \leftarrow \Kc_g\backslash \Sc_g(t)$.

\item If $\Kc_g \neq \emptyset$,
update $t \leftarrow t+1$ and go to step 1).
Otherwise, stop.
\end{enumerate}
\end{algorithm}

\subsection{ReDOS-PBR for Proportional Fairness}

The proportionally fair (PF) scheduling  algorithm exploits
multiuser diversity gain with consideration of fairness
\cite{Viswanath&Tse&Laroia:02IT}. In the single-input
single-output (SISO) PF algorithm, the BS keeps track of the
average past served rate  $\mu_{g_k}$ for each user $g_k$ and
selects the user that has the maximum of the current supportable
rate $R_{g_k}(t)=\log(1+|h_{g_k}(t)|^2)$ (determined by the user's
current channel state) divided by the user's past average served
rate $\mu_{g_k}$. That is, the selection criterion is
$\frac{R_{g_k}(t)}{\mu_{g_k}}$ and the average served rate is
updated by a simple first-order autoregressive (AR) filter as
\begin{equation}
\mu_{g_k}(t+1) = \left(1-  \delta \right) \mu_{g_k}(t) + \delta
R_{g_k}(t) I_{\{g_k \in \Sc_g(t)\}}, \label{eq:mu_update2}
\end{equation}
where $I_A$ is the indicator function of event $A$, and $\Sc_g(t)$
is the set of scheduled users at time $t$. In
\cite{Yoo&Goldsmith:06JSAC}, the PF algorithm was extended to
incorporate MIMO situation and was applied to SUS-ZFBF. The main
difference between the SISO and MIMO cases is that the supportable
rate $R(g_k,\Sc_g(t))$ of each user $g_k$ cannot be computed
before user selection, because the rate itself depends on the user
selection in the MIMO case. However, this difficulty was intelligently circumvented
in \cite{Yoo&Goldsmith:06JSAC}, based on the semi-orthogonality of
the selected users. Since ReDOS-PBR also possesses the
semi-orthogonality among the selected users, we can apply the same
idea  as that in \cite{Yoo&Goldsmith:06JSAC} here. Since the
selected users are semi-orthogonal, we approximate the supportable
rate simply by
\begin{equation}
R(g_k,\Sc_g(t)) \approx \log(1+\Rc(g_k)) =: \hat{R}(g_k)(t).
\end{equation}
Thus, in the modified ReDOS-PBR for proportional fairness
(ReDOS-PBR-PF), for each reference direction, after the CQI
feedback phase, we select
\begin{equation}\label{eq:PFS_decom}
\kappa_{g,i}=\mathop{\arg\max}_{k \in \Wc_{g,i}} \frac{\hat{R}(g_k)(t)}{\mu_{g_k}(t)}
~~~\text{for}~ i =1,\cdots,r_g^*.
\end{equation}
Then, the BS collects CSI from the selected users, transmits data
after post-selection beam refinement, computes the exact served
rate for the scheduled users, and update $\mu_{g_k}$ by
\eqref{eq:mu_update2}.

One requirement for ReDOS-PBR-PF to compute \eqref{eq:PFS_decom}
for all users at each scheduling interval $t$ is that all users
should report their CQI (the reference beam index and quasi-SINR)
to the BS at every interval $t$.  This can be done simply by
setting $\alpha = \alpha_{min}$ for all users by Lemma
\ref{lemma:alpha_min}. However, CQI feedback can be reduced by
exploiting the property of PF itself and distributed and
individual control $\alpha$ at each UT. Note that once a user
$g_k$ is served, $\mu_{g_k}$ increases suddenly and the selection
criterion in \eqref{eq:PFS_decom} decreases suddenly. Hence, user
$g_k$ will not be selected in the next scheduling interval unless
user $g_k$'s channel vector at the next scheduling interval is
highly aligned with some reference beam direction with large
magnitude. Therefore, the served user can increase its own
$\alpha$ denoted by $\alpha_{g_k}(t)$ suddenly by some step
$\Delta_{\alpha,up}$, targeting bigger chance for good channel
realization. When the user is not served, $\alpha_{g_k}$ is
reduced by $\Delta_{\alpha,down}$ (say,
$\Delta_{\alpha,down}=\Delta_{\alpha,up}/T$ with $T > 1$). Then,
$\alpha_{g_k}(t)$ comes back to $\alpha_{min}$ in some time and
user $g_k$ surely reports CQI again. Here, $\Delta_{\alpha,up}$
and $\Delta_{\alpha,down}$ are system design parameters which
should be determined properly.  Note that $\Delta_{\alpha,up}$
and $\Delta_{\alpha,down}$ can be used not only for feedback reduction but also for fairness enhancement, since it is highly likely that a served user will not be served again successively.
Such an efficient
semi-orthogonality and feedback control is possible for ReDOS-PBR
because  cone-containment checking for semi-orthogonality is done
individually  at UTs.  Summarizing the above-mentioned idea, we
now present the proposed ReDOS-PBR-PF:

\vspace{0.5em}

\begin{algorithm}[ReDOS-PBR-PF]  \label{algo:ReDOS_PBR_PF}

~
\begin{enumerate}
\item[0)] Initialize $\alpha_{g_k}(1)=\alpha_{min}$, $\mu_{g_k}(1)
= \mu >0$, $\forall g,k$, and  $t=1$, and
$\Delta_{\alpha,up}>\Delta_{\alpha,down}>0$. (Now each user has
its own $\alpha_{g_k}(t)$.)

\item At time $t$, each user $g_k$ computes $\Ic_{g_k}$ in
\eqref{eq:align} based on  its own $\alpha_{g_k}(t)$. Then, follow
the remaining sub-steps in step 1) of original ReDOS-PBR.

\item The BS chooses the set of users $\Sc_g(t)$ by computing
\eqref{eq:PFS_decom} after the CQI feedback phase.

\item After the CSI feedback phase, the BS serves the scheduled
users in $\Sc_{g}(t)$ with ZFBF. Then, the BS updates
$\mu_{g_k}(t)$ according to \eqref{eq:mu_update2} with the
actually served rate $R(g_k,\Sc_g(t))$.

\item The users in $\Sc_g(t)$ update $\alpha_{g_k}(t+1) \leftarrow
\alpha_{g_k}(t) + \Delta_{\alpha,up}$ and other unserved  users
update $\alpha_{g_k}(t+1) \leftarrow \alpha_{g_k}(t) -
\Delta_{\alpha,down}$. (Users know whether they are served or not
during the scheduled user paging time.) When $\alpha_{g_k}(t+1)
\notin [\alpha_{min},1)$, $\alpha_{g_k}(t+1) \leftarrow
\alpha_{g_k}(t)$.

\item Update $t \leftarrow t +1$ and go to step 1).
\end{enumerate}

\end{algorithm}

\vspace{0.5em}

In the above algorithm, UTs exploit $\alpha$ for efficient CQI
feedback control, but UTs can exploit $\hat{R}(g_k)(t)$ in
addition to $\alpha_{g_k}(t)$ for the same purpose since each UT
can compute $\hat{R}(g_k)(t)=\log(1+\Rc(g_k))$ by itself. There
can be various ways to combine $(\alpha_{g_k}(t),
\hat{R}_{g_k}(t))$ for efficient distributed CQI feedback control.

\vspace{0.5em}

\begin{remark}[On extension to the case of UTs with multiple receive antennas]
ReDOS-PBR can be extended without difficulty to the case in which
UTs have multiple receive antennas. In this case, each antenna can
be regarded as a different user, and ReDOS-PBR for single-antenna
UTs can be applied \cite{Yoo&Goldsmith:06JSAC}. Here, a UT with
multiple receive antenna imposes a restriction that the candidate
set for one receive antenna and that of another receive antenna
are different.
\end{remark}

\section{Numerical Results}
\label{sec:num_res}

\begin{figure}[tp]
    \centerline{
\scalefig{0.5} \epsfbox{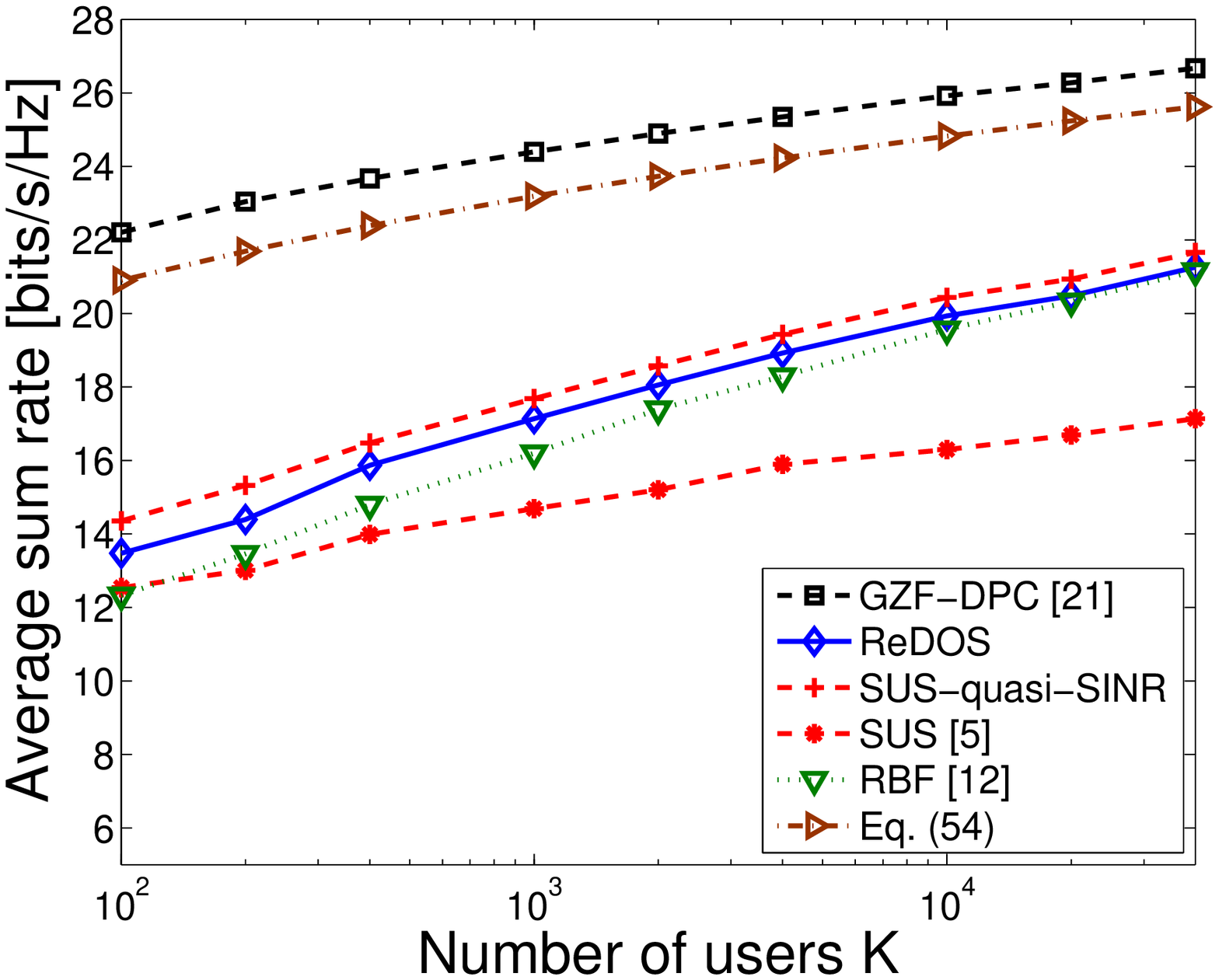}
\scalefig{0.51} \epsfbox{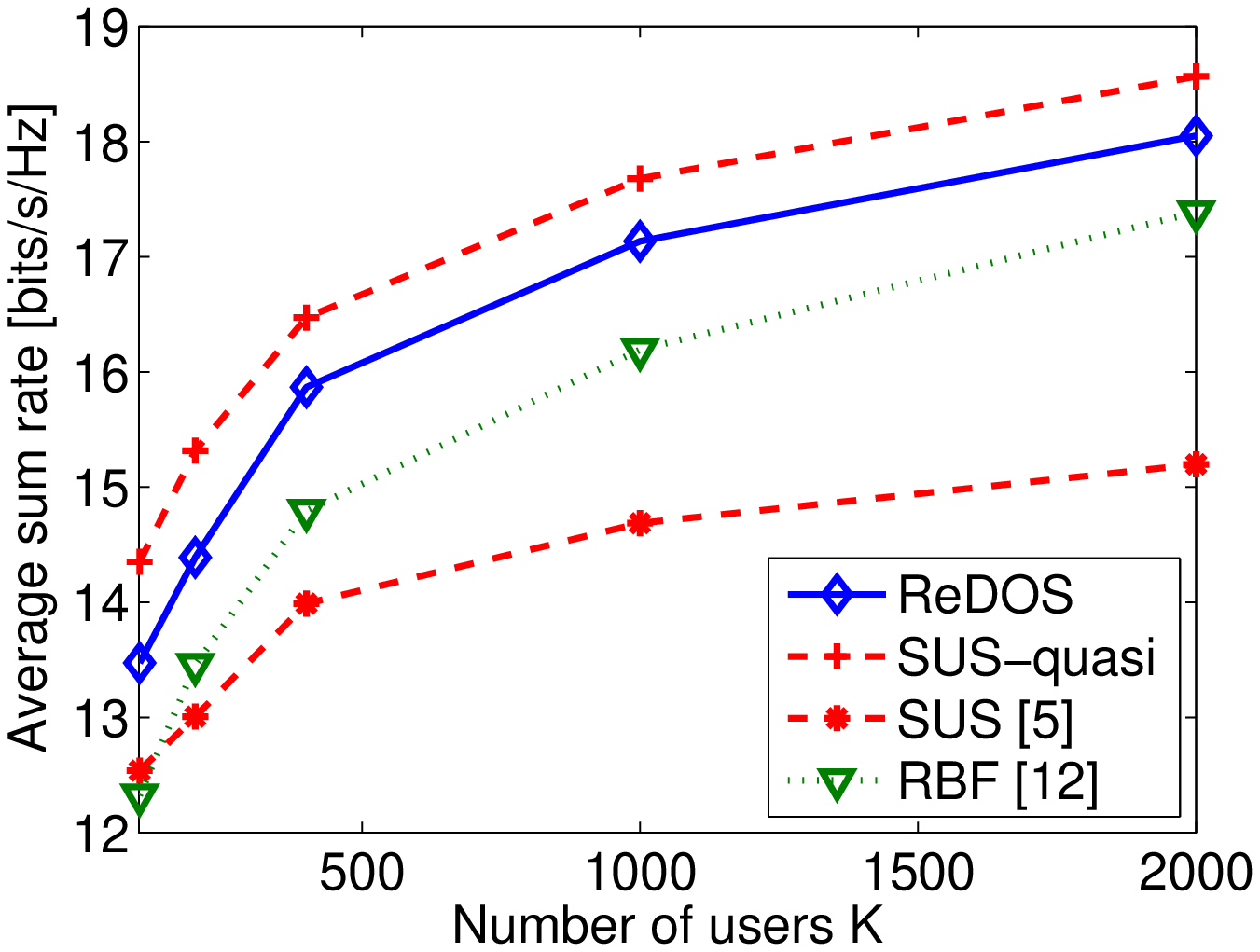}}
    \caption{Multi-group performance: (a) average sum rate w.r.t. the number of users and (b) the same figure as (a) with the range from $K=100$ to $K=2000$}
    \label{fig:multi_group_scaling_law}
\end{figure}

In this section, we provide some numerical results regarding the proposed user-scheduling-and-beamforming method. First, we verified the asymptotic analysis in Section \ref{sec:AUS_opt}. To verify the asymptotic analysis, we considered a small MISO downlink system (with two groups ($G=2$) and inter-group interference) to which DPC-based beamforming \cite{Dimic&Sidiropoulos:05SP} can be applied. The system consisted of a BS with four transmit antennas ($M=4$) and $P = 15$ dB and $K$ single-antenna UTs, and the channel vectors were independently generated according to the model \eqref{eq:channel_linear_comb} with $\Rbf_1=\Ubf_1 \Lambdabf_1 \Ubf_1^H$ and $\Rbf_2=\Ubf_2 \Lambdabf_2 \Ubf_2^H$, where $\Ubf_1=\Fbf_{DFT}^{(4)}(:,1:3)$, $\Ubf_2=\Fbf_{DFT}^{(4)}(:,3:4)$, $\Lambdabf_1 = \mbox{diag}(1,r,r^2)$, $\Lambdabf_2 = \mbox{diag}(1,r)$, $r=0.7$, and $\Fbf_{DFT}^{(4)}$ is the 4-point discrete Fourier transform (DFT) matrix. The pre-beamformer matrices were chosen as $\Vbf_1= \Ubf_1^*=\Ubf_1(:,1:2)$ and $\Vbf_2=\Ubf_2^*=\Ubf_2$ to satisfy the approximate BD condition in Condition  \ref{cond:approxBD}. Fig. \ref{fig:multi_group_scaling_law} shows the result. In the figure, the performance of the DPC-based beamforming in \cite{Dimic&Sidiropoulos:05SP} is shown as the performance reference. (In \cite{Dimic&Sidiropoulos:05SP}, the authors proposed a greedy user selection method based on QR decomposition and the assumption of the availability of DPC.) It is seen that the predicted asymptotic scaling behavior of ReDOS-PBR shown in eq. \eqref{eq:proof_theo2_f} has the same slope as the DPC-based user selection method in \cite{Dimic&Sidiropoulos:05SP}. The actual finite-user sum-rate behavior of several algorithms is also shown in Fig. \ref{fig:multi_group_scaling_law}. We considered ReDOS-PBR, RBF in \cite{Adhikary&Caire:13arXiv}, the original SUS-ZFBF in \cite{Yoo&Goldsmith:06JSAC}, and a modified SUS-ZFBF using quasi-SINR in \eqref{eq:quasiSINRdef}. (Since the original SUS-ZFBF with the channel norm criterion was proposed for the single-cell (or single-group) case, we considered SUS-ZFBF with quasi-SINR for the multi-group case for fair comparison.) It is seen that SUS-ZFBF with quasi-SINR, ReDOS-PBR and RBF all follow the slope of the DPC-based scheme as expected. It is also seen that SUS-ZFBF with the norm criterion does not handle inter-group interference properly. As expected, SUS-ZFBF (with quasi-SINR) performs best, RBF performs worst, and ReDOS-PBR is in-between. In the considered small system case, the performance difference between the three algorithms is not so significant.

With the asymptotic scaling behavior w.r.t. $K$ verified, we considered a more realistic scenario.
We considered a MISO downlink  system
where a BS with $P=15$ dB was equipped with a
ULA of $M=32$ antenna elements and each of $K$ UTs had  a single receive antenna. The UTs were grouped into eight groups ($G=8$), and the BS served four UTs ($r_g^*=4$) simultaneously for each group. The channel covariance matrix and the pre-beamformer matrix of each group were chosen as
\begin{align*}
\Ubf_1 &= \Fbf_{DFT}^{(32)}[:,1:5],  ~\Vbf_1=\Ubf_1^*=\Ubf_1[:,1:4](=\Fbf_{DFT}^{(32)}[:,1:4])\\
\Ubf_2 &= \Fbf_{DFT}^{(32)}[:,5:9],  ~\Vbf_2=\Ubf_2^*=\Ubf_2[:,1:4](=\Fbf_{DFT}^{(32)}[:,5:8])\\
\Ubf_3 &= \Fbf_{DFT}^{(32)}[:,9:13],  ~\Vbf_3=\Ubf_3^*=\Ubf_3[:,1:4](=\Fbf_{DFT}^{(32)}[:,9:12])\\
\Ubf_4 &= \Fbf_{DFT}^{(32)}[:,13:17],  ~\Vbf_4=\Ubf_4^*=\Ubf_4[:,1:4](=\Fbf_{DFT}^{(32)}[:,13:16])\\
\Ubf_5 &= \Fbf_{DFT}^{(32)}[:,17:21], ~\Vbf_5=\Ubf_5^*=\Ubf_5[:,1:4](=\Fbf_{DFT}^{(32)}[:,17:20])\\
\Ubf_6 &= \Fbf_{DFT}^{(32)}[:,21:25], ~\Vbf_6=\Ubf_6^*=\Ubf_6[:,1:4](=\Fbf_{DFT}^{(32)}[:,21:24])\\
\Ubf_7 &= \Fbf_{DFT}^{(32)}[:,25:29], ~\Vbf_7=\Ubf_7^*=\Ubf_7[:,1:4](=\Fbf_{DFT}^{(32)}[:,25:28])\\
\Ubf_8 &= \Fbf_{DFT}^{(32)}[:,29:32], \Vbf_8=\Ubf_8^*=\Ubf_8(=\Fbf_{DFT}^{(32)}[:,29:32]),
\end{align*}
where $\Fbf_{DFT}^{(32)}$ is the 32-point DFT matrix, and $\Lambdabf_i = \text{diag}(1,r,r^2,r^3,r^4)$ with $r=0.6$  for $i=1,\cdots,7$  and $\Lambdabf_8 = \text{diag}(1,r,r^2,r^3)$. This setting of channel covariance matrices and pre-beamformer matrices satisfies the approximate BD condition.  Fig. \ref{fig:multi_group_approximateBD_sumrate} (a) shows the sum-rate performance of the three schemes: SUS-ZFBF, RBF, and ReDOS-PBR. 200 independent channel realizations according to \eqref{eq:channel_linear_comb} were used  for each $K$ and the average sum rate is the average over the 200 channel realizations. (For the figure, the user-selection hyperslab thickness for SUS-ZFBF and the user-selection cone angle for ReDOS-PBR were optimally chosen for each $K$.)
\begin{figure}[tp]
    \centerline{
\scalefig{0.5} \epsfbox{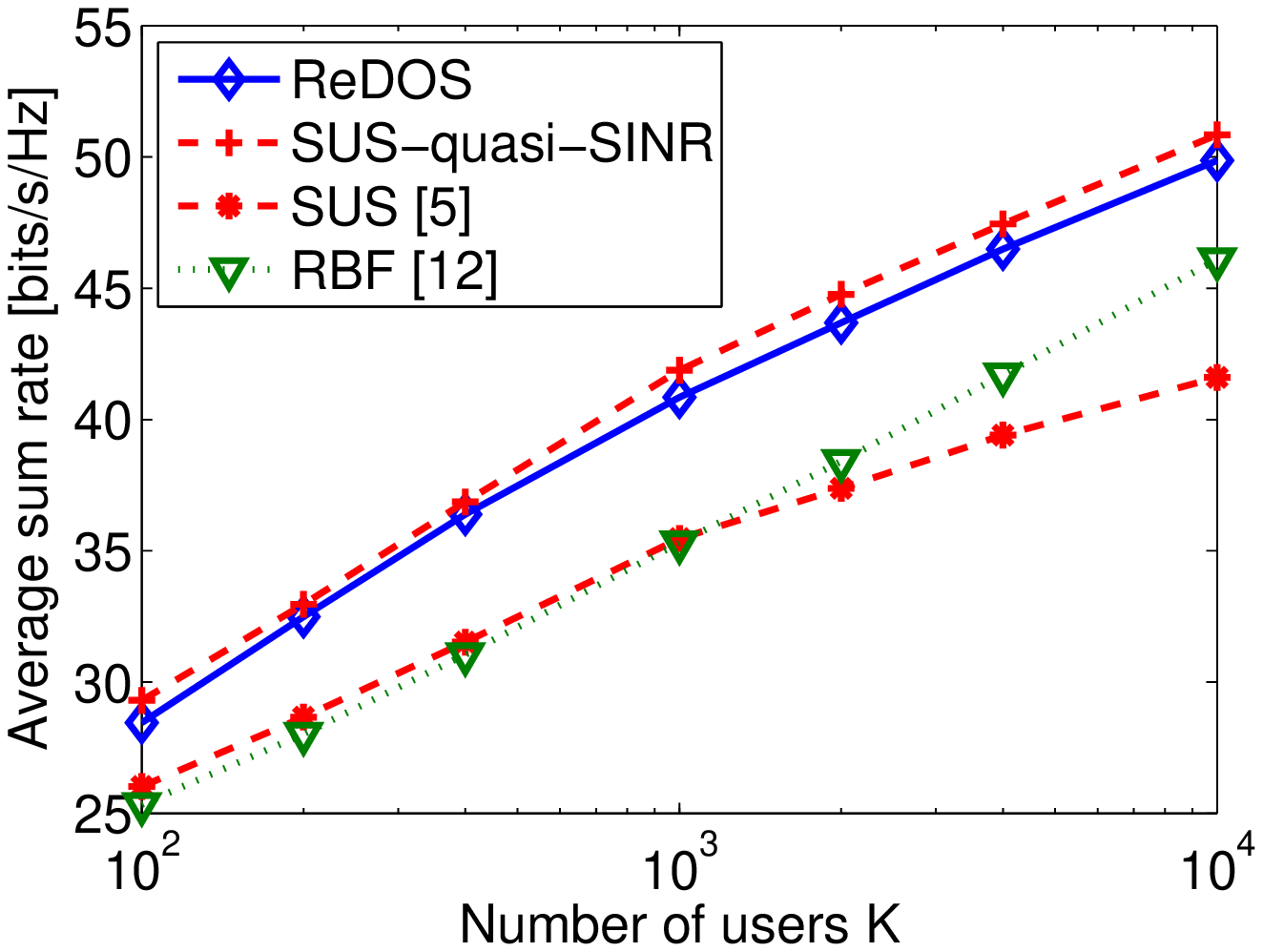}
\scalefig{0.5} \epsfbox{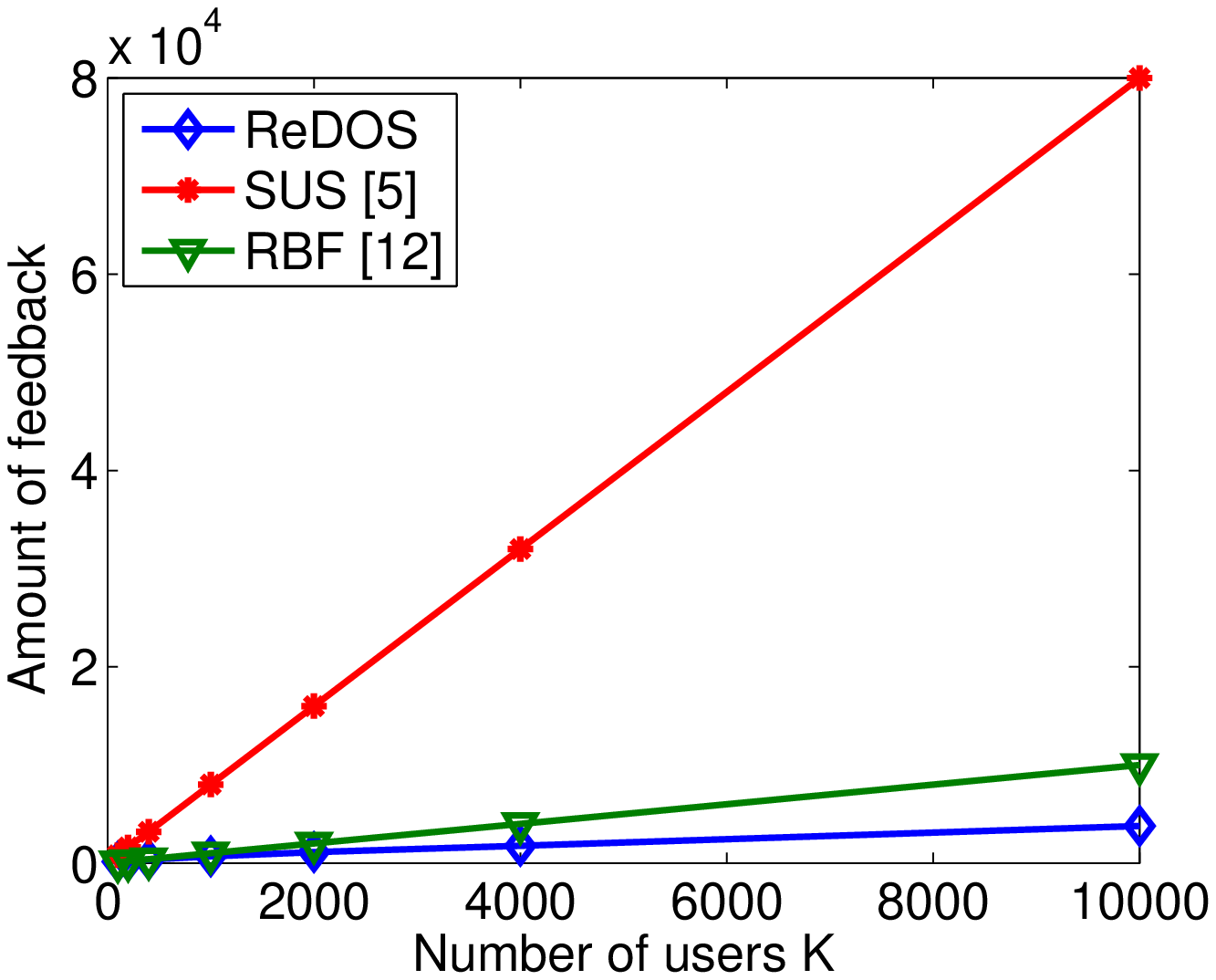}}
    \caption{Multi-group performance: (a) average sum rate performance and
    (b) amount of feedback (number of required real numbers)}
    \label{fig:multi_group_approximateBD_sumrate}
\end{figure}
\begin{figure}[tp]
     \centerline{
\scalefig{0.5} \epsfbox{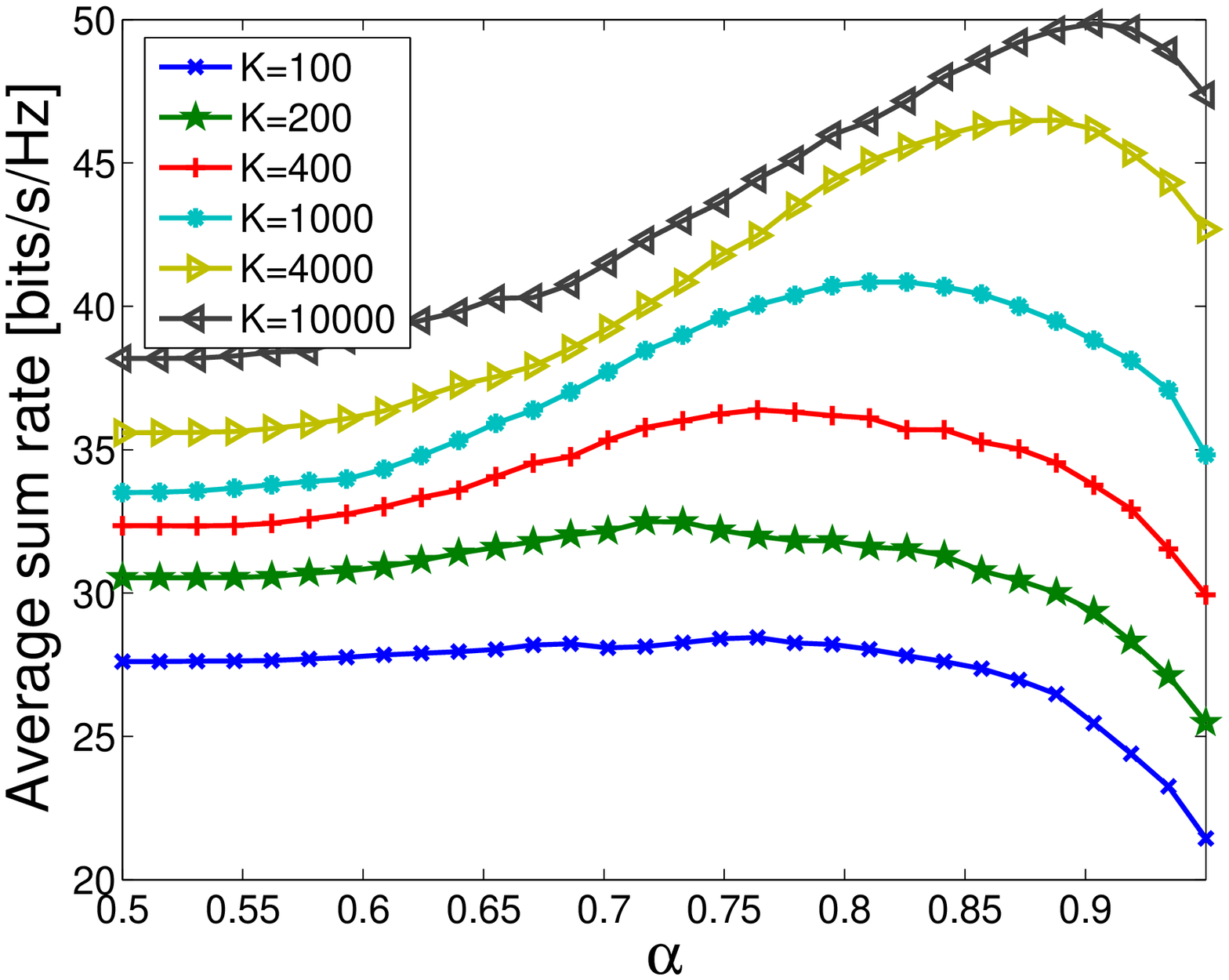}
\scalefig{0.51} \epsfbox{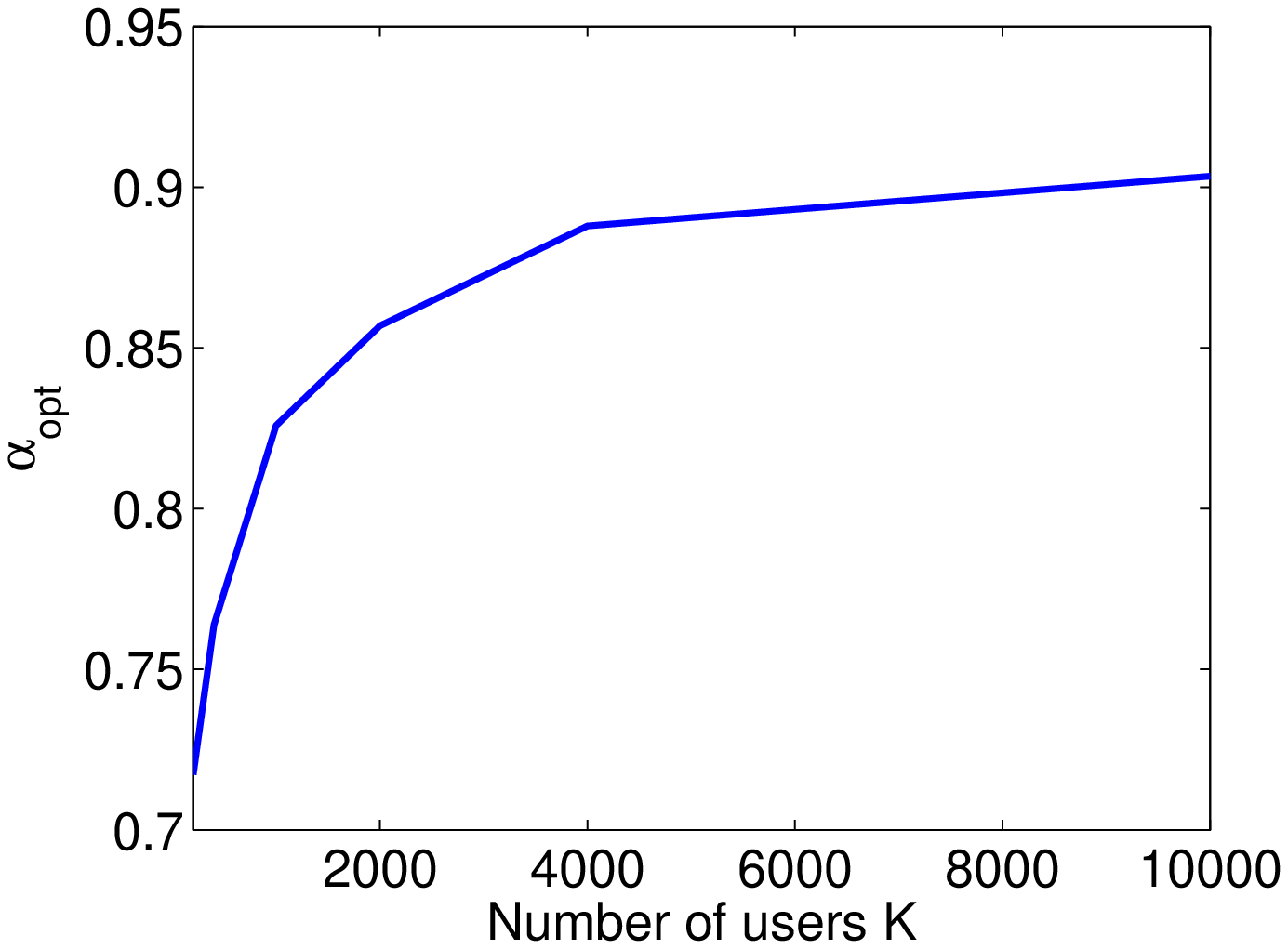}}
    \caption{Multi-group performance: (a) sum rate w.r.t. $\alpha$ and (b) optimal $\alpha$}
    \label{fig:multi_group_approximateBD_alpha}
\end{figure}
Now it is seen that the performance gap between SUS-ZFBF with quasi-SINR and RBF is significant. It is also seen that proposed ReDOS-PBR closely follows SUS-ZFBF with quasi-SINR.  Fig. \ref{fig:multi_group_approximateBD_sumrate} (b) shows the amount of feedback for the same setting as in \ref{fig:multi_group_approximateBD_sumrate} (a). As expected, SUS-ZFBF requires the largest amount of feedback. Note that the amount of feedback required for ReDOS-PBR is even less than RBF! We then investigated the performance sensitivity of ReDOS-PBR w.r.t. $\alpha$ in the same setting as in Fig. \ref{fig:multi_group_approximateBD_sumrate}, and the result is shown in Fig. \ref{fig:multi_group_approximateBD_alpha}. It is seen that optimal $\alpha$ increases as $K$ increases.   An observation of practical importance is that the performance of ReDOS-PBR is quite insensitive w.r.t. $\alpha$ for the practical range of the number of users.

Next, we considered a single-group case for which SUS-ZFBF is originally proposed. The considered system consists of a BS with $M=r_g^*=4$ and $P=10$ [dB], and $K$ UTs each with a single antenna. The channel vector for each user was generated i.i.d. according to the model \eqref{eq:channel_linear_comb}, where for the channel covariance matrix $\Rbf_1$, the exponential correlation model is used, i.e., \cite{Alnaffouri&Sharif&Hassibi:09COM}
\begin{equation}  \label{eq:numer_exp_correl_model}
[\Rbf_1]_{i,j} = \nu^{|i-j|}
\end{equation}
with $0 \le \nu \le 1$.
\begin{figure}[tp]
     \centerline{
\scalefig{0.5} \epsfbox{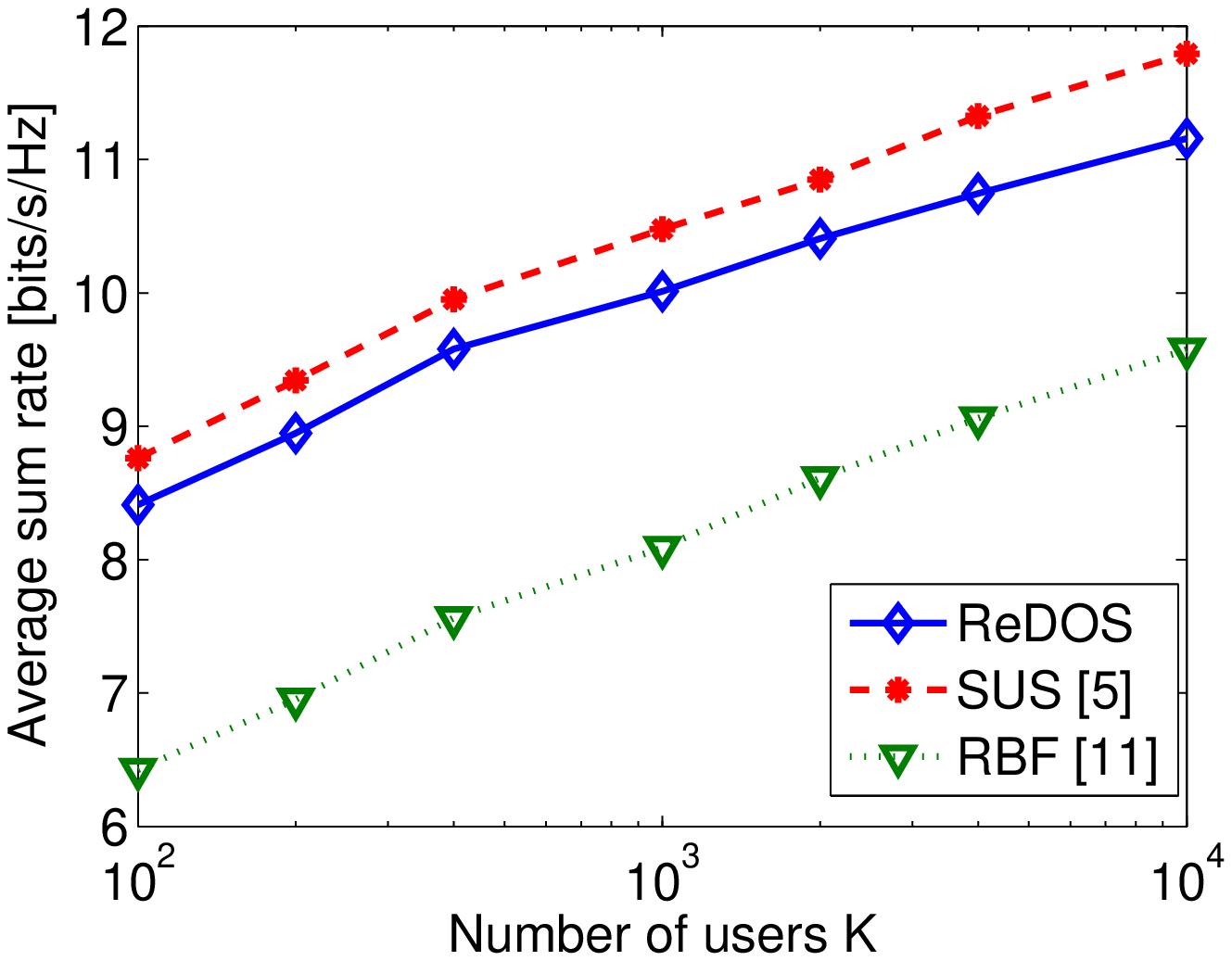}
\scalefig{0.5} \epsfbox{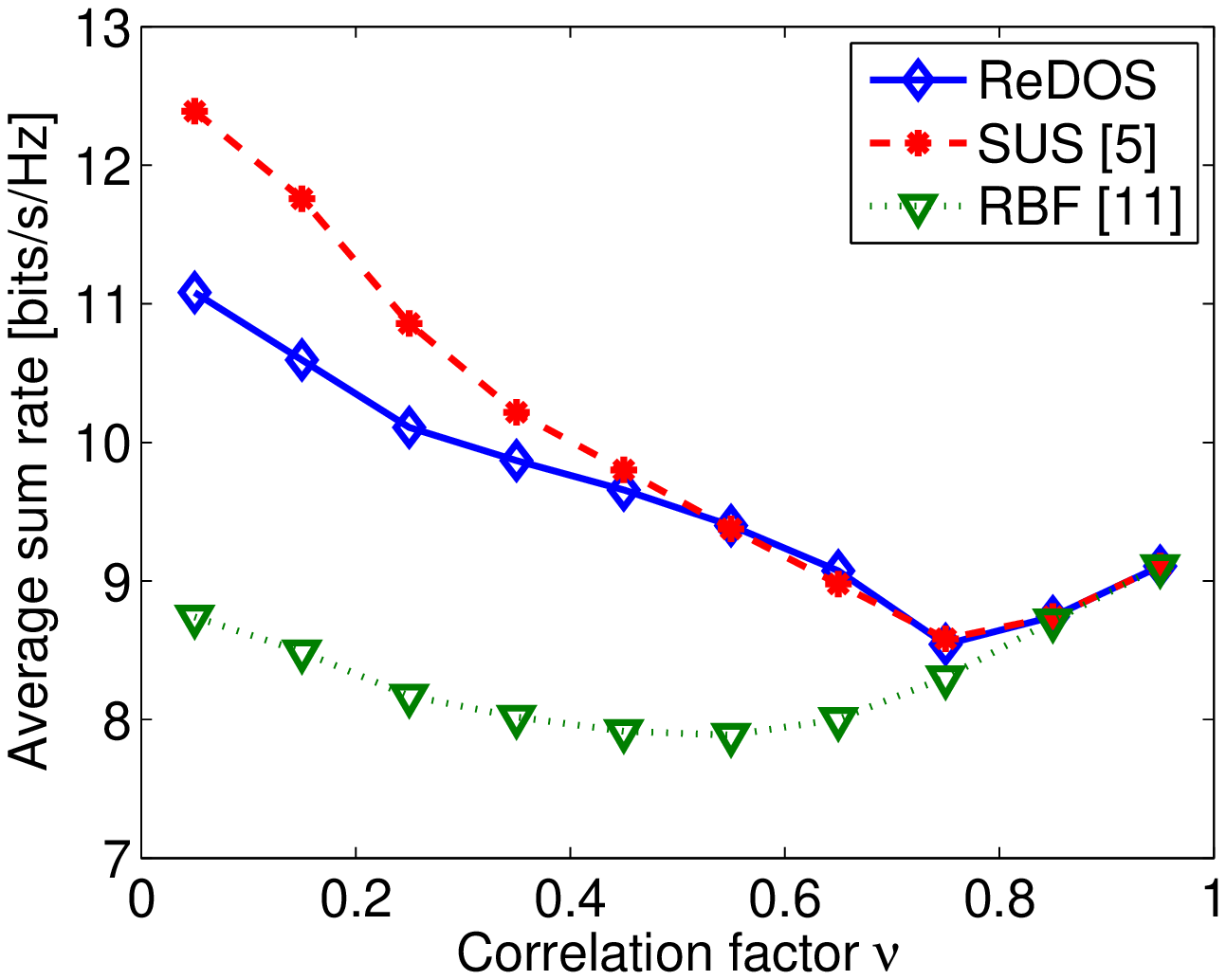}}
    \caption{Single-group performance: (a) average sum rate w.r.t. $K$ and (b) average sum rate w.r.t. the channel correlation factor $\nu$}
    \label{fig:single_group_performance}
\end{figure}
Fig. \ref{fig:single_group_performance} (a) shows the average sum-rate performance of SUS-ZFBF, RBF and ReDOS-PBR for $\nu=0.3$.  Again, there exists a significant performance gap between SUS-ZFBF and RBF, and ReDOS-PBR closely follows SUS-ZFBF. Fig. \ref{fig:single_group_performance} (b) shows the performance of the three schemes w.r.t. the channel correlation factor $\nu$ with $K$ fixed to 1000 for the same setting as in Fig. \ref{fig:single_group_performance} (a).  As expected, when $\nu=0$, i.e., the channel is isotropic, SUS-ZFBF performs best, and when $\nu=1$, i.e., the channel matrix has rank one and only one beam can be supported, all three algorithms perform equally. It is seen that the noticeable gap between SUS-ZFBF and ReDOS-PBR at $\nu=0$ decreases as $\nu$ increases towards one.
This is because when the channel becomes more correlated, there start to exist dominant eigen-directions of the channel and thus, it is enough to make these dominant eigen-directions of the channel the reference beam directions of ReDOS-PBR and to look around these reference directions.

\begin{figure}[tbp]
\centerline{ \SetLabels
\L(0.25*-0.1) (a) \\
\L(0.76*-0.1) (b) \\
\endSetLabels
\leavevmode
\strut\AffixLabels{
\scalefig{0.5}\epsfbox{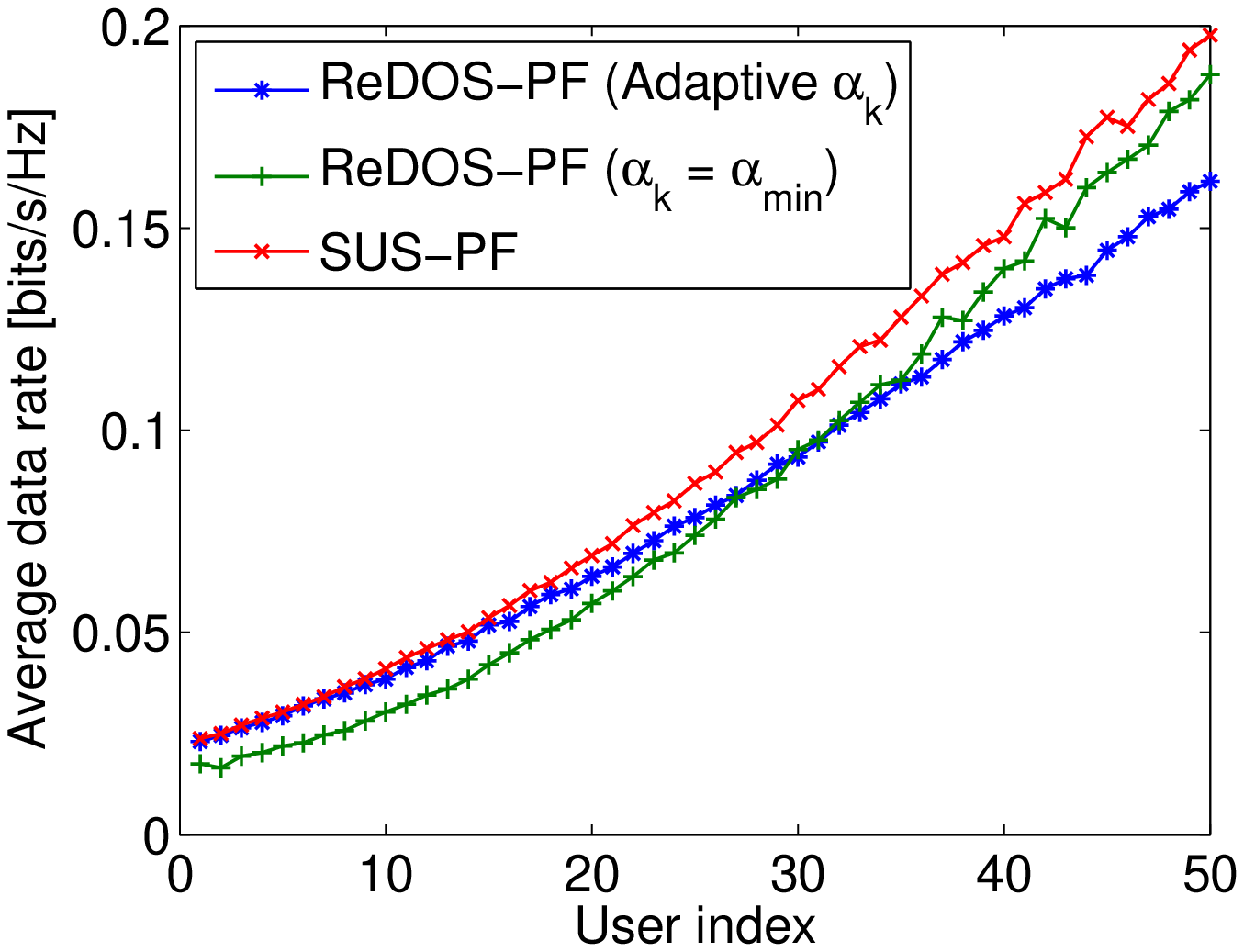}
\scalefig{0.5}\epsfbox{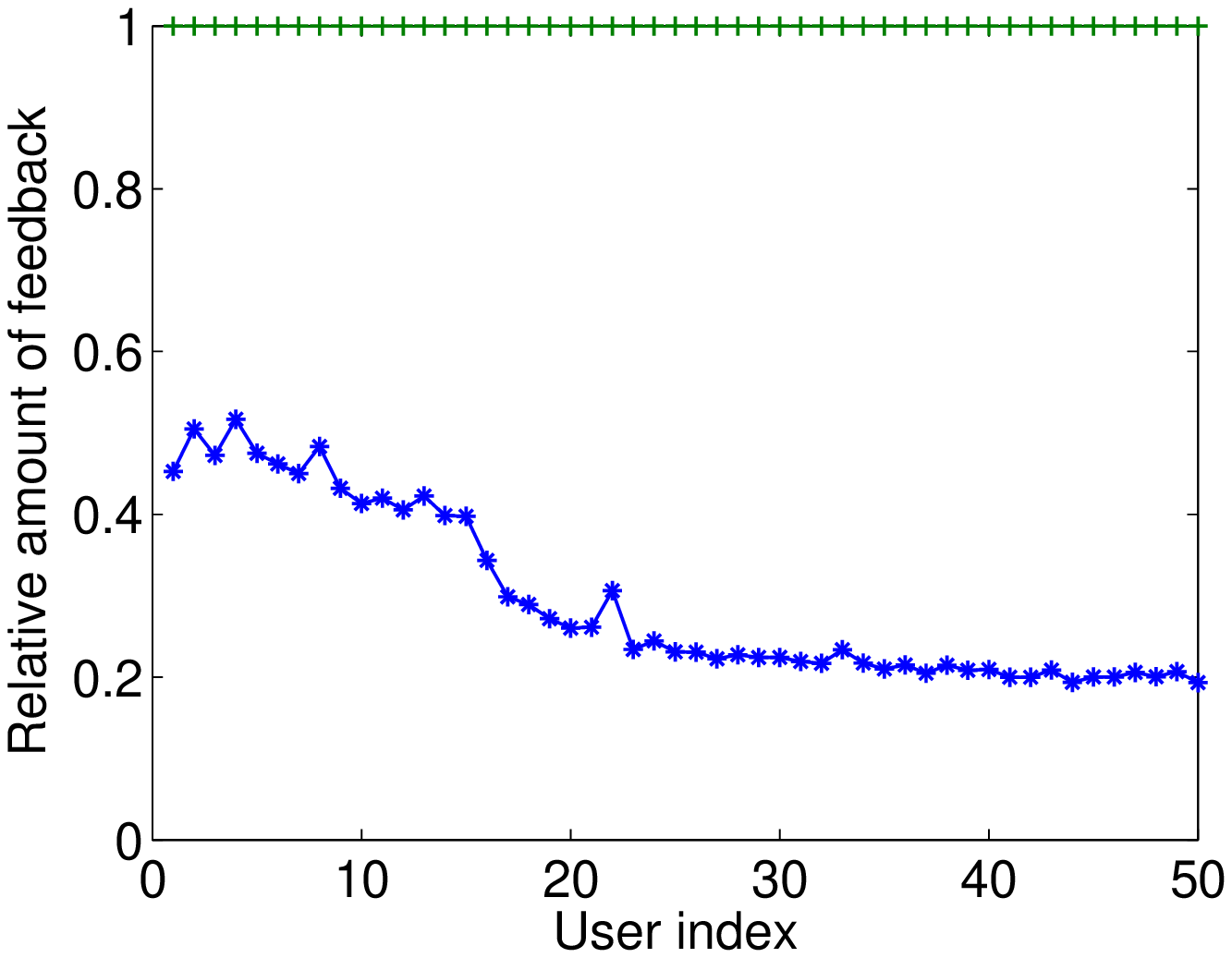} }
} \vspace{0.5cm} \centerline{ \SetLabels
\L(0.25*-0.1) (c) \\
\L(0.76*-0.1) (d) \\
\endSetLabels
\leavevmode
\strut\AffixLabels{
\scalefig{0.5}\epsfbox{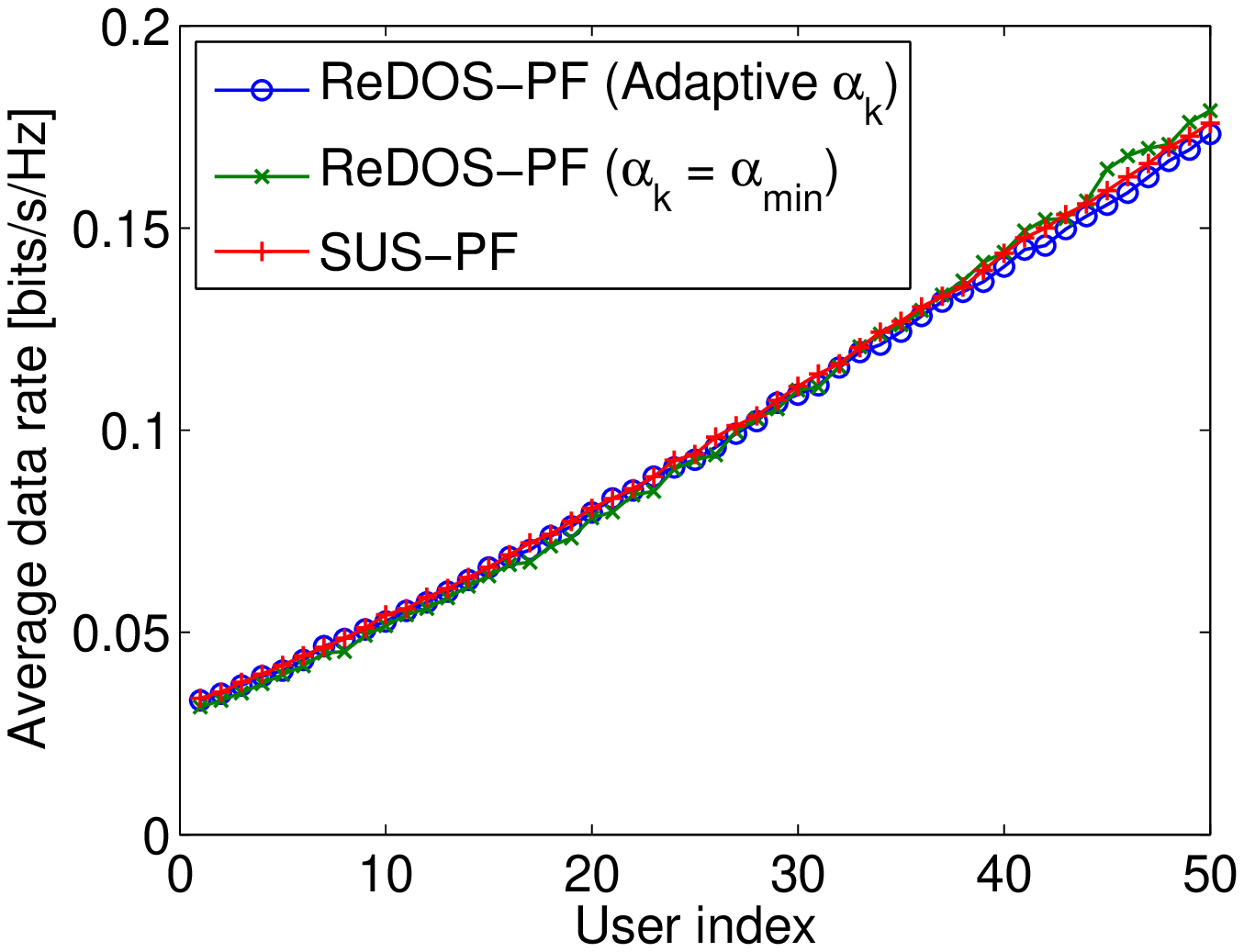}
\scalefig{0.5}\epsfbox{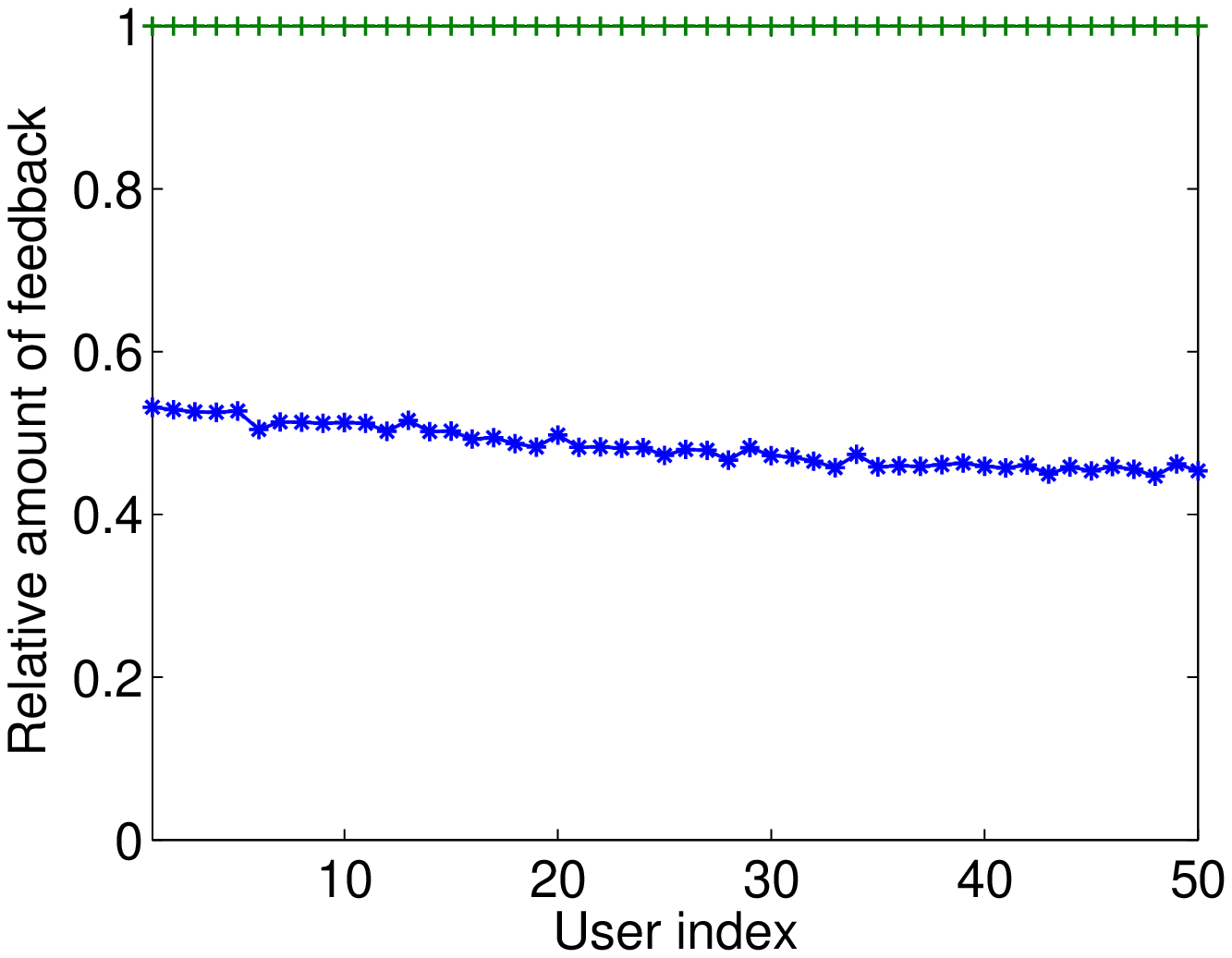}
} } \vspace{0.5cm} \caption{Performance of ReDOS-PBR-PF: (a) each user's served rate, (b) relative amount of feedback between two ReDOS-PBR-PF algorithms: one with fixed $\alpha$ and the other with adaptive $\alpha$, (c) each user's served rate, and (d) relative amount of feedback between two ReDOS-PBR-PF algorithms ((a) and (b): $\nu=0.1$ and $\Delta_{\alpha,up}=0.1$ and $\Delta_{\alpha,down}=\Delta_{\alpha,up}/50$, and  (c) and (d): $\nu=0.3$ and $\Delta_{\alpha,up}=0.2$ and $\Delta_{\alpha,down}=\Delta_{\alpha,up}/100$  )} \label{fig:ReDOSPBR_PF}
\end{figure}

Finally, we examined the performance of ReDOS-PBR-PF. We considered two ReDOS-PBR-PF algorithms: One with fixed $\alpha=\alpha_{min}$ for all users and the other with adaptive $\alpha$ for each user with steps $\Delta_{\alpha,up}$ and $\Delta_{\alpha,down}$ described in Algorithm \ref{algo:ReDOS_PBR_PF}. To simplify the simulation, we just considered the single-group system considered in Fig. \ref{fig:single_group_performance}. Here, we fixed $K=50$ and  BS $P=0$ [dB].  The channel vector for each user $k$ was generated as
\begin{equation}
\hbf_k \sim \sqrt{l_k}\Rbf_1^{1/2}\etabf_k,
\end{equation}
where $\Rbf_1$ is given in \eqref{eq:numer_exp_correl_model}; $\etabf_k \stackrel{i.i.d.}{\sim} \Cc\Nc({\mathbf{0}}, \Ibf_4)$; and  the large-scale fading effect is captured in $l_k$.  The large-scale fading factor $l_k$ for 50 users were designed such that the lowest power user has $l_k=1$ and the highest power user has $l_k=100$ (20 dB difference), and  other users' power is equally spaced in dB scale in the 20 dB power range. We ran 10,000 scheduling intervals. For each interval, the channel vector for each user was generated independently as described in the above. Fig. \ref{fig:ReDOSPBR_PF} (a) shows the average served rate for 50 users  (users are ordered in an ascending order of their $l_k$ values) over 10,000 scheduling intervals, when the channel is almost isotropic, i.e., $\nu =0.1$. It is seen in this case that there is some loss of ReDOS-PBR-PF compared to SUS-ZFBF-PF. Note that ReDOS-PBR-PF with fixed $\alpha=\alpha_{min}$ tracks SUS-ZFBF-PF for all users with equal gap, but ReDOS-PBR-PF with adaptive $\alpha$ sacrifices high-SNR users and gives more chances to low-SNR users. (This is evident in Fig. \ref{fig:ReDOSPBR_PF} (b).) This is because high-SNR users have more chances to be selected and thus, increase their $\alpha_k$ to reduce this increased chance.  Fig. \ref{fig:ReDOSPBR_PF} (b) shows the relative amount of feedback for ReDOS-PBR-PF with adaptive $\alpha$ to that of ReDOS-PBR-PF with fixed $\alpha$. It is seen that the amount of feedback is significantly reduced by adapting $\alpha$.  Figs. (c) and (d) show the performance and the relative amount of feedback in the case of $\nu=0.3$. It is seen that when the channel correlation increases, the performance difference between ReDOS-PBR-PF and SUS-ZFBF-PF is negligible.

\section{Conclusion}
\label{sec:conclusion}

In this paper, we have proposed a new efficient user-scheduling-and-beamforming
method for massive MU-MIMO broadcast channels. The proposed method takes advantage of two existing user-scheduling-and-beamforming methods, SUS-ZFBF and RBF, for MU-MIMO broadcast channels sitting on opposite sides on the scale of feedback overhead. The proposed scheduling-and-beamforming method is
asymptotically optimal as the number of users increases. The proposed method yields `nearly-optimal' user-selection-and-beamforming under the linear beamforming framework for MU-MIMO downlink, based  on CQI-only feedback from possibly all users and CSI feedback  from only the scheduled users.

\begin{appendices}

\section{Inner Product Between Two  Vectors in Two Different Cones}
\label{appen:cone_inner_prod}

\begin{lemma} \label{lemma:ConeDistance} For two channel vectors contained in two different user-selection cones, i.e.,
 ${\hbf}_{\kappa_{g,i}} \in \Cc_{g,i}$ and ${\hbf}_{\kappa_{g,j}} \in \Cc_{g,j}$, $i \ne j$,
 the inner product between the corresponding normalized effective channel vectors
  $\tilde{\gbf}_{\kappa_{g,i}}$ and  $\tilde{\gbf}_{\kappa_{g,j}}$  with norm one is bounded by
\begin{equation}
|\tilde{\gbf}_{\kappa_{g,i}}^H \tilde{\gbf}_{\kappa_{g,j}}|
\le 2 \alpha \sqrt{1-\alpha^2}~~ \text{for} ~~ i \neq j,
\end{equation}
when $\alpha \ge 1/\sqrt{2}$ (i.e., the angle $\theta \le \pi/4$
in Fig.  \ref{fig:aus_region}).
\end{lemma}

\begin{proof}
Let $\tilde{\gbf}_{\kappa_{g,i}}=\sum_{m=1}^{r_g^*}
c_{\kappa_{g,i}}^m \ebf_{m}^{(g)}$ and
$\tilde{\gbf}_{\kappa_{g,j}}=\sum_{m=1}^{r_g^*} c_{\kappa_{g,j}}^m
\ebf_{m}^{(g)}$, where $\ebf_{m}^{(g)}$ is the $m$-th column of
$\Ibf_{r_g^*}$. Then, we have $\sum_m
|c_{\kappa_{g,i}}^m|^2=\sum_m |c_{\kappa_{g,j}}^m|^2=1$ and
\begin{align}
|\tilde{\gbf}_{\kappa_{g,i}}^H \tilde{\gbf}_{\kappa_{g,j}}|
&= \left|\sum_{m=1}^{r_g^*} \bar{c}_{\kappa_{g,i}}^m c_{\kappa_{g,j}}^m\right| \nonumber \\
&\le \sum_{m=1}^{r_g^*}|\bar{c}_{\kappa_{g,i}}^m|\cdot |c_{\kappa_{g,j}}^m|  \nonumber\\
&= |\bar{c}_{\kappa_{g,i}}^i| \cdot |c_{\kappa_{g,j}}^i| + |\bar{c}_{\kappa_{g,i}}^j|\cdot |c_{\kappa_{g,j}}^j|
+ \sum_{m=1, m \neq i,j}^{r_g^*}|\bar{c}_{\kappa_{g,i}}^m| |c_{\kappa_{g,j}}^m|  \nonumber\\
&{\le} |\bar{c}_{\kappa_{g,i}}^i| \cdot |c_{\kappa_{g,j}}^i| +
|\bar{c}_{\kappa_{g,i}}^j|\cdot |c_{\kappa_{g,j}}^j| +
\left(\sum_{m=1, m \neq
i,j}^{r_g^*}|\bar{c}_{\kappa_{g,i}}^m|^2\right)^{\frac{1}{2}}
\left(\sum_{m=1, m \neq
i,j}^{r_g^*}|c_{\kappa_{g,j}}^m|^2\right)^{\frac{1}{2}},
\label{eq:cone_inner_upper}
\end{align}
where $\bar{c}$ is the complex conjugate of $c$, and the last step
follows from the Cauchy-Schwarz inequality.

Now consider the RHS in \eqref{eq:cone_inner_upper}. First, fix
$\{c_{\kappa_{g,j}}^m\}_{m=1}^{r_g^*}$ and
$|\bar{c}_{\kappa_{g,i}}^i|$, and view the RHS in
\eqref{eq:cone_inner_upper} as a function of
$\{c_{\kappa_{g,i}}^m, m=1,\cdots,r_g^* ~\mbox{and}~ m \ne i ~|~\sum_{m=1,m\ne i}^{r_g^*}|c_{\kappa_{g,i}}^m|^2=1-|c_{\kappa_{g,i}}^i|^2 \}$.
Then, the RHS in \eqref{eq:cone_inner_upper} is in the form of
$a+bx+cy$, where the constants $a,b,c \ge 0$ are given by
$a=|\bar{c}_{\kappa_{g,i}}^i| \cdot |c_{\kappa_{g,j}}^i|$,
$b=|c_{\kappa_{g,j}}^j|$, and $c=\left(\sum_{m=1, m \neq
i,j}^{r_g^*}|c_{\kappa_{g,j}}^m|^2\right)^{\frac{1}{2}}$, and the
variables $x,y\ge 0$ are given by $x=|\bar{c}_{\kappa_{g,i}}^j|$
and $y=\sum_{m=1, m \neq
i,j}^{r_g^*}|\bar{c}_{\kappa_{g,i}}^m|^2$, with a constraint
$x^2+y^2=1-|\bar{c}_{\kappa_{g,i}}^i|^2$.
 This convex optimization is solved by using the Karush-Kuhn-Tucker conditions \cite{Boyd&Vandenberghe:book}, and the solution is given by
 \begin{equation}
 x = b\sqrt{\frac{1-|c_{\kappa_{g,i}}^i |^2}{1-|c_{\kappa_{g,j}}^i|^2}} ~~\mbox{and}~~ y= c\sqrt{\frac{1-|c_{\kappa_{g,i}}^i |^2}{1-|c_{\kappa_{g,j}}^i|^2}}.
 \end{equation}
Substituting this $x,y$ into the RHS of \eqref{eq:cone_inner_upper}, we have
\begin{align}
|\tilde{\gbf}_{\kappa_{g,i}}^H \tilde{\gbf}_{\kappa_{g,j}}|
&\le a + bx +cy =a + (b^2+c^2)\sqrt{\frac{1-|c_{\kappa_{g,i}}^i |^2}{1-|c_{\kappa_{g,j}}^i|^2}}  \nonumber\\
&\stackrel{(a)}{=} |{c}_{\kappa_{g,i}}^i| \cdot |c_{\kappa_{g,j}}^i| + \sqrt{1-|{c}_{\kappa_{g,i}}^i|^2}\cdot \sqrt{1-|{c}_{\kappa_{g,j}}^i|^2}
\label{eq:cone_upper2}
\end{align}
where (a) follows from $b^2+c^2=1-|c_{\kappa_{g,j}}^i|^2$. Now, the RHS in \eqref{eq:cone_upper2} is expressed in terms of  $|{c}_{\kappa_{g,i}}^i|$ and $|c_{\kappa_{g,j}}^i|$ which comprised the constant term $a$ in the previous optimization of $a+bx+cy$. Here,
we have the following conditions for the terms in the RHS in \eqref{eq:cone_upper2}:
\begin{align}
|c_{\kappa_{g,i}}^i| &\ge \alpha  \label{eq:cone_inner_prod_fig_cond1}\\
\sqrt{1-|{c}_{\kappa_{g,i}}^i|^2} &\le  \sqrt{1-\alpha^2}\\
|c_{\kappa_{g,j}}^i| &= \sqrt{1-\sum_{m=1,m\ne i}^{r_g^*} |c_{\kappa_{g,j}}^m|^2} \le \sqrt{1-|c_{\kappa_{g,j}}^j|^2} \le \sqrt{1-\alpha^2}\\
\sqrt{1-|{c}_{\kappa_{g,j}}^i|^2} &\ge \alpha, \label{eq:cone_inner_prod_fig_cond4}
\end{align}
where \eqref{eq:cone_inner_prod_fig_cond1} is valid by the cone-containment condition.
\begin{figure}[t]
\begin{psfrags}
\small
                \psfrag{bac}[l]{$(|c_{\kappa_{g,i}}^i|,\sqrt{1-|{c}_{\kappa_{g,i}}^i|^2})$}
                \psfrag{a}[l]{$(|c_{\kappa_{g,j}}^i|,\sqrt{1-|{c}_{\kappa_{g,j}}^i|^2})$}
                \psfrag{1}[c]{$1$}
                \centerline{ \scalefig{0.4} \epsfbox{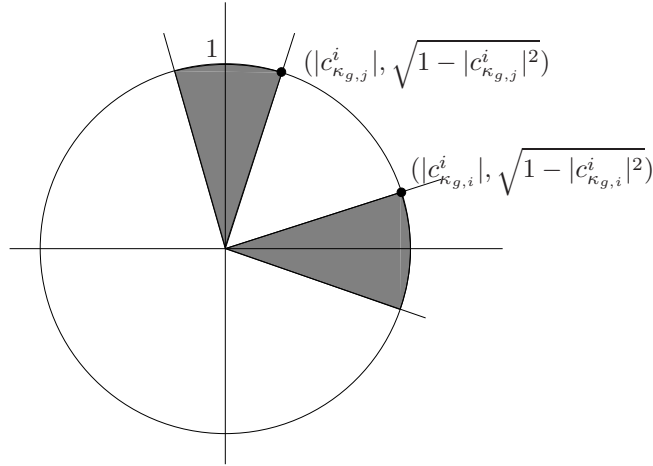} }
    \caption{Maximum inner product between two cones}
    \label{fig:inner_prod_cone_lemma}
\end{psfrags}
\end{figure}
The RHS in \eqref{eq:cone_upper2} is the inner product between two points $(|c_{\kappa_{g,i}}^i|,\sqrt{1-|{c}_{\kappa_{g,i}}^i|^2})$ and  $(|c_{\kappa_{g,j}}^i|,\sqrt{1-|{c}_{\kappa_{g,j}}^i|^2})$ with constraints \eqref{eq:cone_inner_prod_fig_cond1} to \eqref{eq:cone_inner_prod_fig_cond4}. The situation is depicted in Fig.  \ref{fig:inner_prod_cone_lemma}. The maximum inner product occurs between $(\alpha,\sqrt{1-\alpha^2})$ and $(\sqrt{1-\alpha^2},\alpha)$ and is given by $2\alpha\sqrt{1-\alpha^2}$. Therefore, we have
\begin{equation}
|\tilde{\gbf}_{\kappa_{g,i}}^H \tilde{\gbf}_{\kappa_{g,j}}|
\le 2 \alpha \sqrt{1-\alpha^2}~~ \text{for} ~~ i \neq j.
\end{equation}
Without the condition $\alpha \ge 1/\sqrt{2}$, the two shaded regions in Fig.  \ref{fig:inner_prod_cone_lemma} overlap, and we have a trivial upper bound of one.
\end{proof}

\section{Basic Extreme Value Theory}
\label{appen:extrem}
First,  we present two required theorems regarding the asymptotic
behavior of the maximum of $K$ i.i.d. random variables when
$K$ increases without bound.

\begin{theorem} \label{thm:elt}
(\cite{David:03book,Smirnov:49tr,Castillo:book}) Let
$Z_1,\cdots,Z_K$ be i.i.d. random variables with a common
cumulative density function (CDF) $F(\cdot)$. Suppose that there
exist  sequences $\{a_i>0\}_{i=1}^K$ and $\{b_i\}_{i=1}^K$ of
normalizing constants
 such that
\begin{equation}
\underset{K \to \infty}{\lim}F^K(a_Kz + b_K) = G(z),
\end{equation}
where $F^K(\cdot)$ is $F(\cdot)$ to the power of $K$.
Then, $G(z)$ must be one of the following three types of functions:
\begin{align}
(i)~~& G_1(z) = \left\{
\begin{array}{ll}
0,& z\le 0 \\
e^{-z^{-\alpha}},& z>0, ~\alpha>0
\end{array}
 \right. \\
(ii)~~ & G_2(z)=\left\{
\begin{array}{ll}
e^{-(-z)^\alpha},& z \le 0, ~\alpha>0 \\
1,& z>0
\end{array}
\right. \\
(iii) ~~& G_3(z) = e^{-e^{-z}}.
\end{align}
\end{theorem}

\vspace{1em}

\begin{theorem} \label{thm:smirnov}
 (\cite{Smirnov:49tr,Madda&Sadrabadi&Khandani:08IT})
 For
distribution function $F^K$ and $G_l(z)$, we have
\begin{equation}
\underset{K \to \infty}{\lim} F^K(a_Kz + b_K) = G_l(z)
\end{equation}
if and only if
\begin{equation}  \label{eq:theo_theo4}
\underset{K \to \infty}{\lim}K[1- F(a_Kz+b_K)] = -\log[G_l(z)],
\end{equation}
where $l \in \{ 1,2,3 \}$,  for two sequences $\{a_i>0\}_{i=1}^K$ and $\{b_i\}_{i=1}^K$.
\end{theorem}

\vspace{1em}

\begin{definition}
[Generalized chi-square distribution
\cite{Hammarwall&Bengtsson&Ottersten:08SP}]\label{def:gen_chi} If
$X_i \stackrel{i.i.d.}{\sim} \mathcal{CN}(0,1)$ for
$i=1,\cdots,L$, then the variable
$\chi_{\text{Gen}}^2(\lambda_1,\cdots,\lambda_L) :=
\sum\limits_{i=1}^L \lambda_i|X_i|^2$ with  $\lambda_1 > \lambda_2
> \cdots > \lambda_L >0$ is called a generalized chi-square random
variable  with order $L$ and parameters $\lambda_1,  \lambda_2,
 \cdots, \lambda_L$. Then,
$\chi_{\text{Gen}}^2(\lambda_1,\cdots,\lambda_L)$ has the pdf
\begin{equation}
f_{\chi^2_{\text{Gen}}}(z) = \sum\limits_{i=1}^L
\frac{e^{-z/\lambda_i}} {\lambda_i \prod_{j=1,j\neq i}^{L}
(1-\frac{\lambda_j}{\lambda_i})}, ~~\text{for}~~z\ge0.
\end{equation}
Furthermore, its CDF is given by
\begin{equation}
F_{\chi^2_{\text{Gen}}}(z) = \sum\limits_{i=1}^L
\frac{1-e^{-z/\lambda_i}}{\xi_i} ~~\text{and}~~
\sum\limits_{i=1}^L \frac{1}{\xi_i} = 1,
\end{equation}
where
\begin{equation} \label{eq:GenChicm}
\xi_i = \prod_{j=1, j \neq i}^L (1-\frac{\lambda_j}{\lambda_i}).
\end{equation}
\end{definition}

\vspace{1em}

Now, in a way similar to the technique used in \cite{Madda&Sadrabadi&Khandani:08IT}, we further generalize the generalized chi-square
distribution,  and  define a generalized CDF\footnote{In extreme value theory, typically the maximum of i.i.d. random variables is considered and thus, only the upper tail behavior of the CDF matters \cite{David:03book,Smirnov:49tr,Castillo:book}. } from the CDF of
$\chi_{\text{Gen}}^2(\lambda_1,\cdots,\lambda_L)$:
\begin{equation} \label{eq:gen_cdf}
F(z) = \left\{
\begin{array}{ll}
1 - \zeta\sum_{i=1}^L \frac{e^{-z/\lambda_i}}{\xi_i}, &
 z   \ge z_\tau,  \\
\tilde{F}(z), &  z < z_\tau,
\end{array}
\right.
\end{equation}
for  $0< \zeta < 1$ is a fixed constant, $z_\tau (< \infty)$ is a fixed finite threshold, and $\tilde{F}(z)$ is an arbitrary monotone-increasing continuous function satisfying $\tilde{F}(0)=0$ and $\tilde{F}(z_\tau)=  1 - \zeta\sum_{i=1}^L \frac{e^{-z_\tau/\lambda_i}}{\xi_i}$.  Then, this is a valid CDF over $z \ge 0$ since $F(0)=0$, $F(\infty)=1$, and $F(z)$ is continuous and monotone increasing.
Based on the
two theorems in the above, we derive the following lemma regarding
the newly defined CDF in  \eqref{eq:gen_cdf}, necessary for proof
of (\ref{eq:multiuserdiversitygain}, \ref{eq:multiDivugi}).

\vspace{1em}

\begin{lemma} \label{lem:ext}
Let $Z_1,Z_2,\cdots,Z_K$ be $K$ i.i.d. random variables with the CDF in \eqref{eq:gen_cdf} with $\lambda_1 > \lambda_2 > \cdots > \lambda_L$.
 Then, the limiting
behavior of $F^K$ belongs to  type $(iii)$ in  Theorem
\ref{thm:elt} with normalizing sequences
\begin{equation}  \label{eq:append_extrem_akbk}
a_K = \lambda_1, ~~ b_K = \lambda_1 (\log K +
\log(\zeta/\xi_1)),
\end{equation}
and therefore, we have
\begin{equation}
\text{Pr}\{Z_{\max} > \lambda_1 \log K - \lambda_1 \log\log K +
\lambda_1\log (\zeta/\xi_1)\} \ge 1- O\left(\frac{1}{K}\right),
\end{equation}
where $Z_{\max}$ denotes the maximum of $\{Z_i\}_{i=1}^K$.
\end{lemma}

\vspace{0.5em}

\begin{proof}  Compute $K[1-F(a_Kz+b_K)]$ with $a_K$ and $b_K$ in \eqref{eq:append_extrem_akbk} when $a_K z+b_K \ge z_\tau$:
\begin{align}
\underset{K \to \infty}{\lim} K[1-F(a_Kz+b_K)]
&= \underset{K \to \infty}{\lim} K \zeta
\left[\sum\limits_{i=1}^L\frac{e^{-(a_Kz+b_K)/\lambda_i}}{\xi_i} \right] \\
&= \underset{K \to \infty}{\lim} K\zeta
\left[\sum\limits_{i=1}^L\frac{e^{-[z+\log K +
\log(\zeta/\xi_1)]\frac{\lambda_1}{\lambda_i}}}{\xi_i} \right] \\
&=\underset{K \to \infty}{\lim} K  \label{eq:lem2_pr}
\left[\frac{e^{-z}}{K} +
\zeta\sum\limits_{i=2}^L\frac{[e^{-z}(\xi_1/\zeta)]^{\frac{\lambda_1}{\lambda_i}}}{\xi_i K^{\frac{\lambda_1}{\lambda_i}}} \right] \\
&\overset{(a)}{=} e^{-z} = - \log [G_3(z)],
\end{align}
where $(a)$ follows from the fact that
$\frac{\lambda_1}{\lambda_i}>1$ for $i=2,\cdots,L$,  and
 the second term in the
RHS of \eqref{eq:lem2_pr}  vanishes as $K \to \infty$. By
Theorems  \ref{thm:elt} and \ref{thm:smirnov}, the limiting behavior of $F^K$ belongs to  type $(iii)$ in  Theorem \ref{thm:elt}  with the normalizing sequences $a_K$ and $b_K$ in \eqref{eq:append_extrem_akbk}, when  $a_Kz + b_K \ge z_\tau$ for sufficiently large $K$. Hence, we have
\begin{align}
\underset{K \to \infty}{\lim} F^K(\lambda_1 z + \lambda_1 \log K +
\lambda_1 \log (\zeta/\xi_1)) = e^{-e^{-z}},
\end{align}
when  $a_Kz + b_K \ge z_\tau$ for sufficiently large $K$. This implies
\begin{equation} \label{eq:append_extreme_event}
\underset{K \to \infty}{\lim} \text{Pr}\{Z_{\max} > \lambda_1 z +
\lambda_1 \log K + \lambda_1 \log (\zeta/\xi_1)\} =1-
e^{-e^{-z}},
\end{equation}
because $F^K$ is the CDF of $Z_{\max}=\max\{Z_1,\cdots,Z_K\}$.
By substituting $z=-\log\log K$ and removing the limit operator, we get
\begin{equation}
\text{Pr}\{Z_{\max} > \lambda_1 \log K - \lambda_1 \log\log K +
\lambda_1\log (\zeta/\xi_1)\} \ge 1- O\left(\frac{1}{K}\right)
\end{equation}
since $a_Kz + b_K \ge z_\tau$ for sufficiently large $K$ with
$z=-\log\log K$ and $(a_K, b_K)$ in \eqref{eq:append_extrem_akbk}
due to the term ``$\log K$'' in $b_K$.
\end{proof}

\section{Proof of (\ref{eq:multiuserdiversitygain}, \ref{eq:multiDivugi})}
\label{sec:Appen-proof}

Now we prove (\ref{eq:multiuserdiversitygain},
\ref{eq:multiDivugi}) in the proof of Theorem
\ref{theo:OptimalityProof} under the conditions of Theorem
\ref{theo:OptimalityProof}. The impact of no quasi-SINR feedback
by the users whose channel vectors are not contained in the
user-selection cones, is incorporated by defining $\phi_{g_k}^i$
in \eqref{eq:phigkdef}. To handle the inter-group interference, we
here define new random variables. For each $i \in
\{1,\cdots,r_g^*\}$, we define random variables
$\bar{\phi}_{g_k}^i$,  $k=1,\cdots,K'$, as
\begin{align}
\bar{\phi}_{g_k}^i
= \left\{
\begin{array}{ll}
\phi_{g_k}^i, & k \in \Vc_g(\epsilon) \\
0, & \text{otherwise}
\end{array}
\right.
\end{align}
where
\begin{equation}
\Vc_g(\epsilon) := \left\{k: \sum_{g' \neq g} \|\hbf_{g_k}^H\Vbf_{g'}\|^2 \le \frac{\epsilon}{r_g^*} \right\}
\end{equation}
for some constant $\epsilon >0$.
Let us define the following sets:
\begin{align}
\overline{\Wc}_{g,i}(\alpha) &:= \{\hbf_{g_k}:
k \in \Wc_{g,i}(\alpha) \}, ~~~i \in \{1,\cdots,r_g^\star\}\\
\overline{\Vc}_g(\epsilon) &:= \{\hbf_{g_k}:
k \in \Vc_g(\epsilon) \},  \label{eq:overlineVggGamma}
\end{align}
where $\Wc_{g,i}(\alpha)$ is defined in Algorithm \ref{algo:ReDOS}. (The dependence of $\Wc_{g,i}$ on $\alpha$ is explicitly shown here. $\overline{\Wc}_{g,i}(\alpha)$ and  $\overline{\Vc}_g(\epsilon)$ are simply denoted by $\overline{\Wc}_{g,i}$ and $\overline{\Vc}_g$, respectively, in case of no confusion.)
Note that the fixed and chosen $\bar{\alpha}$ satisfies
\begin{equation}
\bar{\alpha} > \sqrt{\frac{1+\sqrt{\frac{r_g^*-2}{r_g^*-1}}}{2}},
\end{equation}
and this implies $\alpha > \frac{1}{\sqrt{2}}$ for any $r_g^* \ge
2$. Then, the user-selection cones are disjoint (see Remark
\ref{remark:full_disjoint_cones}) and hence,  we can rewrite
$\overline{\Wc}_{g,i}$ as
\begin{align}
\overline{\Wc}_{g,i}
&= \left\{\hbf_{g_k} : \frac{|(\hbf_{g_k}^H \Vbf_g)\ebf_{i}^{(g)}|^2}{
\|\hbf_{g_k}^H \Vbf_g\|^2} \ge \alpha^2 \right\}, ~~\because \gbf_{g_k}^H = \hbf_{g_k}^H \Vbf_g = \hbf_{g_k}^H \Ubf_g^* \\
&= \left\{\hbf_{g_k} : \frac{\lambda_{g,i}|\eta_{g_k,i}|^2}{
\sum_{m=1}^{r_g^\star}\lambda_{g,m} |\eta_{g_k,m}|^2 } \ge \alpha^2 \right\}, ~~~i=1,\cdots,r_g^*,
\label{eq:W_gi}
\end{align}
where $\lambda_{g,i}$ is the $i$-th largest eigenvalue of
$\Rbf_g$, and $\eta_{g_k,m}$ is the $m$-th element of
$\etabf_{g_k}$ given in the channel model
(\ref{eq:channel_linear_comb}$\sim$\ref{eq:model_channel_component}). This is because from (\ref{eq:channel_linear_comb}$\sim$\ref{eq:model_channel_component})
\begin{align}
\hbf_{g_k} &= \Ubf_{g} \Lambdabf_{g}^{1/2} \etabf_{g_k}=\sum_{i=1}^{r_g} \eta_{g_k,i} \sqrt{\lambda_{g,i}} \ubf_{g,i},\nonumber\\
\gbf_{g_k}^H &= \hbf_{g_k}^H \Vbf_g = \hbf_{g_k}^H \Ubf_g^*,  \nonumber\\
&= [\eta_{g_k,1}^*\sqrt{\lambda_{g,1}},
\eta_{g_k,2}^*\sqrt{\lambda_{g,2}}, \cdots,
\eta_{g_k,r_g^*}^*\sqrt{\lambda_{g,r_g^*}}].
\label{eq:append_ggbf_explicit}
\end{align}
Now consider
$\overline{\Vc}_g$ in \eqref{eq:overlineVggGamma}. This set can be rewritten as
\begin{align}
\overline{\Vc}_g &= \left\{
\hbf_{g_k} : \sum_{g' \neq g} \|\hbf_{g_k}^H \Vbf_{g'}\|^2 \le \frac{\epsilon}{r_g^*}
\right\} \\
&=\left\{
\hbf_{g_k} : \sum_{g' \neq g} \|\etabf_{g_k}^H\Lambdabf^{1/2}\underbrace{\Ubf_g^H \Vbf_{g'}}_{\mbox{\small see \eqref{eq:appendapproxBDUV}}}\|^2 \le \frac{\epsilon}{r_g^*}
\right\}, ~~~~(\Vbf_{g'}=\Ubf_{g'}^*) \\
&\stackrel{(a)}{=}\left\{
\hbf_{g_k} : \sum_{g' \neq g} \| \sum_{m=r_g^\star + 1}^{r_g} \eta_{g_k,m} \sqrt{\lambda_{g,m}}\xbf_{g,g'}^{(m)}  \|^2 \le \frac{\epsilon}{r_g^*}
\right\}.  \label{eq:append_overVgg_last}
\end{align}
Step (a) is by the approximate BD condition in Condition \ref{cond:approxBD} assumed for Theorem \ref{theo:OptimalityProof}, i.e., \cite{Adhikary&Caire:13arXiv}
\begin{equation}  \label{eq:appendapproxBDUV}
\Ubf_g^H \Vbf_{g'}=\Ubf_g^H \Ubf_{g'}^*
= \left[
\begin{array}{c}
{\bf 0}_{r_g^\star \times r_{g'}^\star} \\
\Xbf_{g,g'}
\end{array}
\right],
\end{equation}
where $\Xbf_{g,g'}$ is some matrix of size $(r_g - r_g^\star)
\times r_{g'}^\star$ which can be a non-zero matrix, and
$\xbf_{g,g',m}$  in \eqref{eq:append_overVgg_last} is the $m$-th
row vector of $\Ubf_g^H\Vbf_{g'}$. One key observation regarding
$\overline{\Wc}_{g,i}$ and $\overline{\Vc}_g$ is that  the event
of $\hbf_{g_k} \in \overline{\Wc}_{g,i}$ and the event of
$\hbf_{g_k} \in \overline{\Vc}_g$ are independent under the
approximate BD condition, because the former event depends only on
$\{\eta_{g_k,1},\cdots,\eta_{g_k,r_g^*}\}$,  the latter event
depends only on $\{\eta_{g_k,r_g^*+1},\cdots,\eta_{g_k,r_g}\}$,
and the random variables
$\eta_{g_k,1},\eta_{g_k,2},\cdots,\eta_{g_k,r_g}$ are i.i.d.
(Please see \eqref{eq:model_channel_component}.)

Now, we obtain a lower bound on  the complementary CDF (CCDF) of $\phi_{g_k}^i$ of user $g_k$:
\begin{align}
\text{Pr}\{\phi_{g_k}^i \ge z \} &\ge \text{Pr}\{\bar{\phi}_{g_k}^i \ge z\} \nonumber\\
&\stackrel{(a)}{=} \text{Pr}\{\bar{\phi}_{g_k}^i \ge z, \hbf_{g_k}
\in \overline{\Wc}_{g,i}, \hbf_{g_k} \in \overline{\Vc}_g\} \nonumber \\
&= \text{Pr}\{\hbf_{g_k}
\in \overline{\Wc}_{g,i}, \hbf_{g_k} \in \overline{\Vc}_g\}
 \cdot \text{Pr}\left\{\bar{\phi}_{g_k}^i \ge z | \hbf_{g_k}
\in \overline{\Wc}_{g,i}, \hbf_{g_k} \in \overline{\Vc}_g\right\} \nonumber \\
&\stackrel{(b)}{=}\text{Pr}\{\hbf_{g_k} \in \overline{\Wc}_{g,i},
\hbf_{g_k} \in \overline{\Vc}_g\}
 \text{Pr}\left\{\frac{\|\gbf_{g_k}\|^2}{
\frac{1}{\rho} + r_g^*\sum_{g' \neq g}
\|\hbf_{g_k}^H\Vbf_{g'}\|^2} \ge z \bigg| \hbf_{g_k}
\in \overline{\Wc}_{g,i}, \hbf_{g_k} \in \overline{\Vc}_g
\right\} \nonumber\\
&\stackrel{(c)}{\ge}\text{Pr}\{\hbf_{g_k} \in
\overline{\Wc}_{g,i}, \hbf_{g_k} \in \overline{\Vc}_g\}
 \text{Pr}\left\{\frac{\|\gbf_{g_k}\|^2}{
\frac{1}{\rho} + \epsilon} \ge z \bigg| \hbf_{g_k}
\in \overline{\Wc}_{g,i}, \hbf_{g_k} \in \overline{\Vc}_g
\right\} \nonumber\\
&\stackrel{(d)}{=}\text{Pr}\{\hbf_{g_k} \in \overline{\Wc}_{g,i},
\hbf_{g_k} \in \overline{\Vc}_g\}
 \text{Pr}\left\{\frac{\|\gbf_{g_k}\|^2}{
\frac{1}{\rho} + \epsilon} \ge z \bigg| \hbf_{g_k}
\in \overline{\Wc}_{g,i} \right\} \nonumber\\
&\stackrel{(e)}{=}\text{Pr}\{\hbf_{g_k} \in
\overline{\Wc}_{g,i}\}\text{Pr}\{\hbf_{g_k} \in \overline{\Vc}_g\}
 \text{Pr}\left\{\frac{\|\gbf_{g_k}\|^2}{
\frac{1}{\rho} + \epsilon} \ge z \bigg| \hbf_{g_k}
\in \overline{\Wc}_{g,i} \right\} \nonumber\\
&\overset{(f)}{\ge}\text{Pr}\{\hbf_{g_k} \in
\overline{\Wc}_{g,i}\}\text{Pr}\{\hbf_{g_k} \in \overline{\Vc}_g\}
\text{Pr}\left\{\|\gbf_{g_k}\|^2
 \ge z'\right\}, ~~z' = z(1/\rho + \epsilon). \label{eq:gen}
\end{align}
Here, (a) is because the events $\{\bar{\phi}_{g_k}^i \ge z\}$ and
$\{\bar{\phi}_{g_k}^i \ge z, \hbf_{g_k} \in \overline{\Wc}_{g,i},
\hbf_{g_k} \in \overline{\Vc}_g\}$ are the same for $z>0$ due to
the definition of $\bar{\phi}_{g_k}^i$; (b) follows because
conditioned on $\{\hbf_{g_k} \in \overline{\Wc}_{g,i}, \hbf_{g_k}
\in \overline{\Vc}_g\}$, $\bar{\phi}_{g_k}^i = \Rc(g_k)$; (c) is
valid because conditioned on $\{\hbf_{g_k} \in
\overline{\Vc}_g\}$, $\Rc(g_k) \ge \|\gbf_{g_k}\|^2/(1/\rho +
\epsilon)$; (d) is valid because the events $\{\hbf_{g_k} \in
\overline{\Wc}_{g,i}\}$ and $\{ \hbf_{g_k} \in \overline{\Vc}_g\}$
are independent, and the event $\frac{\|\gbf_{g_k}\|^2}{
\frac{1}{\rho} + \epsilon} \ge z$ is independent of $\{ \hbf_{g_k}
\in \overline{\Vc}_g\}$; (e) is valid because the events
$\{\hbf_{g_k} \in \overline{\Wc}_{g,i}\}$ and $\{ \hbf_{g_k} \in
\overline{\Vc}_g\}$ are independent; and finally (f) follows from
Lemma \ref{lem:multi_proof}.

For given $\alpha < 1$ and $\epsilon>0$,  define
$\zeta_{g,i}(\alpha,\epsilon)$  as
\begin{equation}
\zeta_{g,i}(\alpha,\epsilon):=\text{Pr}\{\hbf_{g_k} \in
\overline{\Wc}_{g,i}(\alpha)\}\text{Pr}\{\hbf_{g_k} \in
\overline{\Vc}_g(\epsilon)\} > 0.
\end{equation}
Note that $\zeta_{g,i}(\alpha,\epsilon) \in (0,1)$ is a positive
constant, when $\alpha < 1$ and $\epsilon>0$ are given, since we
have a strictly positive probability for the event $\{\hbf_{g_k}
\in \overline{\Wc}_{g,i}(\alpha)\}$ and a strictly positive
probability for $\{\hbf_{g_k} \in \overline{\Vc}_g(\epsilon)\}$.
Now, we define new i.i.d. random variables $\Psi_{g_k}$ for
$k=1,\cdots,K'$ that have the common complementary CDF (CCDF)
constructed as
\begin{equation}  \label{eq:append_mud_new_pdf}
\text{Pr}\{\Psi_{g_k} \ge z\} = \left\{
\begin{array}{ll}
\zeta_{g,i}(\alpha,\epsilon) \cdot \text{Pr}\{\|\gbf_{g_k}\|^2 \ge
z\}, &
 z   \ge z_\tau,  \\
\tilde{F}_C(z), &  z < z_\tau,
\end{array}
\right.
\end{equation}
where $\tilde{F}_C(z)$ is constructed arbitrarily such that
\eqref{eq:append_mud_new_pdf} is a CCDF. Then, the corresponding
CDF of \eqref{eq:append_mud_new_pdf} is given by
\begin{equation}  \label{eq:append_constructed_Fz}
F(z) = \left\{
\begin{array}{ll}
1 - \zeta_{g,i}(\alpha,\epsilon) \sum_{j=1}^{r_g^*} \frac{e^{-z/\lambda_{g,j}}}{\xi_{g,j}}, & z \ge z_\tau \\
1-\tilde{F}_C(z), &  z < z_\tau,
\end{array}
\right.
\end{equation}
since $\|\gbf_{g_k}\|^2$ is
$\chi_{\text{Gen}}^2(\lambda_{g,1},\cdots,\lambda_{g,r_g^*})$
defined in Definition \ref{def:gen_chi} (see
\eqref{eq:append_ggbf_explicit}), where the parameters
$\lambda_{g,1},\cdots, \lambda_{g,r_g^*}$ are the eigenvalues of
the channel covariance matrix $\Rbf_g$ in the channel model
(\ref{eq:channel_linear_comb} $\sim$
\ref{eq:model_channel_component}). The CDF
\eqref{eq:append_constructed_Fz} falls into the CDF class of
\eqref{eq:gen_cdf} and hence, we can apply Lemma \ref{lem:ext}.
Applying
 Lemma \ref{lem:ext}, we have
\begin{equation}
\text{Pr}\{\Psi_{\tilde{\kappa}_{g,i}} > u' \} \ge 1 -
O\left(\frac{1}{K'} \right)
\end{equation}
where $\Psi_{\tilde{\kappa}_{g,i}} :=
\max\{\Psi_{g_1},\cdots,\Psi_{g_{K'}}\}$ and $u' =
\lambda_{g,1}\log K' - \lambda_{g,1} \log \log K' +
\lambda_{g,1}\log \frac{\zeta_{g,i}(\alpha,\epsilon)}{\xi_1}$.
Therefore,  we obtain
\begin{align}
1 - O\left(\frac{1}{K'} \right) ~\le~
\text{Pr}\{\Psi_{\tilde{\kappa}_{g,i}} > u'\}
~&\overset{(a)}{\le}~
\text{Pr}\left\{\phi_{\tilde{\kappa}_{g,i}}^i > \frac{u'}{1/\rho +
\epsilon }
\right\}  \\
&\overset{(b)}{\le} \text{Pr}\{\phi_{\kappa_{g,i}}^i > u_g^i\},
\end{align}
where $u_g^i=\frac{u'}{1/\rho + \epsilon }$ (see
\eqref{eq:multiDivugi}). Here,  (a) follows from  the definition
of $\Psi_{g_k}$ and the  inequality \eqref{eq:gen}, and (b)
follows from the fact that $\kappa_{g,i} = \arg\max \phi_{g_k}^i$.
This concludes the proof. \hfill{$\blacksquare$}

\vspace{1em}

\begin{lemma} \label{lem:multi_proof}
$\text{Pr}\left\{\|\gbf_{g_k}\|^2 \ge z
\big| \hbf_{g_k} \in \overline{\Wc}_{g,i}\right\}
> \text{Pr}\{\|\gbf_{g_k}\|^2 \ge z\}$.
\end{lemma}

\begin{proof}
Let $X:=\lambda_{g,i}|\eta_{g_k,i}|^2$ and $Y:=\sum_{m=1,m\neq
i}^{r_g^\star}
 \lambda_{g,m}|\eta_{g_k,m}|^2$. First, we represent the two
 events $\{\hbf_{g_k} \in \overline{\Wc}_{g,i}\}$ and $\{\|\gbf_{g_k}\|^2 \ge
 z\}$ in terms of $X$ and $Y$.
From \eqref{eq:W_gi}, we have
\begin{align}
\{\hbf_{g_k} \in \overline{\Wc}_{g,i}\}
&= \left\{\frac{X}{X+Y} \ge \alpha^2\right\} \\
&= \left\{X \ge \frac{\alpha^2}{1-\alpha^2}Y\right\} \\
&= \left\{X+Y \ge \frac{1}{1-\alpha^2}Y\right\},
\end{align}
and  $\{\|\gbf_{g_k}\|^2 \ge z\} = \{X+Y \ge z\}$. Thus, we have
{\small
\begin{align*}
&\text{Pr}\left\{\|\gbf_{g_k}\|^2 \ge z
\big| \hbf_{g_k} \in \overline{\Wc}_{g,i}\right\} \\
&= \text{Pr}\left\{X+Y \ge z \Big| X+Y \ge \frac{1}{1-\alpha^2}Y \right\} \\
&\overset{(a)}{=} \text{Pr}\left\{X+Y \ge z, Y \ge z(1-\alpha^2) \Big| X+Y
> \frac{1}{1-\alpha^2}Y\right\} + \text{Pr}\left\{X+Y \ge z, Y < z(1-\alpha^2) \Big| X+Y
> \frac{1}{1-\alpha^2}Y\right\}  \\
&\overset{(b)}{=} \text{Pr}\left\{Y \ge z(1-\alpha^2) \Big| X+Y > \frac{1}{1-\alpha^2}Y \right\} \cdot \text{Pr}\left\{X+Y \ge z \Big| X+Y > \frac{1}{1-\alpha^2}Y,
Y \ge z(1-\alpha^2)\right\}  \\
&~~~+ \text{Pr}\left\{Y < z(1-\alpha^2) \Big| X+Y > \frac{1}{1-\alpha^2}Y \right\}\cdot
\text{Pr}\left\{X+Y \ge z \Big| X+Y > \frac{1}{1-\alpha^2}Y,Y<z(1-\alpha^2) \right\} \\
&\overset{(c)}{\ge} \text{Pr}\left\{Y \ge z(1-\alpha^2) \Big| X+Y > \frac{1}{1-\alpha^2}Y \right\} + \text{Pr}\left\{Y < z(1-\alpha^2) \Big| X+Y > \frac{1}{1-\alpha^2}Y \right\}\text{Pr}\left\{X+Y \ge z\right\}  \\
&\overset{(d)}{\ge}  \left[\text{Pr}\left\{Y \ge z(1-\alpha^2) \Big| X+Y > \frac{1}{1-\alpha^2}Y \right\} + \text{Pr}\left\{Y < z(1-\alpha^2) \Big| X+Y > \frac{1}{1-\alpha^2}Y \right\}\right]\text{Pr}\left\{X+Y \ge z\right\}  \\
&= \text{Pr}\{X+Y \ge z\} = \text{Pr}\{\|\gbf_{g_k}\|^2 \ge z\}.
\end{align*}
} Here, (a) follows from the law of total probability:
\begin{equation}
\text{Pr}\{A|C\} = \text{Pr}\{A, B|C\} + \text{Pr}\{A, B^c|C\};
\end{equation}
(b) holds by Bayes' rule; (c) follows from the fact that {\small
\begin{equation}
\text{Pr}\left\{X+Y \ge z \Big| X+Y > \frac{1}{1-\alpha^2}Y,Y\ge z(1-\alpha^2)\right\}
 =\text{Pr}\left\{X+Y \ge z |X+Y >z \right\}=1
\end{equation}
}
and
{\small
\begin{align}
\text{Pr}\left\{X+Y \ge z \Big| X+Y > \frac{1}{1-\alpha^2}Y,Y<z(1-\alpha^2) \right\}
&= \text{Pr}\left\{X+Y \ge z \Big| X+Y > z-\delta \right\} \\
&= \frac{\text{Pr}\{X+Y \ge z, X+Y > z - \delta\}}{\text{Pr}\{X+Y > z - \delta\}} \\
&= \frac{\text{Pr}\{X+Y \ge z\}}{\text{Pr}\{X+Y > z - \delta\}} \\
&\ge \text{Pr}\{X+Y \ge z\}
\end{align}
} for some $\delta>0$; and (d) is valid because the first term in
the RHS is multiplied by $\text{Pr}\{X+Y \ge z\} \le 1$ from the
previous step.
\end{proof}

\end{appendices}




\begin{thebibliography}{}



\bibitem{Lee&Sung:14SPAWC}   G. Lee and Y. Sung, ``Asymptotically optimal simple user scheduling for massive MIMO downlink with two-stage beamforming,'' submitted to 2014 {\it SPAWC}, Feb., 2014
        
        
\bibitem{Caire&Shamai:03IT} G. Caire and S. Shamai, ``On the achievable throughput of a multi-antenna Gaussian broadcast channel,'' {\it IEEE Trans. Inf. Theory},   vol. 49, no. 7, pp. 1691 - 1706, Jul. 2003
        
\bibitem{Weingarten&Steinberg&Shamai:04ISIT} H. Weingarten, Y. Steinberg and S. Shamai,
        ``The capacity region of the Gaussian MIMO broadcast channel,''
        {\it Proc. of ISIT}, Chicago, IL, 2004
        
\bibitem{Sharif&Hassibi:05IT}   M. Sharif and B. Hassibi, ``On the capacity of MIMO broadcast channels with partial side information,'' {\it IEEE Trans. Inf. Theory},  vol. 51, no. 2, pp. 506 -522, Feb. 2005

\bibitem{Yoo&Goldsmith:06JSAC} T. Yoo and  A. Goldsmith, ``On the optimality of multiantenna broadcast scheduling  using zero-forcing beamforming,'' {\it IEEE J.  Sel. Areas Commun.},  vol. 24,  no. 3,
    pp. 528 - 541, Mar. 2006

\bibitem{Costa:83IT} M. Costa, ``Writing on dirty paper,'' {\it IEEE Trans. Inf. Theory}, vol. 29, no. 3, pp. 439 - 441, May 1983

\bibitem{Liuetal:12COMMAG} L. Liu, R. Chen, S. Geirhofer, K. Sayana, Z. Shi, and Y. Zhou,
        ``Downlink MIMO in LTE-Advanced: SU-MIMO vs. MU-MIMO,'' {\it IEEE Commun. Mag.},
        vol. 50, no. 2, pp. 140 -147, Feb. 2009

\bibitem{Knopp&Humblet:95ICC} R. Knopp and P.A. Humblet, ``Information capacity and power control in single cell multi-user communications,''
        {\it Proc. Intl Conf. Comm.},    pp. 331-335, Seattle, WA, Jun. 1995

\bibitem{Viswanath&Tse&Laroia:02IT} P. Viswanath, D. N. C. Tse, and R. Laroia,
        ``Opportunistic beamforming using dumb antennas,''  {\it IEEE Trans. Inf. Theory},
        vol. 48, no. 6,  pp. 1277 - 1294, Jun. 2002


\bibitem{Adhikary&Nam&Ahn&Caire:13IT}
        A. Adhikary, J. Nam, J. Ahn and G. Caire,
        ``Joint spatial division and multiplexing: The large-scale array regime,''
        {\it IEEE Trans. Inf. Theory}, vol. 59, no. 10, pp. 6441 - 6463,  Oct. 2013
        
        
\bibitem{Alnaffouri&Sharif&Hassibi:09COM} T. Al-Naffouri, M. Sharif and B. Hassibi,
        ``How much does transmit correlation affect the  sum-rate scaling of MIMO Gaussian broadcast channels?,''
        {\it IEEE Trans. Commun.}, vol. 57, no. 2,  pp. 562 -572, Feb. 2009

        
\bibitem{Adhikary&Caire:13arXiv}  A. Adhikary and G. Caire,
        ``Joint spatial division and multiplexing: Opportunistic
        beamforming and user grouping,''
        {\it arXiv preprint arXiv:1305.7252}, 2013

\bibitem{Herdin&Bonek:04ISTMWC}  M. Herdin and E. Bonek,
        ``A MIMO correlation matrix based metric for characterizing non-stationarity,''
        {\it Proc. the IST Mobile and Wireless Communications Summit},
        Lyon, France, Jun. 2004
        

\bibitem{Hoydis&Hoek&Wild&Brink:12ISWCS} J. Hoydis,  C. Hoek,  T. Wild, and S. ten Brink,
        ``Channel measurements for large antenna arrays,''
        {\it Proc. IEEE ISWCS},
        Paris, France, Aug. 2012

\bibitem{Ispas&Dorpinghaus&Ascheid&Zemem:13SP}
        A. Ispas,   M. D$\ddot{\mbox{o}}$rpinghaus,  G. Ascheid, and T. Zemem,
        ``Characterization of non-stationary channels using mismatched Winer filtering,''
        {\it IEEE Trans. Signal Process.}, vol. 64, no. 2,  pp.  274 - 288,  Jan. 2013
        

\bibitem{Ispas&Schneider&Ascheid&Thoma:10VTC}
        A. Ispas,  C. Schneider, G. Ascheid, and R. Thom$\ddot{\mbox{a}}$,
        ``Analysis of local quasi-stationarity regions in an urban macrocell scenario,''
        {\it Proc. IEEE VTC}, Taipei, Taiwan, May 2010
        
        
\bibitem{Noh&Zoltowski&Sung&Love:14JSTSP}      S. Noh,  M. D. Zoltowski,  Y. Sung,  and D. J. Love,
        ``Pilot beam pattern design for channel estimation
        in massive MIMO systems,''
        accepted to {\it IEEE J. Sel. Topics  Signal
        Process.}, available at http:$\slash \slash$arxiv.org$\slash$abs$\slash$1309.7430,
        Dec., 2013

\bibitem{Jakes:book} W. Jakes, {\it Microwave Mobile Communications},
    Wiley,  New York, 1974


\bibitem{Shiu&Foschini&Gans&Kahn:00COM}
        D. Shiu,  G. J. Foschini,  M. J. Gans,  and J. M. Kahn,
        ``Fading correlation and its effect on the capacity of multi element antenna systems,''
        {\it IEEE Trans. Commun.}, vol. 48, no. 3, pp. 502 - 513,  Mar.  2000
        
        
\bibitem{Cover&Thomas:book} T. Cover and J. Thomas,
      {\it Elements of Information Theory}, John Wiley \& Sons, Inc., 1991
      

\bibitem{Dimic&Sidiropoulos:05SP} G. Dimic and N. D. Sidiropoulos,
        ``On downlink  beamforming with greedy user selection: Performance analysis and a simple new algorithm,''
        {\it IEEE Trans. Signal Process.}, vol. 53, no. 10, pp. 3857 - 3868, Oct. 2005
                

\bibitem{Peel&Hochwald&Swindlehurst:05COM} C. B. Peel,  B. M. Hochwald,  and A. L. Swindlehurst,
        ``A vector-perturbation technique for near-capacity multiantenna
        multiuser communication-part I: Channel inversion and regularization,''
        {\it IEEE Trans. Commun.},
        vol. 53, no. 1, pp. 195 - 202, Jan. 2005


\bibitem{Horn&Johnson:book} R. A. Horn and C. R. Johnson,
  {\it Matrix Analysis}, Cambridge University Press, Cambridge, UK, 1985 
  
 
 
\bibitem{Huang&Rao:13WC} Y. Huang and B. Rao,
        ``Random beamforming with heterogeneous users and selective feedback: Individual sum rate and individual scaling laws,''
        {\it IEEE Trans. Wireless Commun.}, vol. 12, no. 5, pp.  2080 - 2090,
        May 2013


\bibitem{Boyd&Vandenberghe:book}
  S. Boyd and L. Vandenberghe,
  {\it Convex Optimization}, Cambridge University Press, New York, NY, 2004
  

\bibitem{David:03book} H. A. David and H. N. Nagaraja, {\it Order Statistics},
   John Wiley \& Sons Inc.,  New York, 2003
   


\bibitem{Smirnov:49tr} N. V. Smirnov, ``Limit distributions for the terms of a variational
  series,'' {\it Trudy Mat. Inst.}, vol. 25, 1949
  
  
\bibitem{Castillo:book}  E. Castillo, {\it Extreme Value Theory in Engineering},
  Academic Press, Inc., San Diego, CA, 1988 
  
  
\bibitem{Madda&Sadrabadi&Khandani:08IT}
        M. A. Maddah-Ali,  M. A. Sadrabadi, and A. K. Khandani,
        ``Broadcast in MIMO systems based on a generalized
        QR decomposition: Signaling and performance analysis,''
        {\it IEEE Trans. Inf. Theory},
        vol. 54, no. 3, pp. 1124 - 1138, Mar. 2008

\bibitem{Hammarwall&Bengtsson&Ottersten:08SP}
        D. Hammarwall,  M. Bengtsson,  and B. Ottersten,
        ``Acquiring partial CSI for spatially selective transmission
        by instantaneous channel norm feedback,''
        {\it IEEE Trans. Signal Process.},
        vol. 56, no. 3, pp. 1188 - 1204, Mar. 2008 
        
   


\end{thebibliography}
 \end{document}